\documentclass[10pt,oneside]{article}

\usepackage{lmodern}
\usepackage{times}

\usepackage{amssymb,amsmath,amsthm}
\usepackage{bbold}
\usepackage[short]{optidef}

\usepackage{framed,multirow,multicol}

\usepackage{xcolor,xspace}
\usepackage[normalem]{ulem}
\usepackage{enumerate}
\usepackage{verbatim}
\usepackage{makeidx,latexsym}

\usepackage[round]{natbib}

\usepackage[english]{babel}
\usepackage{bbm}
\usepackage{pgfplots}

\usepackage{parskip}
\usepackage[letterpaper,margin=1.1in]{geometry}

\usepackage[utf8]{inputenc} %
\usepackage[T1]{fontenc}    %
\usepackage[colorlinks=True,citecolor=blue]{hyperref}       %
\usepackage{url}            %
\usepackage{booktabs}       %
\usepackage{amsfonts}       %
\usepackage{nicefrac}       %
\usepackage{microtype}      %
\usepackage{xcolor}     %
\usepackage{float}
\usepackage{graphicx}
\usepackage{subcaption}

\usepackage{algorithm}
\usepackage{algorithmic}
\usepackage{amsmath} 
\usepackage{amsthm}
\usepackage[capitalize,noabbrev]{cleveref}
\usepackage{bm}

\usepackage{amsmath}

\DeclareMathOperator*{\reward}{Reward}
\DeclareMathOperator*{\resource}{Resource}
\DeclareMathOperator*{\regret}{Regret}

\DeclareMathOperator*{\opt}{OPT}

\DeclareMathOperator*{\var}{Var}
\DeclareMathOperator*{\argmax}{argmax}

\newcommand{\cS}{\mathcal{S}}
\newcommand{\cA}{\mathcal{A}}

\newcommand{\cP}{\mathcal{P}}
\newcommand{\cF}{\mathcal{F}}
\newcommand{\cG}{\mathcal{G}}
\newcommand{\cH}{\mathcal{H}}
\newcommand{\tsh}{(s,h)}
\newcommand{\tshn}{(s',h+1)}

\newtheorem{theorem}{Theorem}
\newtheorem{lemma}{Lemma}[section]

\usepackage{bigfoot}

\DeclareNewFootnote{AAffil}[arabic]
\DeclareNewFootnote{ANote}[fnsymbol]

\usepackage{etoolbox}
\makeatletter
\patchcmd\maketitle{\def\@makefnmark{\rlap{\@textsuperscript{\normalfont\@thefnmark}}}}{}{}{}
\makeatother

\makeatletter
\def\thanksAAffil#1{%
  \footnotemarkAAffil\protected@xdef\@thanks{\@thanks%
        \protect\footnotetextAAffil[\the \c@footnoteAAffil]{#1}}%
}
\def\thanksANote#1{%
  \footnotemarkANote%
  \protected@xdef\@thanks{\@thanks%
        \protect\footnotetextANote[\the \c@footnoteANote]{#1}}%
}
\makeatother

\title{Online Resource Allocation in Episodic Markov Decision Processes}

\author{
	Duksang Lee
	\thanksAAffil{Department of Industrial and Systems Engineering, KAIST, Daejeon 34141, Republic of Korea}
	\and
 William Overman
	\thanksAAffil{Graduate School of Business, Stanford University, Stanford, CA 94305, United States}
 \and
	Dabeen Lee
	\FootnotemarkAAffil{1}$~$
    \thanksANote{Correspondence to <\url{dabeenl@kaist.ac.kr}>}
}	

\date{\today}

\begin{document}

\maketitle

\begin{abstract}
This paper studies a long-term resource allocation problem over multiple periods where each period requires a multi-stage decision-making process. We formulate the problem as an online allocation problem in an episodic finite-horizon constrained Markov decision process with an unknown non-stationary transition function and stochastic non-stationary reward and resource consumption functions. We propose the observe-then-decide regime and improve the existing decide-then-observe regime, while the two settings differ in how the observations and feedback about the reward and resource consumption functions are given to the decision-maker. We develop an online dual mirror descent algorithm that achieves near-optimal regret bounds for both settings. For the observe-then-decide regime, we prove that the expected regret against the dynamic clairvoyant optimal policy is bounded by $\tilde O(\rho^{-1}{H^{3/2}}S\sqrt{AT})$ where $\rho\in(0,1)$ is the budget parameter, $H$ is the length of the horizon, $S$ and $A$ are the numbers of states and actions, and $T$ is the number of episodes. For the decide-then-observe regime, we show that the regret against the static optimal policy that has access to the mean reward and mean resource consumption functions is bounded by $\tilde O(\rho^{-1}{H^{3/2}}S\sqrt{AT})$ with high probability. We test the numerical efficiency of our method for a variant of the resource-constrained inventory management problem.
\end{abstract}

\newpage
\tableofcontents
\newpage

\section{INTRODUCTION}

We consider a long-term online resource allocation problem where requests for service arrive sequentially over episodes and the decision-maker chooses an action that generates a reward and consumes a certain amount of resources for each request. Such resource allocation problems arise in revenue management~\citep{chen2021primaldual,cyclic-demand} and online advertising~\citep{balseiro2022}. Hotels and airlines receive requests for a room or a flight, and they decide how to process them in real-time based on their availability of remaining rooms and flight seats~\citep{Talluri04}. For search engines, when a user arrives with a keyword, they collect bids from relevant advertisers and decide which ad to show to the user~\citep{adwords-mehta}. The decision-maker is informed of or receives stochastic feedback about the reward and resource consumption functions of arriving requests, while the decision-maker makes actions in an online fashion with no knowledge of future requests. 

The online resource allocation problem has been studied in the context of or under the name of the AdWords problem~\citep{devanur-hayes,buchbinder2007}, bandits with knapsacks~\citep{bwk1,bwk2,bwk3},  repeated auctions with budgets~\citep{balseiro-gur}, online stochastic matching~\citep{karp-bipartite,feldman-matching}, online linear programming~\citep{Glover-lp,agrawal2014,gupta2016}, assortment optimization with limited inventories~\citep{Golrezaei2014}, online convex programming~\citep{agrawal-devanur2014}, online binary programming~\citep{Li2020}, online stochastic optimization~\citep{jiang2022}, online resource allocation with concave rewards~\citep{balseiro2020} and nonlinear rewards~\citep{balseiro2022}. %

Although the above literature assumes that the decision-maker makes a \emph{single} action for a request, many of the modern service systems allow \emph{multi-stage} decision-making processes and interactions with customers based on user feedback. For example, medical processes~\citep{medical-mdp} involve sequential decision-making while the prices and costs of medical operations are often predetermined. Multi-stage second-price repeated auctions~\citep{gummadi-auction} consider a bidder who is willing to participate in auctions multiple times until winning an item. Modern recommender systems feature continued interactions with customers~\citep{ad-interact,recommender}. For these applications, actions taken in multiple stages are not necessarily independent, and therefore, it is natural to group a multi-stage decision-making process as an episode. Then this brings about online resource allocation problems where an episode itself involves a sequential decision-making process. Therefore, to model these scenarios, we need a framework to capture multi-stage actions for requests and interactions between service systems and customers.

Motivated by this, we extend the existing framework to consider online resource allocation over an episodic Markov decision process (MDP), generalizing a single action for an episode to multi-stage actions. More generally, we formulate the problem as an episodic finite-horizon constrained MDP (CMDP)~\citep{Altman,efroni2020} whose goal is to maximize the cumulative reward by allocating a long-term resource budget over episodes. Basically, the decision-maker prepares a policy for an episode, after running which the decision-maker observes the cumulative reward and resource consumption over the episode. There is a long-term budget for the total resource consumption for all episodes, so the decision-maker can keep track of the remaining budget but cannot observe the reward and resource consumption functions of future episodes. Therefore, the problem is to prepare policies based on past observations and feedback about the reward and resource consumption functions, the remaining resource budget, and the estimation of the unknown transition kernel. The main challenge here is to deal with uncertainties in not only the reward and resource consumption functions of future episodes but also the unknown transition function of the underlying MDP.

We consider two settings, the \emph{observe-then-decide} regime and the \emph{decide-then-observe} regime, which differ in how the observations and feedback about the random reward and resource consumption functions are given to the decision-maker. The first setting is related to the contextual MDP literature \citep{hallak2015contextual,modi-contextual,contextual-mdp-offline-regression,contextual-mdp-online-regression} and assumes that the sampled reward and resource consumption functions of each episode are revealed at the beginning of the episode. Here, we define and study the regret against the \emph{dynamic clairvoyant optimal policy}. The second setting considers essentially the stochastic finite-horizon episodic constrained MDP~\citep{efroni2020,NEURIPS2020_ae95296e,liu2021learning,bura2022dope,pmlr-v162-chen22i, pmlr-v151-wei22a,pmlr-v206-wei23b} and assumes that the random reward and resource consumption functions for an episode can be observed only after the episode.

\paragraph{Our Contributions}

This paper develops and studies a formulation for the long-term sequential allocation problem for multi-step decision processes by providing an integrated view of the online resource allocation problem and the constrained Markov decision process. %

We propose what we call the observe-then-decide regime for finite-horizon episodic CMDPs, and we define and study the regret against the dynamic optimal policy that has access to the reward and resource consumption functions of all episodes and the transition kernel. We prove if the reward and resource consumption functions are i.i.d. over episodes, then our online dual mirror descent algorithm guarantees that the expected dynamic regret is bounded above by $O(\rho^{-1}{H^{3/2}}S\sqrt{AT}\left(\ln HSAT\right)^2)$ where $\rho\in(0,1)$ is the budget parameter, $H$ is the length of the horizon for each episode, $S$ is the number of states, $A$ is the number of actions, and $T$ is the number of episodes. Our algorithm is designed to stop before violating the long-term resource budget constraint.

For the decide-then-observe regime under both the full information and the bandit feedback settings, we show that the online dual mirror descent algorithm achieves $O(\rho^{-1}{H^{3/2}}S\sqrt{AT}\left(\ln (HSAT/\delta)\right)^2)$ regret with probability at least $1-\delta$ against the static optimal policy that has access to the mean reward and mean resource consumption functions and the transition kernel. This result does not require Slater's condition or the existence of a strictly safe policy in contrast to and improves upon~\cite{efroni2020,liu2021learning,pmlr-v151-wei22a} whose regret bounds have a suboptimal dependence on $H$ for the case of long-term budget constraints.

Our regret bounds for the observe-then-decide and decide-then-observe regimes are nearly tight as there is a lower bound of $\Omega(H^{3/2}\sqrt{SAT})$ due to~\cite{pmlr-v132-domingues21a} based on an instance with a deterministic reward function and no resource budget constraint.

Lastly, we test the numerical performance of our online dual mirror descent method for a variant of the online resource-constrained inventory management problem.

\section{PROBLEM SETTING}\label{sec:setting}

\paragraph{Finite-horizon Episodic CMDP}

We model the online resource allocation problem with a \emph{finite-horizon episodic MDP}. A finite-horizon  MDP is defined by a tuple $(\cS,\cA, H, \left\{P_h\right\}_{h=1}^{H-1},p)$ where $\cS$ is the finite state space with $|\cS|=S$, $\cA$ is the finite action space with $|\cA|=A$, $H$ is the finite horizon, $P_h:\cS\times \cA\times \cS\to [0,1]$ is the transition kernel at step $h\in[H]$, and $p$ is the initial distribution of the states. Here, $P_h(s'\mid s,a)$ is the probability of transitioning to state $s'$ from state $s$ when the chosen action is $a$ at step $h\in[H-1]$. Equivalently, we may define a single \emph{non-stationary} transition kernel $P:\cS\times \cA\times\cS\times [H]\to[0,1]$ with $P(s'\mid s,a,h)=P_h(s'\mid s,a)$  and $P(s'\mid s,a,H)=p(s')$ for $(s,a,s',h)\in\cS\times \cA\times \cS\times[H-1]$. We assume that $\{P_h\}_{h=1}^{H-1}$ and thus $P$ are \emph{unknown}. 

Before an episode begins, the decision-maker prepares a \emph{stochastic policy} $\pi:\cS\times [H]\times \cA\to[0,1]$ where $\pi(a\mid s,h)$ is the probability of selecting action $a\in \cA$ in state $s\in \cS$ at step $h$. Here, $\pi$ can be viewed as a \emph{non-stationary policy} as it may change over the horizon, and this is due to the non-stationarity of the transition kernel $P$ over steps $h\in[H]$. Given a policy $\pi_t$ for episode $t\in[T]$, the MDP proceeds with trajectory $\{s_h^{P,\pi_t}, a_h^{P,\pi_t}\}_{h\in[H]}$ generated by $P$.

The reward and resource consumption functions of episode $t\in[T]$ is given by $f_t,g_t:\cS\times \cA\times [H]\to[0,1]$, i.e., choosing action $a\in \cA$ at state $s\in \cS$ at step $h$ generates a reward $f_t(s,a,h)$ and consumes resources of amount $g_t(s,a,h)$. Here, functions $f_t$ and $g_t$ are non-stationary over $h\in[H]$. Throughout the paper, we assume that $(f_t,g_t)$ for each episode $t\in[T]$ is an i.i.d. sample from an \emph{unknown} distribution $\mathcal{D}$ with mean $(f,g)$ where $f,g:\cS\times \cA\times [H]\to[0,1]$.

The budget for the total resource consumption over $T$ episodes is given by $TH\rho$ for some $\rho\in(0,1)$. Hence, the goal is to find policies $\pi_1,\ldots, \pi_T$ for the $T$ episodes to maximize the total cumulative reward, given by
$$\reward\left(\vec\gamma,\vec\pi\right):=\sum_{t=1}^{T}\sum_{h=1}^{H}f_{t}\left(s_{h}^{P,\pi_t},a_{h}^{P,\pi_t},h\right)$$
while satisfying the resource budget constraint
$$
\resource\left(\vec\gamma,\vec\pi\right):=\sum_{t=1}^{T}\sum_{h=1}^{H}g_{t}\left(s_{h}^{P,\pi_t},a_{h}^{P,\pi_t},h\right)\leq TH\rho.$$
Here, we use short-hand notations $\vec{\gamma}=(\gamma_1,\ldots, \gamma_T)$ where $\gamma_t=(f_t,g_t)$ and $\vec{\pi}=(\pi_1,\ldots,\pi_T)$. The decision-maker selects policies $\pi_1,\ldots, \pi_T$ in an \emph{online} fashion because the decision-maker is oblivious to the reward and resource consumption functions of \emph{future} episodes as well as the true transition kernel $P$.

\paragraph{Observe-then-decide Regime} The first setting we consider assumes that the decision-maker may observe $\gamma_t=(f_t,g_t)$ before each episode $t$ begins and may adapt to $\gamma_t$ as well as the history $(\gamma_1,\ldots,\gamma_{t-1})$. The performance of the decision-maker can be compared to the best possible performance achievable when $\vec\gamma$ and $P$ are all available in advance, which is given by
\begin{align*}
\begin{aligned}
\vec \pi^* \quad\in \quad  \argmax_{\vec\pi=(\pi_{1},\ldots,\pi_{T})} \quad \mathbb{E}\left[\reward\left(\vec\gamma,\vec\pi\right)\mid \vec\gamma,\vec\pi, P\right] \quad
\text{s.t.}  \quad  \mathbb{E}\left[\resource\left(\vec\gamma,\vec\pi\right)\mid \vec\gamma,\vec\pi, P\right]\leq TH\rho.
\end{aligned}
\end{align*}
This setting is related to the contextual MDP framework~\citep{hallak2015contextual,modi-contextual,contextual-mdp-offline-regression,contextual-mdp-online-regression} in that $\gamma_t$ itself serves as the context for episode $t$. The distinction from these works is that the distribution $\mathcal{D}$ of reward and resource consumption functions may include infinitely many functions with arbitrary structures. 

To measure the performance of a learning algorithm that produces policies $\pi_1,\ldots, \pi_T$, we consider 
\begin{align*}
\regret\left(\vec\gamma,\vec\pi\right)&=\opt(\vec\gamma)-\reward\left(\vec\gamma,\vec\pi\right)
 \end{align*}
 where $\opt(\vec \gamma) = \reward(\vec\gamma, \vec\pi^*)$ is the cumulative reward under the optimal policies $\vec\pi^*=(\pi_1^*,\ldots, \pi_T^*)$.
 The observe-then-decide regime covers applications in online advertising~\citep{devanur-hayes,buchbinder2007,ad-interact} and recommender systems with interactions~\citep{recommender} and extends the online resource allocation framework without Markovian transitions~\citep{balseiro2022}.

\paragraph{Decide-then-observe Regime} The second setting is that the decision-maker prepares a policy $\pi_t$ based on the history $(\gamma_1,\ldots,\gamma_{t-1})$ and then observes the associated reward and resource consumption functions. For the full-information setting, we observe $f_t(s,a,h)$ and $g_t(s,a,h)$ for every $(s,a,h)\in \cS\times \cA\times[H]$. For the bandit feedback setting, we observe $f_t(s_h,a_h,h)$ and $g_t(s_h,a_h,h)$ for only the chosen action $a_h$ at state $s_h$ at each step $h\in[H]$. Here, the performance is evaluated against a single optimal policy applied to all episodes. Let $\pi^*$ be a policy defined as
\begin{equation*}
\begin{aligned}
\pi^* \quad \in \quad\argmax_{\pi}\quad  \mathbb{E}\left[\sum_{h=1}^{H}f\left(s_{h}^{P,\pi},a_{h}^{P,\pi},h\right)\mid f,\pi, P\right] \quad
\text{s.t.}  \quad \mathbb{E}\left[\sum_{h=1}^{H}g\left(s_{h}^{P,\pi},a_{h}^{P,\pi},h\right)\mid g,\pi, P\right]\leq H\rho.
\end{aligned}
\end{equation*}
Then we define the regret as
\begin{align*}
\regret\left(\vec\gamma,\vec\pi\right)=%
\opt(\vec\gamma)-\reward\left(\vec\gamma,\vec \pi\right).
 \end{align*}
 where $\opt(\vec\gamma)=\sum_{t=1}^{T}\sum_{h=1}^{H}f_t(s_{h}^{P,\pi^*},a_{h}^{P,\pi^*},h)$ is the cumulative reward under $\pi^*$ applied to all episodes.
 The full information setting can model inventory management with observable demands~\citep{chen2021primaldual,cyclic-demand}, and the bandit case is what is standard in the stochastic episodic finite-horizon CMDP literature~\citep{efroni2020,liu2021learning,bura2022dope,Bai_Bedi_Agarwal_Koppel_Aggarwal_2022,Bai_Singh_Bedi_Aggarwal_2023,9867805,pmlr-v151-wei22a,pmlr-v206-wei23b}.

\paragraph{Zero Long-term Constraint Violation}

 We do not allow violating the resource consumption constraint. Hence, an algorithm needs to stop if the remaining budget is not enough to continue the process. To better present our analysis, we assume that even after stopping the process, we take action $a_{\text{stop}}$ with $f_t(s,a_{\text{stop}},h)=g_t(s,a_{\text{stop}},h)=0$ for $(s,h,t)\in\cS\times [H]\times[T]$ as if we are still running our algorithm.

\section{FORMULATION}\label{sec:reformulation}

\subsection{Occupancy Measures}\label{sec:occupancy}

Our framework for the long-term online resource allocation problem over an episodic finite-horizon CMDP is based on reformulations via \emph{occupancy measures}~\citep{Altman,Zimin2013,rosenberg2019,cohen2021}. Given a policy $\pi$ and  a transition kernel $P$, let $\bar q^{P, \pi}:\cS\times\cA\times\cS\times [H]\to[0,1]$ be defined as
\begin{equation}\label{occupancy}
\bar q^{P,\pi}(s,a,s',h)=\mathbb{P}\left[s^{P,\pi}_{h}=s,\ a^{P,\pi}_{h}=a,\ s^{P, \pi}_{h+1}=s'\mid \pi,P\right]
\end{equation}
for $(s,a,s',h)\in \cS\times \cA\times \cS\times [H]$. Note that any $\bar q$ defined as in~\eqref{occupancy} has the following properties. 
\begin{align*}
\sum_{(s,a,s')\in \cS\times \cA\times \cS}\bar q(s,a,s',h)&=1,\quad h\in [H]\tag{C1}\label{item:occu1}  \\
\sum_{(s',a)\in\cS\times \cA}\bar q(s,a,s',h)&=\sum_{(s',a)\in\cS\times\cA}\bar q(s',a,s,h-1),\quad s\in\cS,\ h=2,\ldots, H.\tag{C2}\label{item:occu2}  
\end{align*}
The \emph{occupancy measure} $q^{P,\pi}:\cS\times\cA\times[H]\to[0,1]$ associated with policy $\pi$ and transition kernel $P$ is defined as
\begin{equation}\label{occupancy'}
 q^{P,\pi}(s,a,h)=\sum_{s'\in\cS} \bar q^{P, \pi}(s,a,s',h).\tag{C3}
\end{equation}
Then it follows that
$$q^{P,\pi}(s,a,h)=\mathbb{P}\left[s^{P,\pi}_{h}=s,\ a^{P,\pi}_{h}=a\mid\pi,P\right].$$
Hence, if a policy $\pi$ is chosen, then the occupancy measure for a loop-free MDP with transition kernel $P$ is determined. Conversely, any $q\in\cS\times\cA\times [H]\to[0,1]$ with $\bar q:\cS\times\cA\times\cS\times [H]\to[0,1]$ satisfying~\eqref{item:occu1},~\eqref{item:occu2},~\eqref{occupancy'} induces a transition kernel $P^{q}$ and a policy $\pi^{q}$ given as follows:
\begin{align}\label{induced}
\begin{aligned}
P^{q}(s'\mid s,a,h)=\frac{\bar q(s,a,s',h)}{\sum_{s''\in \cS}\bar q(s,a,s'',h)},\quad
\pi^{q}(a\mid s,h)=\frac{q(s,a,h)}{\sum_{b\in \cA} q(s,b,h)}.
\end{aligned}
\end{align}
\begin{lemma}\label{lemma:valid-occupancy}
Let $q:\cS\times \cA\times[H]\to[0,1]$. Then $q$ is a valid occupancy measure that induces transition kernel $P$ if and only if there exists $\bar q:\cS\times \cA\times\cS\times [H]\to[0,1]$ that satisfies~\eqref{item:occu1},~\eqref{item:occu2},~\eqref{occupancy'}, and $P^{q}=P$.
\end{lemma}

Therefore, there is a one-to-one correspondence between the set of policies and the set of occupancy measures that give rise to transition kernel $P$. Moreover, the cumulative reward for episode $t$ under reward function $f_t$, policy $\pi_t$, and transition kernel $P$ can be written in terms of occupancy measure $q^{P,\pi_t}$ associated with $\pi_t$ and $P$.
\begin{align*}%
\begin{aligned}
&\mathbb{E}\left[\sum_{h=1}^{H}f_{t}\left(s_{h}^{\pi_t}(t),a_{h}^{\pi_t}(t),h\right)\mid \vec\gamma,\vec\pi, P \right]=\sum_{(s,a,h)\in \cS\times\cA\times[H]}q^{P,\pi_t}\left(s,a,h\right)f_{t}\left(s,a,h\right).
\end{aligned}
\end{align*}
We may express occupancy measure $q^{P,\pi_t}$ as an $(S\times A\times H)$-dimensional vector $\bm{q^{P,\pi_t}}$ whose entries are given by $q^{P,\pi_t}\left(s,a,h\right)$ for $(s,a,h)\in\cS\times \cA\times [H]$. Similarly, we define $\bm{\bar q^{P,\pi_t}}$ as the vector whose entries are $\bar q^{P,\pi_t}(s,a,s',h)$ for $(s,a,s',h)\in\cS\times \cA\times\cS\times [H]$. Moreover, we define vector $\bm{f_t}$ whose entry corresponding to $(s,a,h)\in \cS\times\cA\times [H]$ is given by $f_t\left(s,a,h\right)$. Then the right-hand side of the above equation is equal to $\langle \bm{f_t},\bm{q^{P,\pi_t}}\rangle$, the inner product of  $\bm{f_t}$ and $\bm{q^{P,\pi_t}}$. Likewise, we define vector $\bm{g_t}$ to represent the resource consumption function $g_t$. Consequently,
$$\mathbb{E}\left[\reward(\vec\gamma, \vec \pi)\mid \vec\gamma,\vec\pi,P \right]=\sum_{t=1}^T\langle \bm{f_t},\bm{q^{P,\pi_t}}\rangle,\quad\mathbb{E}\left[\resource(\vec\gamma, \vec \pi)\mid \vec\gamma,\vec\pi,P \right]=\sum_{t=1}^T\langle\bm{g_t}, \bm{q^{P,\pi_t}}\rangle.$$
Then the policy optimization problem can be reformulated as
\begin{align}\label{reformulation}
\begin{aligned}
\opt(\vec{\gamma})= \max_{\bm{q_1},\ldots,\bm{q_T}\in \Delta(P)}\quad \sum_{t=1}^T\langle \bm{f_t},\bm{q_t}\rangle\quad \text{s.t.}\quad \sum_{t=1}^T\langle \bm{g_t},\bm{q_t}\rangle\leq TH\rho
\end{aligned}
\end{align}
where $\Delta(P)$ is the set of all valid occupancy measures inducing transition kernel $P$. More precisely, $\Delta(P)$ is defined as
$$\Delta(P)=\left\{\bm{q}\in [0,1]^{S\times A\times H}:\ \exists \bm{\bar q}\in [0,1]^{S\times A\times S\times H}\text{ satisfying}~\eqref{item:occu1},~\eqref{item:occu2},~\eqref{occupancy'},~P^{q}=P\right\}$$
where $P^{\bar q}$ is defined as in~\eqref{induced}. Hence, $\Delta(P)$ is a polytope, and therefore, \eqref{reformulation} corresponds to an online linear programming instance. 

For the decide-then-observe regime, the optimal policy $\pi^*$ satisfies
\begin{equation*}
\pi^*  \in \argmax_{\pi}\quad  \langle \bm{f},\bm{q}\rangle \quad
\text{s.t.} \quad \langle \bm{g},\bm{q}\rangle\leq H\rho.
\end{equation*}

\subsection{Confidence Sets}\label{sec:confidence_sets}

In contrast to the standard online optimization problems, the feasible set $\Delta(P)$ is unknown as we do not have access to the true transition kernel $P$. To remedy this issue, we obtain empirical transition kernels $\bar P_1,\ldots, \bar P_T$ to estimate the true transition kernel $P$, based on which we construct \emph{relaxations} $\Delta(P,1),\ldots, \Delta(P,T)$ of the feasible set $\Delta(P)$ by building \emph{confidence sets} for $P$. %
Another issue with applying the existing online dual algorithms, e.g., the dual mirror descent by~\cite{balseiro2022}, is that the relaxations $\Delta(P,1),\ldots, \Delta(P,T)$ are not i.i.d., but instead we will show that they contain $\Delta(P)$ with high probability, which turns out to be sufficient for our analysis.

To construct confidence sets for estimating the non-stationary transition kernel $P$, we extend the framework of~\cite{Jin2020} developed for the loop-free setting and that of~\cite{ssp-adversarial-unknown} for the stochastic shortest path setting to our finite-horizon non-stationary setting. As~\cite{ssp-adversarial-unknown}, we update the confidence set for each episode $t\in[T]$, in contrast to~\cite{Jin2020} where the confidence set is updated over \emph{epochs} and an epoch may consist of multiple episodes. The distinction is that the transition function is estimated for each  distinct step $h\in[H]$, and this leads to the $O(H^{3/2})$ dependence on $H$ in our regret upper bounds, instead of $O(H)$.

We maintain counters to keep track of the number of visits to each tuple $(s,a,h)$ and tuple $(s,a,s',h)$. For each $t\in[T]$, we define $N_t(s,a,h)$ and $M_t(s,a,s',h)$ as the number of visits to tuple $(s,a,h)$ and the number of visits to tuple $(s,a,s',h)$ up to the first $t-1$ episodes, respectively, for $(s,a,s',h)\in\cS\times\cA\times\cS\times[H]$. Given $N_t(s,a,h)$ and $M_t(s,a,s',h)$, we define the empirical transition kernel $\bar P_t$ for episode $t$ as
\begin{equation}\label{empirical}
\bar{P}_{t}(s'\mid s,a,h)=\frac{M_{t}(s,a,s',h)}{\max\{1,N_{t}(s,a,h)\}}.
\end{equation}
Next, for some confidence parameter $\delta\in(0,1)$, we define the confidence radius $\epsilon_t(s'\mid s,a,h)$  for $(s,a,s',h)\in \cS\times \cA\times\cS\times [H]$ and $t\in[T]$ as 
\begin{align}\label{confidence-radius}
\begin{aligned}
\epsilon_t(s'\mid s,a,h)=2\sqrt{\frac{\bar{P}_{t}(s'\mid s,a,h)\ln\left({{HSAT}/{\delta}}\right)}{\max\{1,N_{t}(s,a,h)-1\}}}+\frac{14\ln\left({{HSAT}/{\delta}}\right)}{3\max\{1,N_{t}(s,a,h)-1\}}.
\end{aligned}
\end{align}
Based on the empirical transition kernel and the radius, we define the confidence set $\cP_t$ for episode $t$ as 
\begin{equation}\label{confidence-set}
\cP_t= \left\{\widehat P:
     \left|\widehat P(s'\mid s,a,h)-\bar P_t(s'\mid s,a,h)\right| 
      \leq \epsilon_{t}(s'\mid s,a,h)\ \ \forall (s,a,s',h)
\right\}.
\end{equation}
Then, by the empirical Bernstein inequality due to~\cite{Maurer-bernstein}, we show the following. 
\begin{lemma}\label{lemma:confidence}
With probability at least $1-4\delta$, the true transition kernel $P$ is contained in the confidence set $\cP_t$ for every episode $t\in[T]$.
\end{lemma}
For episode $t\in [T]$, we define $\Delta(P,t)$ as 
\begin{equation}\label{feasible-set}
\Delta(P,t)=\left\{\bm{q}\in [0,1]^{S\times A\times H}:\ \exists \bm{\bar q}\in [0,1]^{S\times A\times S\times H}\text{ satisfying}~\eqref{item:occu1},~\eqref{item:occu2},~\eqref{occupancy'},~P^{q}\in\cP_t\right\}.
\end{equation}
The following is a direct consequence of \Cref{lemma:confidence}.
\begin{lemma}\label{lemma:relaxation}
With probability at least $1-4\delta$, $\Delta(P)\subseteq \Delta(P,t)$ for every episode $t\in[T]$.
\end{lemma}

\subsection{Optimistic Estimators for Reward and Resource Consumption Functions}\label{sec:estimator}

Our framework applies the \emph{optimism in the face of uncertainty} principle to estimate the unknown mean reward and resource consumption functions $f$ and $g$. Throughout the paper, we denote by $\widehat f_t$ and $\widehat g_t$ the estimators of $f$ and $g$, respectively. 

For the observe-then-decide setting, we set
$$\widehat f_t = f_t\quad\text{and}\quad \widehat g_t = g_t.$$
For the decide-then-observe regime, let $\mathbbm{1}_t(s,a,h)$ be the indicator variable for the event that $f_t(s,a,h)$ and $g_t(s,a,h)$ are observed for state-action pair $(s,a)\in\cS\times\cA$ at step $h\in[H]$ of episode $t\in[T]$. Then $\mathbbm{1}_t(s,a,h)=1$ for any $(s,a,h)$ and $t$ under the full-information setting while $\mathbbm{1}_t(s,a,h)=1$ only if $(s,a)$ is visited at step $h$ of episode $t$. With this, we define the empirical estimates $\bar f_t$ and $\bar g_t$ of $f$ and $g$ as follows.
\begin{align*}
\bar f_t(s,a,h) = \frac{\sum_{k=1}^{t-1}f_k(s,a,h)\mathbbm{1}_k(s,a,h)}{C_t(s,a,h)},\quad
\bar g_t(s,a,h) = \frac{\sum_{k=1}^{t-1}g_k(s,a,h)\mathbbm{1}_k(s,a,h)}{C_t(s,a,h)}.
\end{align*}
where $C_t(s,a,h) =\max\{1,\sum_{k=1}^{t-1}\mathbbm{1}_k(s,a,h)\}$. Then
\begin{align*}
\widehat f_t(s,a,h) &= \min\left\{1,\bar f_t(s,a,h) +7R_t(s,a,h)
 + 2\sqrt{R_t(s,a,h) \bar f_t(s,a,h)}\right\},\\
\widehat g_t(s,a,h) &= \max\left\{0,\bar g_t(s,a,h) -7R_t(s,a,h)
 - 2\sqrt{R_t(s,a,h) \bar g_t(s,a,h)}\right\}
\end{align*}
where $$R_t(s,a,h)=\frac{\ln(2HSAT/\delta)}{C_t(s,a,h)}.$$ In contrast to~\cite{ssp-adversarial-unknown}, we estimate the reward and resource consumption functions for each distinct step $h\in[H]$. 
\begin{lemma}\label{lemma:estimator}
	With probability at least $1-2\delta$, for any $(s,a,h)\in\cS\times\cA\times[H]$ and $t\in[T]$, 
	\begin{align*}
		0\leq &\widehat f_t(s,a,h) - f(s,a,h)\leq 
	8\sqrt{R_t(s,a,h) f(s,a,h)} +34R_t(s,a,h) 
		,\\
		0\leq &g(s,a,h)-\widehat g_t(s,a,h) \leq 
		8\sqrt{R_t(s,a,h) g(s,a,h)}+34R_t(s,a,h).
	\end{align*}
\end{lemma}

\section{ONLINE DUAL METHOD}\label{sec:dualmethod}

We present our online dual method (\Cref{alg:online-alloc-mdp-unknown}) for online resource allocation in episodic finite-horizon MDPs.
\begin{algorithm}[h!]
\caption{Online Dual Method for Online Resource Allocation in Episodic Finite-horizon MDPs}
\label{alg:online-alloc-mdp-unknown}
\begin{algorithmic}
\STATE {\bfseries Initialize:} dual variable $\lambda_1$, budget $B=TH\rho$, episode counter $t=1$, counters 
$N(s,a,h)=0$ and $M(s,a,s',h)=0$
for $(s,a,s',h)\in \cS\times\cA\times \cS\times [H]$, and step size $\eta>0$.
\FOR{$t=1,\ldots, T$}
\STATE {\bfseries (1. Confidence set construction)} 
\STATE Set counters $N_t\leftarrow N$ and $M_t\leftarrow M$.
\STATE Compute $\bar P_t$, $\epsilon_t$, and $\cP_t$ as in~\eqref{empirical},~\eqref{confidence-radius}, and~\eqref{confidence-set}.
\STATE {\bfseries (2. Policy update)} 
\STATE Compute estimators $\widehat f_t$ and $\widehat g_t$  as in \Cref{sec:estimator}.
\STATE Deduce policy  $\pi_t = \pi^{\widehat q_t}$ as in~\eqref{induced} where $\bm{\widehat q_t}\in \argmax_{\bm{q}\in \Delta(P,t)}\left\{\langle \bm{\widehat f_t},\bm{q}\rangle - \lambda_t \langle\bm{\widehat g_t},\bm{q}\rangle\right\}$ and $\Delta(P,t)$ is defined in~\eqref{feasible-set}.
\STATE {\bfseries (3. Policy execution)} 
\STATE Sample state $s_1$ from distribution $p(\cdot)$.
\FOR{$h=1,\ldots, H$}
\STATE Sample  $a_h$ from  $\pi_t(\cdot \mid s_h, h)$, accrue $f_t(s_h,a_h,h)$, and update $B\leftarrow B - g_t(s_h,a_h,h)$.
\IF{$B<1$}
\STATE {\bfseries Return}
\ENDIF
\STATE Observe  $s_{h+1}$ determined by $P(\cdot\mid s_h,a_h,h)$.
\STATE Update counters $N(s,a,h)\leftarrow N(s,a,h)+1$ and $M(s,a,s',h)\leftarrow M(s,a,s',h)+1$.
\ENDFOR
\STATE {\bfseries (4. Dual update)} 
\STATE %
Update $\lambda_{t+1}=\max\left\{0,\lambda_t - \eta\left(H\rho - \langle \bm{\widehat g_t}, \bm{\widehat q_t}\rangle\right)\right\}$.
\ENDFOR
\end{algorithmic}
\end{algorithm}
As explained in \Cref{sec:occupancy}, the online resource allocation problem can be reformulated as an online linear optimization problem where each decision is encoded by an occupancy measure that corresponds to a policy for an episode. Then we adapt the \emph{online dual mirror descent} algorithm by~\cite{balseiro2022} originally developed for nonlinear reward and resource consumption functions under the observe-then-decide regime.  %

\Cref{alg:online-alloc-mdp-unknown} proceeds with four parts in each episode. At the beginning of each episode $t\in[T]$, it first obtains $\Delta(P,t)$, which is a relaxation of the feasible set $\Delta(P)$ with high probability by \Cref{lemma:relaxation}, by constructing the confidence set $\cP_t$. Second, the algorithm prepares a policy $\pi_t$ based on the current dual solution $\lambda_t$, optimistic reward function estimator $\widehat f_t$, optimistic resource consumption function estimator $\widehat g_t$, and the set $\Delta(P,t)$. Third, the algorithm runs the episode with policy $\pi_t$. Lastly, the algorithm prepares dual solution $\lambda_{t+1}$ for the next episode based on the outcomes of episode $t$.

To be more specific, the policy update part works as follows. Given the dual solution $\lambda_t$ prepared before episode $t$ starts, we take
$$\bm{\widehat q_t}\in \argmax_{\bm{q}\in \Delta(P,t)}\left\{\langle \bm{\widehat f_t},\bm{q}\rangle - \lambda_t \langle\bm{\widehat g_t},\bm{q}\rangle\right\}.$$
Note that $f_t - \lambda_t g_t$ is the reward function $f_t$ penalized by the resource consumption function $g_t$, and $\widehat f_t - \lambda_t \widehat g_t$ is an optimistic estimator of $f_t - \lambda_t g_t$. Then based on~\eqref{induced}, we deduce policy $\pi^{\widehat q_t}$ associated with the occupancy measure $\widehat q_t$ whose vector representation is $\bm{\widehat q_t}$ as in~\eqref{induced}. For ease of notation, we denote $\pi_t=\pi^{\widehat q_t}$. 

Computing the occupancy measure $\widehat q_t$ can be done by solving a linear program as $\langle\bm{\widehat f_t}-\lambda_t\bm{\widehat g_t},\bm{q}\rangle$ is linear and $\Delta(P,t)$ is a polytope with respect to the vector representation $\bm{q}$ of occupancy measure $q$. In fact, the associated policy $\pi_t$ as well as $\widehat q_t$ can also be computed by an efficient backward dynamic programming algorithm~\citep{Jin2020,efroni2020}.

Next, the algorithm executes policy $\pi_t$ for episode $t$. The algorithm stops if the remaining budget becomes less than 1. Remember that $g_t(s,a,h)\in[0,1]$ for any $(s,a,h,t)\in\cS\times\cA\times [H]\times[T]$. Hence, we would not violate the resource consumption constraint if we run the process only when the remaining resource budget is greater than or equal to 1. %
 
At the end of each episode, the algorithm updates the dual variable for the resource consumption constraint. The dual update rule 
$$\lambda_{t+1}=\max\left\{0,\lambda_t - \eta\left(H\rho - \langle \bm{\widehat g_t}, \bm{\widehat q_t}\rangle\right)\right\}$$
follows the standard dual-based algorithm. %
Note that we have a single resource consumption constraint, in which case  $H\rho - \langle \bm{\widehat g_t}, \bm{\widehat q_t}\rangle$ is a scalar. In fact, our framework easily extends to multiple resource constraints, for which we use a vector of dual variables $\bm{\lambda}\in\mathbb{R}_+^m$ where $m$ is the number of resource constraints. 

We have the following guarantees on the performance of our online dual method.
\begin{theorem}\label{theorem:regret1}
	Under the observe-then-decide regime, \Cref{alg:online-alloc-mdp-unknown} with step size $\eta = 1/(\rho H\sqrt{T})$ guarantees
	\begin{align*}\mathbb{E}\left[\regret\left(\vec\gamma,\vec\pi\right)\mid P\right]
		= O\left(\rho^{-1}\left({H^{3/2}}S\sqrt{AT} +{H^{5/2}}S^2A \right)\left(\ln HSAT\right)^2\right)
		\end{align*}
	where the expectation is taken with respect to the randomness of the reward and resource consumption functions and the randomness in the trajectories of episodes. 
\end{theorem}

\begin{theorem}\label{theorem:regret2}
	Under the decide-then-observe regime, \Cref{alg:online-alloc-mdp-unknown} with step size $\eta = 1/(\rho H\sqrt{T})$ guarantees
	\begin{align*}
 \regret\left(\vec\gamma,\vec\pi\right)
		= O\left(\rho^{-1}\left({H^{3/2}}S\sqrt{AT} +{H^{5/2}}S^2A \right)\left(\ln \frac{HSAT}{\delta}\right)^2\right)
	\end{align*}
	with probability at least $1-16\delta$.
\end{theorem}
Note that the bound for the observe-then-decide regime is on the expected regret while we provide a high probability bound for the decide-then-observe regime. This is because we take a dynamic policy as a benchmark for the first setting while we take a static optimal policy with respect to mean reward and resource functions for the second setting. Another remark is that all settings incur regrets of the same asymptotic growth. This is because the largest regret factor in each setting comes from learning the unknown transition function. 

Our regret upper bounds are nearly optimal as demonstrated by the following regret lower bound. There is a  gap of a $O(\sqrt{S})$ factor as well as a polylog factor.

\begin{theorem}{\rm \citep{pmlr-v132-domingues21a}}\label{theorem:regret3}
There is an instance of a finite-horizon episodic MDP with determistic reward and resource consumption functions and unknown transition function for which any algorithm incurs a regret of $\Omega(H^{3/2}\sqrt{SAT})$.
\end{theorem}
In fact, the instance of~\cite{pmlr-v132-domingues21a} has no resource budget constraint, which is equivalent to setting $g_t=0$ for $t\in[T]$ in our setting.

\section{REGRET ANALYSIS}\label{sec:regretanalysis}

Let $T_{\text{stop}}$ be the episode in or right after which \Cref{alg:online-alloc-mdp-unknown} terminates. For $t> T_{\text{stop}}$, we set $\pi_t(a_{\text{stop}} \mid s, h)=1$ for any $(s,h)$ where action $a_{\text{stop}}$ incurs no reward and resource consumption. 
Moreover, if \Cref{alg:online-alloc-mdp-unknown} terminates after step $h_{\text{stop}}\in[H]$ in episode $T_{\text{stop}}$, then we take action $a_{\text{stop}}$ for step $h> h_{\text{stop}}$. %

For the observe-then-decide regime, we have
\begin{align*}
\regret\left(\vec\gamma,\vec\pi\right)&=\underbrace{\opt(\vec\gamma) - \sum_{t=1}^T \langle \bm{f_t}, \bm{ q_t^*}\rangle}_{\text{(I)}} + \underbrace{ \sum_{t=1}^T \langle \bm{f_t}, \bm{q_t^*}-\bm{\widehat q_t}\rangle}_{\text{(II)}} + \underbrace{\sum_{t=1}^T \langle \bm{f_t}, \bm{\widehat q_t}-\bm{q_t}\rangle}_{\text{(III)}}+ \underbrace{\sum_{t=1}^T \langle \bm{f_t}, \bm{q_t}\rangle-\reward\left(\vec\gamma,\vec\pi\right)}_{\text{(IV)}}
\end{align*}
where $q_t^*$ denotes the occupancy measure $q^{P,\pi_t^*}$ for $t\in[T]$ and $\vec \pi^*=(\pi_1^*,\ldots,\pi_T^*)$ is the benchmark optimal policy for the observe-then-decide setting. Terms (I) and (IV) are due to the randomness in the trajectories, and each of them is the sum of some martingale difference sequence. Term (II) is the regret associated with our online dual method, and term (III) is incurred from learning the unknown transition kernel.

For the decide-then-observe regime, we have
 \begin{align*}
	\regret\left(\vec\gamma,\vec\pi\right)
	&=\underbrace{\opt(\vec\gamma) - \sum_{t=1}^T \langle \bm{f}, \bm{q^*}\rangle}_{\text{(I)}}+\underbrace{\sum_{t=1}^T \langle \bm{f}, \bm{q^*}\rangle-\sum_{t=1}^T \langle \bm{\widehat f_t}, \bm{\widehat q_t}\rangle}_{\text{(II)}}+ \underbrace{\sum_{t=1}^T \langle \bm{\widehat f_t}, \bm{\widehat q_t}-\bm{q_t}\rangle}_{\text{(III)}}\\
 &\quad+\underbrace{\sum_{t=1}^T \langle \bm{\widehat f_t}-\bm{f}, \bm{q_t}\rangle}_{\text{(IV)}} + \underbrace{\sum_{t=1}^T \langle \bm{f}, \bm{q_t}\rangle- \reward\left(\vec\gamma,\vec\pi\right)}_{\text{(V)}}
\end{align*}
where $q^*$ denotes the occupancy measure $q^{P,\pi^*}$ and $\pi^*$ is the benchmark optimal policy for the decide-then-observe setting. Terms (I) and (V) are due to the randomness in the trajectories and the reward function, and as before, each of them is the sum of some martingale difference sequence. Term (II) is associated with our online dual method, term (III) comes from learning the unknown transition kernel, and term (IV) is incurred while learning the mean reward function.

\subsection{Regret under the Observe-then-decide Regime}

Let $n^{P,\phi}(s,a,h)$ be defined as the indicator variable for the event that state-action pair $(s,a)\in \cS\times \cA$ is visited at step $h\in[H]$ of an episode under transition function $P$ and an arbitrary policy $\phi$. %
Moreover, for an arbitrary function $\ell:\cS\times\cA\times[H]\to[0,1]$, we have
$$\sum_{h=1}^H \ell\left(s_{h}^{P,\phi},a_{h}^{P,\phi},h\right)= \langle\bm{\ell},\bm{n^{P,\phi}}\rangle$$
where $\bm{\ell},\bm{n^{P,\phi}}$ is the vector representation of $\ell,n^{P,\phi}$. 

For $t\in[T]$, let $\phi_t$ denotes an arbitrary policy for episode $t$, and let $P_t$ be any transition kernel from $\cP_t$. Then the following lemmas hold.

\begin{lemma}\label{lemma:regret-term3}
Let $\ell_t:\cS\times \cA\times[H]\to[0,1]$ be an arbitrary reward function for episode $t\in[T]$.
Then with probability at least $1-6\delta$, 
\begin{align*}
\sum_{t=1}^T\left| \langle   \bm{\ell_t}, \bm{q^{P,\phi_t}}-\bm{n^{P,\phi_t}}\rangle\right|%
\end{align*}
\end{lemma}
Note that (I) can be written as $\sum_{t=1}^T \langle \bm{f_t},\bm{n_t^*}-\bm{q_t^*}\rangle$ where $n_t^*=n^{P,\pi_t^*}$ and (IV) equals $\sum_{t=1}^T \langle \bm{f_t},\bm{q_t}-\bm{n_t}\rangle$ where $n_t=n^{P,\pi_t}$, so they can be bounded by \Cref{lemma:regret-term3}. Term (III) can be bounded based on the next lemma.

\begin{lemma}\label{lemma:regret-term2}
Let $\ell_t:\cS\times \cA\times[H]\to[0,1]$ be an arbitrary function for $t\in[T]$.
Then with probability at least $1-4\delta$, 
\begin{align*}
\sum_{t=1}^T\left| \langle \bm{\ell_t}, \bm{q^{P_t,\phi_t}}-\bm{q^{P,\phi_t}}\rangle\right|=O\left(\left(H^{3/2}S\sqrt{AT} +H^{5/2}S^2A\right)\left(\ln\frac{HSAT}{\delta}\right)^2 \right).
\end{align*}
\end{lemma}
The last component is to bound term (II), which comes from the online dual mirror descent algorithm. The subtle part is to consider a stopping time $\tau$ not to violate the resource budget. To bound (III), we will need to an upper bound on
$$\frac{1}{\rho}\sum_{t=1}^\tau \langle\bm{g_t},\bm{n_t}-\bm{q_t}\rangle+\frac{1}{\rho}\sum_{t=1}^\tau \langle\bm{g_t},\bm{q_t}-\bm{\widehat q_t}\rangle,$$
for which we can apply Lemmas~\ref{lemma:regret-term3} and~\ref{lemma:regret-term2}.

\subsection{Regret under the Decide-then-observe Regime}

Term (I) equals $$\sum_{t=1}^T\langle\bm{f_t},\bm{n^*}-\bm{q^*}\rangle+\sum_{t=1}^T\langle\bm{f_t}-\bm{f},\bm{q^*}\rangle$$
where $n^*=n^{P,\pi^*}$. Here, the first sum can be bounded based on \Cref{lemma:regret-term3} while the second one can be bounded using a concentration inequality. \Cref{lemma:regret-term2} applies to bound term (III). To bound term (IV), we show the following lemma.
\begin{lemma}\label{lemma:second-term4}
Suppose that the statements of \Cref{lemma:estimator} hold. %
Then%
\begin{align*}
	\sum_{t=1}^T \langle \bm{\widehat f_t} - \bm{f}, \bm{q_t}\rangle,\ \sum_{t=1}^T \langle\bm{g}- \bm{\widehat g_t} , \bm{q_t}\rangle=O\left((H\sqrt{SAT} + HSA)\left(\ln \frac{HSAT}{\delta}\right)^2\right).
\end{align*}
\end{lemma}

Next, to bound term (V), we show the following.
\begin{lemma}\label{lemma:second-term5}
	Suppose that the statements of \Cref{lemma:estimator} hold. Then the following statement holds.%
\begin{align*}
	\sum_{t=1}^T \langle \bm{f}, \bm{q_t}\rangle- \sum_{t=1}^T \langle \bm{f_t}, \bm{n_t}\rangle,\ \sum_{t=1}^T \langle \bm{g_t}, \bm{n_t}\rangle-\sum_{t=1}^T \langle \bm{g}, \bm{q_t}\rangle= O\left(\left(H\sqrt{SAT} + H^{3/2} S\sqrt{A}\right)\left(\ln\frac{HSAT}{\delta}\right)^2 \right).
\end{align*}
\end{lemma}
To finish the proof of~\Cref{theorem:regret2}, we bound term (II), which is incurred from running the online dual mirror descent algorithm. As the previous setting, we consider a stopping time $\tau$ not to violate the resource budget. Moreover, we need to bound
\begin{align*}\frac{1}{\rho}\sum_{t=1}^\tau\left( \langle\bm{ g_t},\bm{n_t}\rangle-\langle\bm{g},\bm{q_t}\rangle\right)+\frac{1}{\rho}\sum_{t=1}^\tau\langle\bm{g}-\bm{\widehat g_t},\bm{q_t}\rangle+\frac{1}{\rho}\sum_{t=1}^\tau \langle\bm{\widehat g_t},\bm{q_t}-\bm{\widehat q_t}\rangle
	\end{align*}
 for which we can apply Lemmas~\ref{lemma:regret-term2}--\ref{lemma:second-term5}. 

\section{NUMERICAL EXPERIMENT}

To present a potential use case for our setting we examine a representative inventory management problem, often characterized by cyclical finite horizon episodes~\citep{cyclic-demand}. Specifically, we consider a finite-horizon episodic MDP over $T$ episodes. In each episode $t \in [T]$, a customer of type $k \in [K]$ presents a demand $d$ for item $i_k$, with the sequence of future arrivals undisclosed to the decision-maker.

Each episode is a single-item inventory scenario, defined by MDP $\mathcal{M}=(\mathcal{S}, \mathcal{A}, H, \{P_h\}_{h=1}^{H-1}, p)$. States $s \in \mathcal{S}=[S]$ represent inventory levels, and actions $a \in \mathcal{A}$ denote order quantities, with $S$ being the maximum inventory. There's a fixed order cost $c_f$, holding cost $c_h$, and demand $d$, starting each cycle with zero inventory.

For inventory state $s \in \mathcal{S}$, the reward function is $f_k(s,a,h) = r_k \min(d, s)$, with $r_k \in \mathbb{R}$ for $k \in [K]$. The cost per step comprises the quantity ordered, the fixed order cost if ordering occurs, and the holding charge, expressed as $g_k(s,a,h) = a+c_f + c_h \max(s-d, 0)$.

Our framework can accommodate non-stationary and indeterminate transition kernels ${P_h}_{h=1}^{H-1}$ within $\mathcal{M}$. Transitions from state $(s,a,h)$ are governed by $s \to s + a - d - \zeta_s + \zeta_a$, where $\zeta_s$ is a state-dependent shock, assumed to follow $\text{Pois}(z(s))$. The term $\zeta_a \sim \mathcal{N}(0,\sigma(a))$ reflects uncertainties in receiving the order amount.

In \cref{fig:exp_comparison} we contrast the efficacy of \Cref{alg:online-alloc-mdp-unknown} with the traditional Economic Order Quantity (EOQ) policy, setting parameters as $T=30$ episodes, $H=5$ horizon length, $S=15$ maximum inventory, and $K=3$ customer types with returns $r_0 > r_1 > r_2$. We implement these policies and our algorithm and generate the results on an Apple M1 MacBook Pro. Results are averaged across 30 runs, where randomness is over the customer arrival types (uniformly distributed) as well as the $\zeta_s$ and $\zeta_a$ shocks in the state transitions. The EOQ policy, defined by $a = \sqrt{2dc_f/c_h}$ at $s=0$, disregards customer type and its respective $r_k$. Consequently, our algorithm's superiority over the EOQ approach is to be expected. We also benchmark against a "selective" strategy, excluding customer type $2$ (return $r_2$). As per \cref{fig:exp_comparison}, our algorithm still obtains higher mean cumulative reward, implying nuanced advantages beyond simply customer prioritization.

\begin{figure}[ht]
\centering
\includegraphics[width=0.6\textwidth]{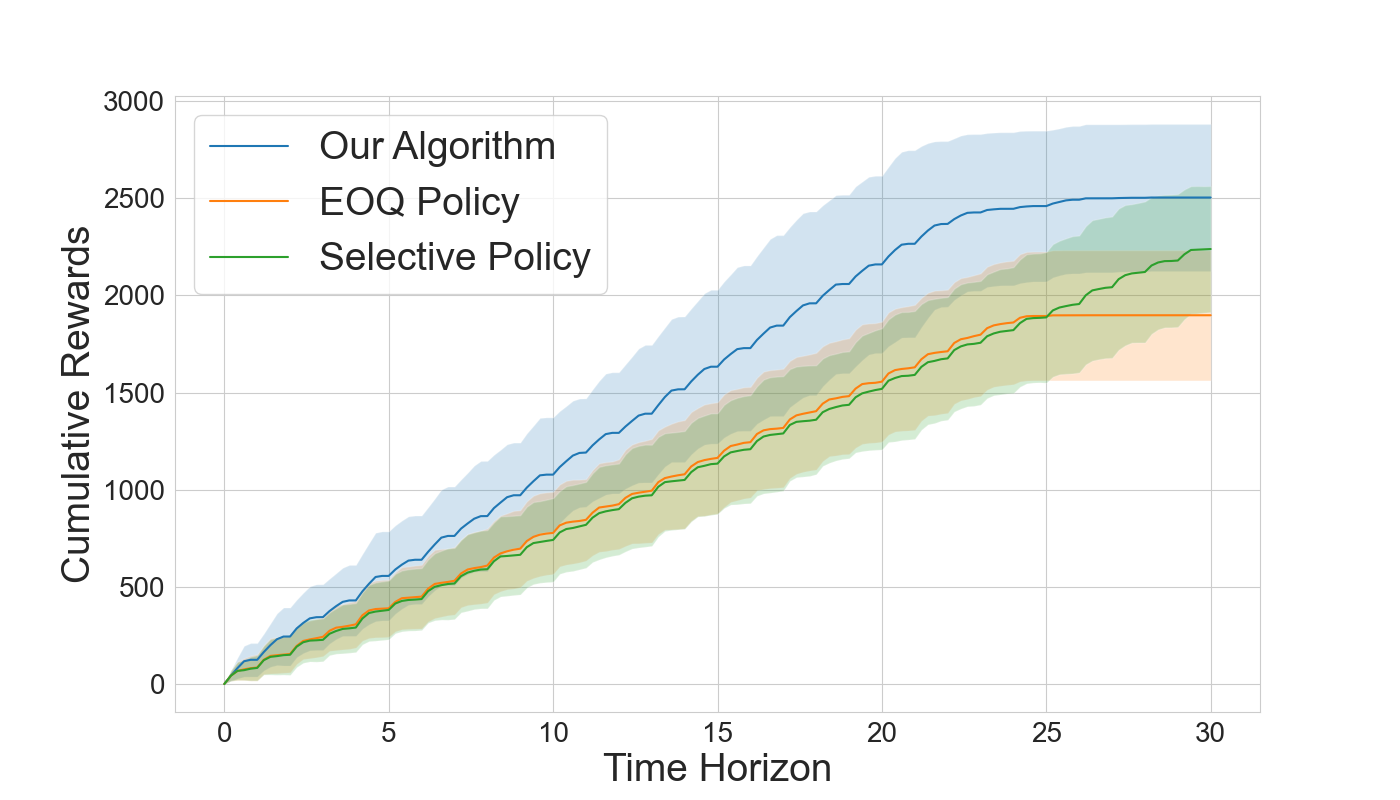}
\caption{Reward Comparison: Algorithm vs. EOQ Policy vs. Selective Policy}
\label{fig:exp_comparison}
\end{figure}

\paragraph{Acknowledgements}
This research is supported, in part, by the KAIST Starting Fund (KAIST-G04220016), the FOUR Brain Korea 21 Program (NRF-5199990113928), the National Research Foundation of Korea (NRF-2022M3J6A1063021).

\bibliography{mybibfile}

\newpage

\appendix

\section{AUXILIARY MEASURES AND NOTATIONS}

In this section, we define some auxiliary measures and functions that are useful for the analysis of the online dual algorithm (\Cref{alg:online-alloc-mdp-unknown}).

Given a policy $\pi$, we define the \emph{reward-to-go function} for a state $s\in\cS$ at step $h$ with reward function $\ell:\cS\times\cA\times[H]\to[0,1]$ and transition kernel $P$ as follows.
\begin{equation}\label{reward-to-go}
J^{P,\pi,\ell}(s,h) = \mathbb{E}\left[\sum_{j=h}^{H}\ell\left(s_{j}^{P,\pi},a_{j}^{P,\pi},j\right)\mid \ell,\pi,P, s_h^{P,\pi}=s\right].
\end{equation}
Similarly, we define the \emph{state-action value function} for $(s,a)\in\cS\times \cA$ at step $h$ with reward function $\ell:\cS\times\cA\times[H]\to[0,1]$ and transition kernel $P$ as follows.
\begin{equation}\label{value-function}
Q^{P,\pi,\ell}(s,a,h) = \mathbb{E}\left[\sum_{j=h}^{H}\ell\left(s_{j}^{P,\pi},a_{j}^{P,\pi},j\right)\mid \ell,\pi,P, s_h^{P,\pi}=s,a_h^{P,\pi}=a\right].
\end{equation}
Furthermore, given a policy $\pi$ and  a transition kernel $P$, we define $q^{P,\pi}\left(s,a,h\mid s',m\right)$ as
\begin{equation}\label{occupancy-conditional}
q^{P,\pi}\left(s,a,h\mid s',m\right)=\mathbb{P}\left[s^{ P, \pi}_{h}=s,\ a^{P,\pi}_{h}=a\mid \pi,P, s^{ P, \pi}_{m}=s'\right]
\end{equation}
for $(s,a,s')\in \cS\times \cA\times \cS$ and $1\leq m\leq h\leq H$.

Given two vectors $\bm{u},\bm{v}\in \mathbb{R}^{S\times A\times H}$, let $\bm{u}\odot \bm{v}$ be defined as the vector obtained from coordinate-wise products of $\bm{u}$ and $\bm{v}$, i.e. $(\bm{u}\odot \bm{v})_i = u_i\odot v_i$ for $i\in[SAH]$. Let $\vec h$ be an $(S\times A\times H)$-dimensional vector all of whose coordinates are $h$. 

We define $\xi_1$ as 
$\xi_1=\left\{f_1,g_1\right\}$
and for $t\geq 2$, we define $\xi_t$ as 
$$\left\{s_1^{P,\pi_{t-1}},a_1^{P,\pi_{t-1}},\ldots, s_h^{P,\pi_{t-1}}, a_h^{P,\pi_{t-1}}, f_t,g_t\right\}$$
where $\pi_{t-1}$ denotes the policy for episode $t-1$ and $$\left(s_1^{P,\pi_{t-1}},a_1^{P,\pi_{t-1}},\ldots, s_h^{P,\pi_{t-1}}, a_h^{P,\pi_{t-1}}\right)$$
is the trajectory generated under policy $\pi_{t-1}$ and transition kernel $P$. Then for $t\in[T]$, let $\cF_t$ be defined as the $\sigma$-algebra generated by the random variables in $\xi_1\cup\cdots\cup \xi_{t}$. Then it follows that $\cF_1,\ldots, \cF_T$ give rise to a filtration.

We define $\zeta_1$ as 
$\zeta_1=\emptyset$
and for $t\geq 2$, we define $\zeta_t$ as 
$$\left\{f_{t-1},g_{t-1},s_1^{P,\pi_{t-1}},a_1^{P,\pi_{t-1}},\ldots, s_h^{P,\pi_{t-1}}, a_h^{P,\pi_{t-1}}\right\}.$$ Then for $t\in[T]$, let $\cG_t$ be defined as the $\sigma$-algebra generated by the random variables in $\zeta_1\cup\cdots\cup \zeta_{t}$. Then it follows that $\cG_1,\ldots, \cG_T$ give rise to a filtration.

\section{MISSING PROOFS FOR SECTION 3 %
}

In this section, we first prove \Cref{lemma:valid-occupancy} which characterizes valid occupancy measures for a finite-horizon MDP. Then we prove the lemmas in \Cref{sec:confidence_sets}
which describe important properties of the confidence sets estimating the true transition kernel. Lastly, we prove the lemmas in \Cref{sec:estimator}
which delineate the accuracy of our optimistic estimators of the mean reward function $f$ and the mean resource consumption function $g$.

\subsection{Valid Occupancy Measures}

First, we provide the proof of \Cref{lemma:valid-occupancy} which is based on the reduction to the loop-free MDP setting. 

\begin{proof}[\bf Proof of \Cref{lemma:valid-occupancy}]
Given the finite-horizon MDP associated with transition kernel $P$, we may define a loop-free MDP as follows. We define its state space as $\cS':=\cS\times [H+1]$, which can be viewed as $H+1$ layers $\cS\times\{h\}$ for $h\in[H+1]$. Its transition kernel $P'$ is given by $P'(\tshn\mid \tsh,a) = P(s'\mid s, a,h)$ for $(s,a,s',h)\in \cS\times \cA\times \cS\times [H]$. Next, given $\bar q$, we may define an occupancy measure $q'$ for the loop-free MDP as $q'(\tsh,a,\tshn)=\bar q(s,a,s',h)$ for $(s,a,s',h)\in \cS\times \cA\times \cS\times [H]$. Then it follows from~\cite[Lemma 3.1]{rosenberg2019} that $q'$ is a valid occupancy measure for the loop-free MDP with transition kernel $P'$ if and only if $q'$ satisfies
\begin{align}
&\sum_{(s,a,s')\in \cS\times \cA\times \cS}q'(\tsh,a,\tshn)=1\quad\text{for~} h=1,\ldots, H,\tag{C1'}\\
&\sum_{(s',a)\in\cS\times \cA}q'(\tsh,a,\tshn)=\sum_{(s',a)\in\cS\times\cA}q'((s',h-1),a,\tsh)\quad \text{for~}s\in \cS,\ h\in\{2,\ldots, H\},\tag{C2'}
\end{align}
and $P^{q'}=P'$ where $P^{q'}$ is given by 
$$P^{q'}(\tshn\mid \tsh,a)=\frac{q'(\tsh,a,\tshn)}{\sum_{s''\in \cS} q'(\tsh,a,(s'',h+1))}=\frac{\bar q(s,a,s',h)}{\sum_{s''\in \cS}\bar q(s,a,s'',h)}.$$
Here, the conditions are equivalent to (C1), (C2), and $P^{\bar q}=P$. Moreover, $q'$ is a valid occupancy measure with $P'$ if and only if $q$ is a valid occupancy measure with $P$, as required.
\end{proof}

\subsection{Confidence Sets for the True Transition Kernel}

\Cref{lemma:confidence} 
is a modification of \citep[Lemma 2]{Jin2020} to our finite-horizon MDP setting. We prove \Cref{lemma:confidence} 
using the empirical Bernstein inequality provided in \Cref{bernstein}.

\begin{proof}[\bf Proof of \Cref{lemma:confidence}]
We will show that with probability at least $1-4\delta$, 
\begin{equation}\label{star}
\left|P(s'\mid s,a,h)-\bar P_t(s'\mid s,a,h)\right|\leq \epsilon_t(s'\mid s,a,h)
\end{equation}
where
\begin{equation*}
\epsilon_t(s'\mid s,a,h)= 2\sqrt{\frac{\bar{P}_{t}(s'\mid s,a,h)\ln\left({{HSAT}/{\delta}}\right)}{\max\{1,N_t(s,a,h)-1\}}}+\frac{14\ln\left({{HSAT}/{\delta}}\right)}{3\max\{1,N_t(s,a,h)-1\}}
\end{equation*}
holds for every $(s,a,s',h)\in \cS\times \cA\times \cS\times [H]$ and every episode $t\in[T]$. 

Let us first consider the case $N_t(s,a,h)\leq 1$. As we may assume that $HSAT\geq 2$, it follows that
$$\epsilon_t(s'\mid s,a,h)=\frac{14\ln\left({{HSAT}/{\delta}}\right)}{3\max\{1,N_t(s,a,h)-1\}}\geq \frac{14}{3} \ln 2>1.$$
Then \eqref{star} holds because $0\leq P(s'\mid s,a,h),\bar P_t(s'\mid s,a,h)\leq 1$.

Assume that $n= N_t(s,a,h)\geq 2$. Then we define $Z_1, \ldots, Z_n$ as follows.
$$Z_j=\begin{cases}
1,\quad &\text{if the transition after the $j$th visit to $(s,a,h)$ is $s'$},\\
0,\quad&\text{otherwise}.
\end{cases}$$
Then $Z_1,\ldots, Z_n$ are i.i.d. with mean $P(s'\mid s, a, h)$, and we have
$$\sum_{j=1}^n Z_j = M_t(s, a,s', h).$$
Moreover, the sample variance $V_n$ of $Z_1,\ldots, Z_n$ is given by
\begin{align}\label{eq:lemma:confidence-variance}
\begin{aligned}
V_n &= \frac{1}{N_t(s,a,h)(N_t(s,a,h)-1)} M_t(s,a,s', h)\left(N_t(s,a,h)- M_t(s, a, s',h)\right)\\
&=\frac{N_t(s,a,h)}{(N_t(s,a,h)-1)} \bar P_t(s'\mid s, a, h)\left(1- \bar P_t(s'\mid s, a, h)\right).
\end{aligned}
\end{align}
Then it follows from \Cref{bernstein} that with probability at least $1- 2\delta/(HS^2AT)$,
\begin{align}\label{eq:lemma:confidence-1}
\begin{aligned}
P(s'\mid s,a,h)-\bar P_t(s'\mid s,a,h)
&\leq \sqrt{\frac{2\bar P_t(s'\mid s, a, h)\left(1- \bar P_t(s'\mid s, a, h)\right)\ln\left({{HS^2AT}/{\delta}}\right)}{N_t(s,a,h)-1}}+ \frac{7\ln\left(HS^2AT/\delta\right)}{3(N_t(s,a,h)-1)}.
\end{aligned}
\end{align}
Here, as we assumed that $N_t(s,a,h)\geq 2$, we have $N_t(s,a,h)-1=\max\{1,N_t(s,a,h)-1\}$. In addition, we know that $1- \bar P_t(s'\mid s, a, h)\leq 1$ and that $\ln\left({{HS^2AT}/{\delta}}\right)\leq 2\ln\left(HSAT/\delta\right)$. Then \eqref{eq:lemma:confidence-1} implies that with probability at least $1- 2\delta/(HS^2AT)$,
\begin{equation}\label{eq:lemma:confidence-2}
    P(s'\mid s,a,h)-\bar P_t(s'\mid s,a,h)\leq \epsilon_t(s'\mid s,a,h)
\end{equation}
Next, we apply \Cref{bernstein} to variables $1-Z_1,\ldots, 1-Z_n$ that are i.i.d. and have mean $1-\bar P_t(s'\mid s, a, h)$. Moreover, the sample variance of $1-Z_1,\ldots, 1-Z_n$ is also equal to $V_n$ defined as in~\eqref{eq:lemma:confidence-variance}. Therefore, based on the same argument, we deduce  that 
with probability at least $1- 2\delta/(HS^2AT)$,
\begin{equation}\label{eq:lemma:confidence-3}
   - P(s'\mid s,a,h)+\bar P_t(s'\mid s,a,h)\leq \epsilon_t(s'\mid s,a,h).
\end{equation}
By applying union bound to~\eqref{eq:lemma:confidence-2} and~\eqref{eq:lemma:confidence-3}, with probability at least $1- 4\delta/(HS^2AT)$,~\eqref{star} holds for $(s,a,s',h)$.
Furthermore, by applying union bound over all $(s,a,s',h)\in \cS\times \cA\times\cS\times [H]$, it follows that with probability at least $1-4\delta$,~\eqref{star} holds for every $(s,a,s',h)\in \cS\times \cA\times\cS\times [H]$, as required.
\end{proof}

\Cref{lemma:confidence} 
bounds the difference between the true transition kernel $P$ and the empirical transition kernels $\bar P_t$. Based on \Cref{lemma:confidence}, the next lemma bounds the difference between the true transition kernel $P$ and any $\widehat P$ contained in the confidence sets $\cP_t$.  \Cref{lemma:confidence'} is a modification of \citep[Lemma 8]{Jin2020} to our finite-horizon MDP setting.

\begin{lemma}\label{lemma:confidence'}
Let $t\in [T]$. Assume that the true transition kernel satisfies $P\in \cP_t$. Then we have
\begin{equation}\label{starstar}
\left|\widehat P(s'\mid s,a,h)-P(s'\mid s,a,h)\right|\leq \epsilon_t^*(s'\mid s,a,h)
\end{equation}
where
\begin{equation*}
\epsilon_t^\star(s'\mid s,a,h)= 6\sqrt{\frac{P(s'\mid s,a,h)\ln\left({{HSAT}/{\delta}}\right)}{\max\{1,N_t(s,a,h)\}}}+94\frac{\ln\left({{HSAT}/{\delta}}\right)}{\max\{1,N_t(s,a,h)\}}
\end{equation*}
for every $\widehat P\in \cP_t$ and every $(s,a,s',h)\in\cS\times \cA\times \cS\times [H]$.
\end{lemma}
\begin{proof}
We follow the proof of \citep[Lemma B.13]{cohen2020}. Note that 
$$\max\{1,N_t(s,a,h)-1\}\geq \frac{1}{2}\cdot \max\{1,N_t(s,a,h)\}$$
holds for any value of $N_t(s,a,h)$. As we assumed that $P\in \cP_t$, we have that
$$
\bar P_t(s'\mid s,a,h)\leq P(s'\mid s,a,h)+\sqrt{\frac{8\bar P_t(s'\mid s,a,h)\ln\left({{HSAT}/{\delta}}\right)}{\max\{1,N_t(s,a,h)\}}}+\frac{28\ln\left({{HSAT}/{\delta}}\right)}{3\max\{1,N_t(s,a,h)\}}.$$
We may view this as a quadratic inequality in terms of $x=\sqrt{\bar P_t(s'\mid s,a,h)}$. Note that $x^2\leq ax + b +c$ for any $a,b,c\geq 0$ implies that $x\leq a +\sqrt{b}+\sqrt{c}$. Therefore, we deduce that
\begin{align*}
\sqrt{\bar P_t(s'\mid s,a,h)}&\leq \sqrt{P(s'\mid s,a,h)}+\left(2\sqrt{2} + \sqrt{\frac{28}{3}}\right)\sqrt{\frac{\ln\left({{HSAT}/{\delta}}\right)}{\max\{1,N_t(s,a,h)\}}}\\
&\leq \sqrt{P(s'\mid s,a,h)}+13\sqrt{\frac{\ln\left({{HSAT}/{\delta}}\right)}{\max\{1,N_t(s,a,h)\}}}.
\end{align*}
Using this bound on $\sqrt{\bar P_t(s'\mid s,a,h)}$, we obtain the following. 
\begin{align}\label{eq:lemma:confidence'1}
\begin{aligned}
\epsilon_t(s'\mid s,a,h)
&\leq \sqrt{\frac{8\bar P_t(s'\mid s,a,h)\ln\left({{HSAT}/{\delta}}\right)}{\max\{1,N_t(s,a,h)\}}}+\frac{28\ln\left({{HSAT}/{\delta}}\right)}{3\max\{1,N_t(s,a,h)\}}\\
&\leq \sqrt{\frac{8P(s'\mid s,a,h)\ln\left({{HSAT}/{\delta}}\right)}{\max\{1,N_t(s,a,h)\}}}+\left(13\sqrt{8}+\frac{28}{3}\right)\frac{\ln\left({{HSAT}/{\delta}}\right)}{\max\{1,N_t(s,a,h)\}}\\
&\leq3\sqrt{\frac{P(s'\mid s,a,h)\ln\left({{HSAT}/{\delta}}\right)}{\max\{1,N_t(s,a,h)\}}}+47\frac{\ln\left({{HSAT}/{\delta}}\right)}{\max\{1,N_t(s,a,h)\}}\\
&=\frac{1}{2}\cdot \epsilon_t^\star(s'\mid s,a,h)
\end{aligned}
\end{align}
Since we assumed that $P\in \cP_t$, 
$$\left|P(s'\mid s,a,h)-\bar P_t(s'\mid s,a,h)\right|\leq \frac{1}{2}\cdot\epsilon_t^\star(s'\mid s,a,h).$$
Moreover, for any $\widehat P\in\cP_t$, we have
$$\left|\widehat P(s'\mid s,a,h)-\bar P_t(s'\mid s,a,h)\right|\leq \epsilon_t(s'\mid s,a,h)\leq \frac{1}{2}\cdot\epsilon_t^\star(s'\mid s,a,h).$$
By the triangle inequality, it follows that
$$\left|\widehat P(s'\mid s,a,h)-P(s'\mid s,a,h)\right|\leq \epsilon_t^\star(s'\mid s,a,h),$$
as required.
\end{proof}

\subsection{Optimistic Function Estimators}

\Cref{lemma:estimator} is a modification of \citep[Lemma 19]{ssp-adversarial-unknown} to our finite-horizon MDP setting. We prove \Cref{lemma:estimator} using the concentration inequality provided in \Cref{cohen-concentration0}. Although we follow the same proof outline of \citep[Lemma 19]{ssp-adversarial-unknown}, we provide the proof here to make our paper self-contained.

\begin{proof}[\bf Proof of \Cref{lemma:estimator}]
If $\sum_{k=1}^{t-1} \mathbbm{1}_k(s,a,h)=0$, then $\bar f_t(s,a,h)=\bar g_t(s,a,h)=0$ while $7R_t(s,a,h)\geq 7\ln 2 \geq 1$, in which case we have $\widehat f_t(s,a,h) =1$ and $\widehat g_t(s,a,h) = 0$. Hence, if $\sum_{k=1}^{t-1} \mathbbm{1}_k(s,a,h)=0$, the statements of the lemma are trivially satisfied. Thus we may assume that $\sum_{k=1}^{t-1} \mathbbm{1}_k(s,a,h)\geq1$ in which case $C_t(s,a,h) = \sum_{k=1}^{t-1} \mathbbm{1}_k(s,a,h)$. Moreover, we have $C_t(s,a,h)\leq T$.

Applying \Cref{cohen-concentration0} on the first $C_t(s,a,h)$ realizations of the random reward for $(s,a,h)$, we deduce that 
\begin{equation}\label{estimator1}
\left|\bar f_t(s,a,h) - f(s,a,h)\right|\leq 2\sqrt{R_t(s,a,h)\bar f_t(s,a,h)} + 7R_t(s,a,h)
\end{equation}
holds for every $t\in[T]$ with probability at least $1- \delta /(SAH)$. Then by taking the union bound for all $(s,a,h)\in \cS\times \cA\times [H]$, with probability at least $1-\delta$, \eqref{estimator1} holds for all $(s,a,h)\in \cS\times\cA\times [H]$ and all $t\in[H]$. Recall that
$$\widehat f_t(s,a,h) = \min\left\{1, \bar f_t(s,a,h) +2\sqrt{R_t(s,a,h)\bar f_t(s,a,h)} + 7R_t(s,a,h) \right\}.$$
This means thst 
$$ \widehat f_t(s,a,h) -f(s,a,h)=  \widehat f_t(s,a,h)-\bar f_t(s,a,h)+ \bar f_t(s,a,h)-f(s,a,h)\leq 2\left(2\sqrt{R_t(s,a,h)\bar f_t(s,a,h)} + 7R_t(s,a,h)\right)$$
Since $f(s,a,h)\leq 1$, it follows that 
$$0\leq \widehat f_t(s,a,h) - f(s,a,h)$$
holds if \eqref{estimator1} holds. Furthermore, if $x\leq a\sqrt{x} +b$ holds for some $x,a,b\geq 0$, then we have $x\leq (a+\sqrt{b})^2$. Applying this to~\eqref{estimator1} with $x= \bar f_t(s,a,h)$, it follows that
$$\bar f_t(s,a,h)\leq f(s,a,h) + 23R_t(s,a,h) +4\sqrt{R_t(s,a,h)f(s,a,h)}\leq 3f(s,a,h) + 25 R_t(s,a,h).$$
Based on this inequality, we deduce that
$$2\sqrt{R_t(s,a,h)\bar f_t(s,a,h)} + 7R_t(s,a,h)\leq 4\sqrt{R_t(s,a,h) f(s,a,h)} +17 R_t(s,a,h).$$
Therefore, it follows that
$$0\leq \widehat f_t(s,a,h) -f(s,a,h)\leq 8\sqrt{R_t(s,a,h) f(s,a,h)} +34 R_t(s,a,h)$$
holds for all $(s,a,h,t)\in \cS\times\cA\times[H]\times [T]$ with probability at least $1-\delta$. 
as required.
Likewise, $$  0\leq g(s,a,h)-\widehat g_t(s,a,h)\leq 8\sqrt{R_t(s,a,h) f(s,a,h)} +34 R_t(s,a,h)$$
holds for all $(s,a,h,t)\in \cS\times\cA\times[H]\times [T]$ with probability at least $1-\delta$. By taking the union bound, both simultaneously hold with probability at least $1-2\delta$.
\end{proof}

\section{TECHNICAL LEMMAS}

In this section, we prove technical lemmas that are crucial in proving the desired upper bounds on the regret. In particular, our regret analysis heavily depend on Lemmas~\ref{lemma2} and~\ref{lemma9}.

Given two vectors $\bm{u},\bm{v}\in \mathbb{R}^{S\times A\times H}$, let $\bm{u}\odot \bm{v}$ be defined as the vector obtained from coordinate-wise products of $\bm{u}$ and $\bm{v}$, i.e. $(\bm{u}\odot \bm{v})_i = u_i\odot v_i$ for $i\in[SAH]$. Moreover, we define $\vec h$ as an $SAH$-dimensional vector whose coordinate for $(s,a,h)\in \cS\times \cA\times [H]$ is $h$. The following lemma is from \citep{ssp-adversarial-unknown}, and it is useful to bound the variance of 
$\langle \bm{n_t}, \bm{f_t}\rangle$.
\begin{lemma}{\rm \citep[Lemma 2]{ssp-adversarial-unknown}}\label{lemma2}
	Let $\pi_t$ be any policy for episode $t$, and let $q_t$ denote the occupancy measure $q^{P,\pi_t}$. Let $\ell:\cS\times \cA\times[H]\to[0,1]$ be an arbitrary reward function. Then
	$$\mathbb{E}\left[\langle \bm{n_t}, \bm{\ell}\rangle^2\mid \ell, \pi_t, P\right]\leq  2\langle \bm{q_t},\vec h\odot \bm{\ell}\rangle$$
	where $\bm{q_t}, \bm{n_t},\bm{\ell}$ are the vector representations of $q_t,n_t,\ell:\cS\times\cA\times[H]\to\mathbb{R}.$
\end{lemma}
\begin{proof}
For ease of notation, let $\mathbb{E}_t\left[\cdot\right]$ denotes $\mathbb{E}\left[\cdot\mid \ell, \pi_t, P\right]$, and let $s_h$ and $a_h$ denote $s_h^{P,\pi_t}$ and $a_h^{P,\pi_t}$, respectively for $h\in[H]$. Note that
\begin{align*}
\mathbb{E}_t\left[\langle \bm{n_t}, \bm{\ell}\rangle^2\right]
&=\mathbb{E}_t\left[\left(\sum_{h=1}^H\sum_{(s,a)\in\cS\times\cA} n_t(s,a,h)\ell(s,a,h)\right)^2\right]\\
&=\mathbb{E}_t\left[\left(\sum_{h=1}^H \ell(s_h,a_h,h)\right)^2\right]\\
&\leq 2\mathbb{E}_t\left[\sum_{h=1}^H \ell(s_h,a_h,h)\left(\sum_{m=h}^H \ell(s_m,a_m,m)\right)\right]\\
&=2\mathbb{E}_t\left[\sum_{h=1}^H\mathbb{E}_t\left[\ell(s_h,a_h,h)\left(\sum_{m=h}^H\ell(s_m,a_m,m)\right)\mid s_h, a_h\right]\right]\\
&=2\mathbb{E}_t\left[\sum_{h=1}^H\ell(s_h,a_h,h)\mathbb{E}_t\left[\sum_{m=h}^H\ell(s_m,a_m,m)\mid s_h, a_h\right]\right]\\
&=2\mathbb{E}_t\left[\sum_{h=1}^H\ell(s_h,a_h,h)Q^{P,\pi_t,\ell}(s_h,a_h,h) \right]\\
&=2\mathbb{E}_t\left[\sum_{h=1}^H\sum_{(s,a)\in\cS\times\cA}n_t(s,a,h)\ell(s,a,h)Q^{P,\pi_t,\ell}(s,a,h) \right]
\end{align*}
where the first inequality holds because $(\sum_{h=1}^H x_h)^2\leq 2\sum_{h=1}^H x_h (\sum_{m=h}^H x_h)$. Moreover,
\begin{align*}
\mathbb{E}_t\left[\sum_{h=1}^H\sum_{(s,a)\in\cS\times\cA}n_t(s,a,h)\ell(s,a,h)Q^{P,\pi_t,\ell}(s,a,h) \right]
&= \sum_{h=1}^H\sum_{(s,a)\in\cS\times\cA}\ell(s,a,h)Q^{P,\pi_t,\ell}(s,a,h)\mathbb{E}_t\left[n_t(s,a,h)\right]\\
&= \sum_{h=1}^H\sum_{(s,a)\in\cS\times\cA}\ell(s,a,h)Q^{P,\pi_t,\ell}(s,a,h)q_t(s,a,h)\\
&= \langle \bm{q_t}, \bm{\ell}\odot\bm{Q^{P,\pi_t,\ell}}\rangle.
\end{align*}
Therefore, it follows that
$$\mathbb{E}_t\left[\langle \bm{n_t}, \bm{\ell}\rangle^2\right]\leq \langle \bm{q_t}, \bm{\ell}\odot\bm{Q^{P,\pi_t,\ell}}\rangle.$$
Next, observe that
\begin{align*}
\langle \bm{q_t},\bm{\ell}\odot \bm{Q^{P,\pi_t,\ell}}\rangle&\leq \sum_{h=1}^H\sum_{(s,a)\in\cS\times\cA}Q^{P,\pi_t,\ell}(s,a,h)q_t(s,a,h)\\
&=\sum_{h=1}^H\sum_{(s,a)\in\cS\times\cA} \pi(a\mid s, h)Q^{P,\pi_t,\ell}(s,a,h)\left(\sum_{a'\in \cA}q_t(s,a',h)\right)\\
&=\sum_{h=1}^H\sum_{s\in\cS} J^{P,\pi_t,\ell}(s,h)\left(\sum_{a'\in \cA}q_t(s,a',h)\right)\\
&=\sum_{h=1}^H\sum_{s\in\cS} \left(\sum_{m=h}^H\sum_{(s',a')\in \cS\times\cA} q_t(s',a',m\mid s,h)\ell(s',a',m)\right)\left(\sum_{a'\in \cA}q_t(s,a',h)\right)\\
&=\sum_{h=1}^H\sum_{m=h}^H\sum_{(s',a')\in \cS\times\cA}\sum_{s\in\cS}q_t(s',a',m\mid s,h)\left(\sum_{a'\in \cA}q_t(s,a',h)\right)\ell(s',a',m)\\
&=\sum_{h=1}^H\sum_{m=h}^H\sum_{(s',a')\in \cS\times\cA}q_t(s',a',m)\ell(s',a',m)\\
&=\sum_{h=1}^H \sum_{(s,a)\in \cS\times\cA}h\cdot q_t(s,a,h)\ell(s,a,h)\\
&=\langle \bm{q_t},\vec h\odot \bm{\ell}\rangle
\end{align*}
where the first inequality holds because $\ell(s,a,h)\leq 1$ for any $(s,a,h)$, the first equality holds because
$$q_t(s,a,h)= \pi(a\mid s, h)\sum_{a'\in \cA}q_t(s,a',h),$$
the fifth equality follows from 
$$\sum_{s\in\cS}q_t(s',a',m\mid s,h)\left(\sum_{a'\in \cA}q_t(s,a',h)\right)=q_t(s',a',m).$$
Therefore, we get that $\langle \bm{q_t}, \bm{\ell}\odot \bm{Q^{P,\pi_t,\ell}}\rangle\leq \langle \bm{q_t},\vec h\odot \bm{\ell}\rangle$, as required.
\end{proof}

The following lemma is from the first statement of \cite[Lemma 7]{ssp-adversarial-unknown} with a few modifications to adapt the proof to our setting.
\begin{lemma}{\rm \citep[Lemma 7]{ssp-adversarial-unknown}}\label{lemma7}
Let $\pi$ be a policy, and let $\widetilde P,\widehat P$ be two different transition kernels. We denote by $\widetilde q$ the occupancy measure $q^{\widetilde P,\pi}$ associated with $\widetilde P$ and $\pi$, and we denote by $\widehat q$ the occupancy measure $q^{\widehat P,\pi}$ associated with $\widehat P$ and $\pi$. Then
\begin{align*}
\widehat q(s,a,h) - \widetilde q(s,a,h)
&=\sum_{(s',a',s'')\in\cS\times \cA\times \cS}\sum_{m=1}^{h-1}\widetilde q(s',a',m)\left(\widehat P(s''\mid s',a',m)- \widetilde P(s''\mid s',a',m)\right)\widehat q(s,a,h\mid s'',m+1).
\end{align*}
\end{lemma}
\begin{proof}
We prove the first statement by induction on $h$. When $h=1$, note that
$$\widehat q(s,a,h)= \widetilde q(s,a,h) = \pi(a\mid s,1)\cdot p(s).$$
Hence, both the left-hand side and right-hand side are equal to 0. Next assume that the equality holds with $h-1\geq 1$. Then we consider $h$. By the definition of occupancy measures,
\begin{align*}
\widehat q(s,a,h) - \widetilde q(s,a,h)
&=\pi(a\mid s,h)\sum_{(s',a')\in\cS\times \cA}(\widehat P(s\mid s',a',h-1)\widehat q(s',a',h-1)-\widetilde P(s\mid s',a',h-1)\widetilde q(s',a',h-1))\\
&=\underbrace{\pi(a\mid s,h)\sum_{(s',a')\in\cS\times \cA}\widehat P(s\mid s',a',h-1)(\widehat q(s',a',h-1)-\widetilde q(s',a',h-1))}_{\text{Term 1}}\\
&\qquad + \underbrace{\pi(a\mid s,h)\sum_{(s',a')\in\cS\times \cA}\widetilde q(s',a',h-1)(\widehat P(s\mid s',a',h-1)-\widetilde P(s\mid s',a',h-1))}_{\text{Term 2}}.
\end{align*}
To provide an upper bound on Term 1, we use the induction hypothesis for $h-1$:
\begin{align*}
&\widehat q(s',a',h-1) - \widetilde q(s',a',h-1)\\
&=\sum_{(s'',a'',s''')\in\cS\times\cA\times\cS}\sum_{m=1}^{h-2}\widetilde q(s'',a'',m)\left((\widehat P- \widetilde P)(s'''\mid s'',a'',m)\right)\widehat q(s',a',h-1\mid s''',m+1)
\end{align*}
where 
$$(\widehat P- \widetilde P)(s'''\mid s'',a'',m)=\widehat P(s'''\mid s'',a'',m)- \widetilde P(s'''\mid s'',a'',m).$$
In addition, observe that
$$\pi(a\mid s,h)\sum_{(s',a')\in\cS\times \cA} \widehat P(s\mid s',a',h-1)\widehat q(s',a',h-1\mid s''',m+1)=\widehat q(s,a,h\mid s''',m+1).$$
Therefore, it follows that Term 1 is equal to
\begin{align*}
&\sum_{(s'',a'',s''')\in\cS\times\cA\times\cS}\sum_{m=1}^{h-2}\widetilde q(s'',a'',m)\left((\widehat P- \widetilde P)(s'''\mid s'',a'',m)\right)\widehat q(s,a,h\mid s''',m+1)\\
&=\sum_{(s',a',s'')\in\cS\times\cA\times\cS}\sum_{m=1}^{h-2}\widetilde q(s',a',m)\left(\widehat P(s''\mid s',a',m)- \widetilde P(s''\mid s',a',m)\right)\widehat q(s,a,h\mid s'',m+1).
\end{align*}

Next, we upper bound Term 2. Note that
$$\widehat q(s,a,h\mid s'',h) = \pi(a\mid s'',h)\cdot \mathbf{1}\left[s''=s\right].$$
Then it follows that
\begin{align*}
&\pi(a\mid s,h)(\widehat P(s\mid s',a',h-1)-\widetilde P(s\mid s',a',h-1))\\
&= \sum_{s''\in\cS} \mathbf{1}\left[s''=s\right]\cdot \pi(a\mid s'',h)(\widehat P(s''\mid s',a',h-1)-\widetilde P(s''\mid s',a',h-1))\\
&= \sum_{s''\in\cS}\widehat q(s,a,h\mid s'',h)(\widehat P(s''\mid s',a',h-1)-\widetilde P(s''\mid s',a',h-1)),
\end{align*}
implying in turn that Term 2 equals
$$\sum_{(s',a',s'')\in\cS\times \cA\times \cS}\widetilde q(s',a',h-1)(\widehat P(s''\mid s',a',h-1)-\widetilde P(s''\mid s',a',h-1))\widehat q(s,a,h\mid s'',h).$$
Adding the equivalent expression of Term 1 and that of Term 2 that we have obtained, we get the right-hand side of the statement.
\end{proof}

Based on~\Cref{lemma:confidence'} and~\Cref{lemma7}, we show the following lemma, which is a modification of \cite[Lemma 7, the second statement]{ssp-adversarial-unknown}.
\begin{lemma}\label{lemma:confidence''}
Let $\pi$ be a policy, and let $\widetilde P,\widehat P$ be two different transition kernels. We denote by $\widetilde q$ the occupancy measure $q^{\widetilde P,\pi}$ associated with $\widetilde P$ and $\pi$, and we denote by $\widehat q$ the occupancy measure $q^{\widehat P,\pi}$ associated with $\widehat P$ and $\pi$.
If $\widehat P, \widetilde P\in \cP_t$, then we have
\begin{align*}
\left|\langle\bm{\widehat q} - \bm{\widetilde q}, \bm{\ell} \rangle\right|
&=\left|\sum_{(s,a,s',h)\in\cS\times\cA\times\cS\times[H]}\widetilde q(s,a,h)\left(\widehat P(s'\mid s,a,h)-\widetilde P(s'\mid s,a,h)\right)J^{\widehat P, \pi,\ell}(s',h+1)\right|\\
&\leq H\sum_{(s,a,s',h)\in\cS\times \cA\times\cS\times[H]}\widetilde q(s,a,h)\epsilon_t^\star (s'\mid s,a,h)
\end{align*}
where $\bm{\widehat q},\bm{\widetilde q}, \bm{\ell}$ are the vector representations of $\widehat q, \widetilde q, \ell:\cS\times\cA\times[H]\to\mathbb{R}.$
\end{lemma}
\begin{proof}
First, observe that
\begin{align*}
\langle\bm{\widehat q} - \bm{\widetilde q}, \bm{\ell} \rangle&=\sum_{(s,a,h)\in\cS\times\cA\times[H]}\left(\widehat q(s,a,h)-\widetilde q(s,a,h)\right)\ell(s,a,h).
\end{align*}
By \Cref{lemma7}, the right-hand side can be rewritten so that we obtain the following.
\begin{align*}
\langle\bm{\widehat q} - \bm{\widetilde q}, \bm{\ell} \rangle
&=\sum_{(s,a,h)}
\sum_{(s',a',s'')}\sum_{m=1}^{h-1}\widetilde q(s',a',m)\left((\widehat P- \widetilde P)(s''\mid s',a',m)\right)\widehat q(s,a,h\mid s'',m+1)
\ell(s,a,h)\\
&=\sum_{m=1}^{H}\sum_{(s',a',s'')}\widetilde q(s',a',m)\left((\widehat P- \widetilde P)(s''\mid s',a',m)\right)\sum_{(s,a,h):h>m}\widehat q(s,a,h\mid s'',m+1)
\ell(s,a,h)\\
&=\sum_{m=1}^{H}\sum_{(s',a',s'')}\widetilde q(s',a',m)\left((\widehat P- \widetilde P)(s''\mid s',a',m)\right)J^{\widehat P, \pi,\ell}(s'',m+1)\\
&=\sum_{h=1}^{H}\sum_{(s',a',s'')}\widetilde q(s',a',h)\left(\widehat P(s''\mid s',a',h)-\widetilde P(s''\mid s',a',h)\right)J^{\widehat P, \pi,\ell}(s'',h+1).
\end{align*}
Since $\widehat P, P\in\cP_t$, \Cref{lemma:confidence'} implies that
\begin{align*}
\left|\langle\bm{\widehat q} - \bm{\widetilde q}, \bm{\ell} \rangle\right|
&\leq\sum_{h=1}^{H}\sum_{(s',a',s'')}\widetilde q(s',a',h)\left|\widehat P(s''\mid s',a',h)-\widetilde P(s''\mid s',a',h)\right|J^{\widehat P, \pi,\ell}(s'',h+1)\\
&\leq \sum_{h=1}^{H}\sum_{(s',a',s'')}\widetilde q(s',a',h)\epsilon_t^\star(s''\mid s',a',h)J^{\widehat P, \pi,\ell}(s'',h+1)\\
&\leq H\sum_{h=1}^{H}\sum_{(s',a',s'')}\widetilde q(s',a',h)\epsilon_t^\star(s''\mid s',a',h)\\
&=H\sum_{(s,a,s',h)\in\cS\times\cA\times\cS\times[H]}\widetilde q(s,a,h)\epsilon_t^\star(s'\mid s,a,h)
\end{align*}
where the third inequality holds because $J^{\widehat P, \pi,\ell}(s'',h+1)\leq H$, as required.
\end{proof}

\begin{lemma}\label{lemma:confidence'''}
Let $\pi$ be a policy, and let $\widetilde P,\widehat P$ be two different transition kernels. We denote by $\widetilde q$ the occupancy measure $q^{\widetilde P,\pi}$ associated with $\widetilde P$ and $\pi$, and we denote by $\widehat q$ the occupancy measure $q^{\widehat P,\pi}$ associated with $\widehat P$ and $\pi$. Let $(s,h)\in \cS\times [H]$, and consider $\widetilde q(\cdot \mid s,h), \widehat q(\cdot \mid s,h):\cS\times\cA\times\{h,\ldots, H\}$.
If $\widehat P, \widetilde P\in \cP_t$, then we have
$$\left|\langle\bm{\widehat q_{(s,h)}} - \bm{\widetilde q_{(s,h)}}, \bm{\ell_{(h)}} \rangle\right|\leq H\sum_{(s',a',s'',m)\in\cS\times \cA\times\cS\times\{h,\ldots,H\}}\widetilde q(s',a',m\mid s,h)\epsilon_t^\star (s''\mid s',a',m)$$
where $\bm{\widehat q_{(s,h)}},\bm{\widetilde q_{(s,h)}}, \bm{\ell_{(h)}}$ are the vector representations of $\widehat q(\cdot\mid s,h), \widetilde q(\cdot\mid s,h), \ell_{(h)}:\cS\times\cA\times\{h,\ldots, H\}\to\mathbb{R}.$
\end{lemma}
\begin{proof}
The proof follows the same argument used to rove Lemmas \ref{lemma7} and \ref{lemma:confidence''}.
\end{proof}

The following lemma is from \cite[Lemma 4]{ssp-adversarial-unknown} after some changes to adapt to our setting.

\begin{lemma}{\rm \citep[Lemma 4]{ssp-adversarial-unknown}}\label{lemma4}
Let $\pi_t$ be the policy for episode $t$, and let $q_t$ denote the occupancy measure $q^{P,\pi_t}$. Let $\ell:\cS\times \cA\times[H]\to[0,\infty)$ be an arbitrary reward function, and define $\mathbb{V}_t(s,a,h)=\var_{s'\sim P(\cdot\mid s,a,h)}\left[J^{P,\pi_t,\ell}(s',h+1)\right]$. Then
$$\langle \bm{q_t},\bm{\mathbb{V}_t}\rangle\leq \var\left[\langle \bm{n_t},\bm{\ell}\rangle\mid \ell,\pi_t, P\right]$$
where $\bm{q_t},\bm{\mathbb{V}_t}, \bm{n_t},\bm{\ell}$ are the vector representations of $q_t, \mathbb{V}_t,n_t,\ell:\cS\times\cA\times[H]\to\mathbb{R}.$
\end{lemma}
\begin{proof}
For ease of notation, let $s_h$ and $a_h$ denote $s_h^{P,\pi_t}$ and $a_h^{P,\pi_t}$, respectively for $h\in[H]$. Moreover, let $J(s,h)$ denote $J^{P,\pi_t,\ell}(s,h)$ for $(s,h)\in\cS\times[H]$.
Note that
\begin{align*}
\langle \bm{n_t},\bm{\ell}\rangle= \sum_{(s,a,h)\cS\times\cA\times[H]}\ell(s,a,h)n_t(s,a,h)=\sum_{h=1}^H \ell\left(s_h, a_h,h\right).
\end{align*}
For ease of notation, let $\mathbb{E}_t\left[\cdot\right]$ and $\var_t\left[\cdot\right]$ denote $\mathbb{E}\left[\cdot \mid \ell,\pi_t, P\right]$ and $\var\left[\cdot \mid \ell,\pi_t, P\right]$, respectively. 
Then
\begin{align*}
    \mathbb{E}_t\left[\langle \bm{n_t},\bm{\ell}\rangle\right] = \mathbb{E}_t\left[\sum_{h=1}^H \ell\left(s_h, a_h,h\right)\right]= \mathbb{E}_t\left[\mathbb{E}\left[\sum_{h=1}^H \ell\left(s_h, a_h,h\right)\mid \ell,\pi_t,P, s_1\right]\right]
     =\mathbb{E}_t\left[J(s_1,1)\right]
    \end{align*}
Moreover, 
\begin{align*}
\text{Var}_t\left[\langle \bm{n_t},\bm{\ell}\rangle\right]
&=\mathbb{E}_t\left[\left(\sum_{h=1}^H \ell\left(s_h, a_h,h\right)- \mathbb{E}_t\left[J(s_1,1)\right]\right)^2\right]\\
&=\mathbb{E}_t\left[\left(\sum_{h=1}^H \ell\left(s_h, a_h,h\right)- J(s_1,1)+J(s_1,1)-\mathbb{E}_t\left[J(s_1,1)\right] \right)^2\right]\\
&=\mathbb{E}_t\left[\left(\sum_{h=1}^H \ell\left(s_h, a_h,h\right)- J(s_1,1)\right)^2\right]+\mathbb{E}_t\left[\left(J(s_1,1)-\mathbb{E}_t\left[J(s_1,1)\right] \right)^2\right]\\
&\quad + 2\mathbb{E}_t\left[\left(\sum_{h=1}^H \ell\left(s_h, a_h,h\right)- J(s_1,1)\right)\left(J(s_1,1)-\mathbb{E}_t\left[J(s_1,1)\right] \right)\right]\\
&\geq\mathbb{E}_t\left[\left(\sum_{h=1}^H \ell\left(s_h, a_h,h\right)- J(s_1,1)\right)^2\right]
\end{align*}
where the inequality is by
$\mathbb{E}_t\left[J(s_1,1)-\mathbb{E}_t\left[J(s_1,1)\right]\mid s_1\right]=0$ and $\left(J(s_1,1)-\mathbb{E}_t\left[J(s_1,1)\right] \right)^2\geq 0$.
Therefore,
\begin{align*}
\text{Var}_t\left[\langle \bm{n_t},\bm{\ell}\rangle\right]
&\geq\mathbb{E}_t\left[\left(\sum_{h=2}^H \ell\left(s_h, a_h,h\right)- J(s_2,2)+\ell\left(s_1, a_1,1\right)+J(s_2,2)-J(s_1,1)\right)^2\right].
\end{align*}
Note that 
\begin{align}\label{CLlemma4:eq1}
\begin{aligned}
\mathbb{E}_t\left[\sum_{h=2}^H \ell\left(s_h, a_h,h\right)- J(s_2,2)\mid s_1,a_1,s_2\right]&=\mathbb{E}_t\left[\sum_{h=2}^H \ell\left(s_h, a_h,h\right)\mid s_2\right]- J(s_2,2)= 0.
\end{aligned}
\end{align}
Then 
\begin{align*}
\text{Var}_t\left[\langle \bm{n_t},\bm{\ell}\rangle\right]
&\geq\mathbb{E}_t\left[\left(\sum_{h=2}^H \ell\left(s_h, a_h,h\right)- J(s_2,2)\right)^2\right]+\mathbb{E}_t\left[\left(\ell\left(s_1,a_1,1\right)+J(s_2,2)-J(s_1,1)\right)^2\right]\\
&\quad + 2\mathbb{E}_t\left[\mathbb{E}_t\left[\left(\sum_{h=2}^H \ell\left(s_h, a_h,h\right)- J(s_2,2)\right)\left(\ell\left(s_1,a_1,1\right)+J(s_2,2)-J(s_1,1) \right)\mid s_1,a_1,s_2\right]\right]\\
&=\mathbb{E}_t\left[\left(\sum_{h=2}^H \ell\left(s_h, a_h,h\right)- J(s_2,2)\right)^2\right]+\mathbb{E}_t\left[\left(\ell\left(s_1,a_1,1\right)+J(s_2,2)-J(s_1,1) \right)^2\right]\\
&\quad + 2\mathbb{E}_t\left[\left(\ell\left(s_1,a_1,1\right)+J(s_2,2)-J(s_1,1) \right)\mathbb{E}_t\left[\sum_{h=2}^H \ell\left(s_h, a_h,h\right)- J(s_2,2)\mid s_1,a_1,s_2\right]\right]\\
&=\mathbb{E}_t\left[\left(\sum_{h=2}^H \ell\left(s_h, a_h,h\right)- J(s_2,2)\right)^2\right]+\mathbb{E}_t\left[\left(\ell\left(s_1,a_1,1\right)+J(s_2,2)-J(s_1,1)\right)^2\right]
\end{align*}
where the last equality follows from~\eqref{CLlemma4:eq1}.
Here, the second term from the right-most side can be bounded from below as follows.
\begin{align*}
&\mathbb{E}_t\left[\left(\ell\left(s_1,a_1,1\right)+J(s_2,2)-J(s_1,1) \right)^2\right]\\
&=\mathbb{E}_t\left[\left(\ell\left(s_1,a_1,1\right)+\sum_{s'\in \cS}P(s'\mid s_1, a_1,1)J(s',2)-J(s_1,1) +J(s_2,2)-\sum_{s'\in \cS}P(s'\mid s_1, a_1,1)J(s',2)\right)^2\right]\\
&=\mathbb{E}_t\left[\left(\ell\left(s_1,a_1,1\right)+\sum_{s'\in \cS}P(s'\mid s_1, a_1,1)J(s',2)-J(s_1,1) \right)^2\right]\\
&\quad + \mathbb{E}_t\left[\left(J(s_2,2)-\sum_{s'\in \cS}P(s'\mid s_1, a_1,1)J(s',2)\right)^2\right]\\
&\quad + 2\mathbb{E}_t\left[\left(\ell\left(s_1,a_1,1\right)+\sum_{s'\in \cS}P(s'\mid s_1, a_1,1)J(s',2)-J(s_1,1) \right)\left(J(s_2,2)-\sum_{s'\in \cS}P(s'\mid s_1, a_1,1)J(s',2)\right)\right]\\
&=\mathbb{E}_t\left[\left(\ell\left(s_1,a_1,1\right)+\sum_{s'\in \cS}P(s'\mid s_1, a_1,1)J(s',2)-J(s_1,1) \right)^2\right]\\
&\quad + \mathbb{E}_t\left[\left(J(s_2,2)-\sum_{s'\in \cS}P(s'\mid s_1, a_1,1)J(s',2)\right)^2\right]\\
&\geq \mathbb{E}_t\left[\mathbb{V}_t(s_1,a_1,1)\right]
\end{align*}
where third equality holds because
\begin{align*}
&\mathbb{E}_t\left[\left(\ell\left(s_1,a_1,1\right)+\sum_{s'\in \cS}P(s'\mid s_1, a_1,1)J(s',2)-J(s_1,1) \right)\left(J(s_2,2)-\sum_{s'\in \cS}P(s'\mid s_1, a_1,1)J(s',2)\right)\mid s_1,a_1\right]\\
&=\left(\ell\left(s_1,a_1,1\right)+\sum_{s'\in \cS}P(s'\mid s_1, a_1,1)J(s',2)-J(s_1,1)  \right)\mathbb{E}_t\left[J(s_2,2)-\sum_{s'\in \cS}P(s'\mid s_1, a_1,1)J(s',2)\mid s_1,a_1\right]\\
&=\left(\ell\left(s_1,a_1,1\right)+\sum_{s'\in \cS}P(s'\mid s_1, a_1,1)J(s',2)-J(s_1,1)  \right)\times 0
\end{align*}
and the last inequality holds because
\begin{align*}
\mathbb{E}_t\left[\left(J(s_2,2) -\sum_{s'\in \cS}P(s'\mid s_1, a_1,1)J(s',2)\right)^2\right]= \mathbb{E}_t\left[\mathbb{V}_t(s_1,a_1,1)\right].
\end{align*}
Then it follows that
\begin{align*}
\text{Var}_t\left[\langle \bm{n_t},\bm{\ell}\rangle\right]&\geq\mathbb{E}_t\left[\left(\sum_{h=1}^H \ell\left(s_h, a_h,h\right)- J(s_1,1)\right)^2\right]\geq \mathbb{E}_t\left[\left(\sum_{h=2}^H \ell\left(s_h, a_h,h\right)- J(s_2,2)\right)^2\right]+ \mathbb{E}_t\left[\mathbb{V}_t(s_1,a_1,1)\right].
\end{align*}
Repeating the same argument, we deduce that
\begin{align*}
\text{Var}_t\left[\langle \bm{n_t},\bm{\ell}\rangle\right]&\geq\sum_{h=1}^H \mathbb{E}_t\left[\mathbb{V}_t(s_h,a_h,h)\right]= \sum_{(s,a,h)\in\cS\times \cA\times [H]}q_t(s_h,a_h,h)\mathbb{V}_t(s_h,a_h,h)=\langle \bm{q_t},\bm{\mathbb{V}_t}\rangle,
\end{align*}
as required.
\end{proof}

Next, using~\Cref{bernstein2} that states the Bernstein-type concentration inequality for a martingale difference sequence, we prove the following lemma that is useful for our analysis. \Cref{lemma8} is a modification of \citep[Lemma 10]{Jin2020} and \citep[Lemma 8]{ssp-adversarial-unknown} to our finite-horizon MDP setting.

\begin{lemma}\label{lemma8}
With probability at least $1-2\delta$, we have
\begin{align}
\sum_{t=1}^T\sum_{(s,a,h)\in \cS\times \cA\times[H]}\frac{q_t(s,a,h)}{\max\left\{1,N_t(s,a,h)\right\}}&=O\left(SAH\ln T + H\ln\left(H/\delta\right)\right)\label{first-ineq}\\
\sum_{t=1}^T\sum_{(s,a,h)\in \cS\times \cA\times[H]}\frac{q_t(s,a,h)}{\sqrt{\max\left\{1,N_t(s,a,h)\right\}}}&=O\left(H\sqrt{SAT} + SAH\ln T+H\ln\left(H/\delta\right)\right)\label{second-ineq}
\end{align}
\end{lemma}
\begin{proof}
Note that
\begin{equation}\label{mds:bound}
\sum_{t=1}^T\sum_{(s,a)\in \cS\times \cA}\frac{q_t(s,a,h)}{\max\left\{1,N_t(s,a,h)\right\}}=\sum_{t=1}^T\sum_{(s,a)\in \cS\times \cA}\frac{n_t(s,a,h)}{\max\left\{1,N_t(s,a,h)\right\}}+\sum_{t=1}^TY_t
\end{equation}
where
$$Y_t=\sum_{(s,a)\in \cS\times \cA}\frac{-n_t(s,a,h)+q_t(s,a,h)}{\max\left\{1,N_t(s,a,h)\right\}}.$$
As $\mathbb{E}\left[n_t(s,a,h)\mid \pi_t,P\right] = q_t(s,a,h)$ holds for every $(s,a,h)\in\cS\times \cA\times [H]$, we know that $Y_1,\ldots, Y_T$ is a martingale difference sequence. We know that $Y_t\leq 1$ for each $t\in[T]$. Let $\mathbb{E}_t\left[\cdot\right]$ denote $\mathbb{E}\left[\cdot\mid  \cF_t,P\right]$. Since $\pi_t$ is $\cF_t$-measurable, we have $\mathbb{E}_t\left[n_t(s,a,h)\right]=q_t(s,a,h)$.
Then we deduce 
\begin{align*}
\mathbb{E}_t\left[ Y_t^2\right]&=\sum_{(s,a),(s',a')\in\cS\times \cA}\frac{\mathbb{E}_t\left[(n_t(s,a,h)-q_t(s,a,h))(n_t(s',a',h)-q_t(s',a',h))\right]}{\max\left\{1,N_t(s,a,h)\right\}\cdot \max\left\{1,N_t(s',a',h)\right\}}\\
&=\sum_{(s,a),(s',a')\in\cS\times \cA}\frac{\mathbb{E}_t\left[n_t(s,a,h)n_t(s',a',h)-q_t(s,a,h)q_t(s',a',h)\right]}{\max\left\{1,N_t(s,a,h)\right\}\cdot \max\left\{1,N_t(s',a',h)\right\}}\\
&\leq\sum_{(s,a),(s',a')\in\cS\times \cA}\frac{\mathbb{E}_t\left[n_t(s,a,h)n_t(s',a',h)\right]}{\max\left\{1,N_t(s,a,h)\right\}\cdot \max\left\{1,N_t(s',a',h)\right\}}\\
&\leq\sum_{(s,a)\in\cS\times \cA}\frac{\mathbb{E}_t\left[n_t(s,a,h)\right]}{\max\left\{1,N_t(s,a,h)\right\}}\\
&=\sum_{(s,a)\in\cS\times \cA}\frac{q_t(s,a,h)}{\max\left\{1,N_t(s,a,h)\right\}}
\end{align*}
where the second equality holds because it follows from $\mathbb{E}_t\left[n_t(s,a,h)\right] = q_t(s,a,h)$ for $(s,a,h)\in\cS\times\cA\times[H]$ that $$\mathbb{E}_t\left[q_t(s,a,h)n_t(s',a',h)\right]=\mathbb{E}_t\left[q_t(s',a',h)n_t(s,a,h)\right]=q_t(s,a,h)q_t(s',a',h),$$
the second inequality holds because $n_t(s,a,h)n_t(s',a',h)=0$ if $(s,a)\neq (s',a')$, and the last equality holds true because $\mathbb{E}_t\left[n_t(s,a,h)\right] = q_t(s,a,h)$ for any $(s,a,h)\in\cS\times\cA\times[H]$. Then we may apply \Cref{bernstein2} with $\lambda=1/2$, and we deduce that with probability at least $1-\delta/H$,
$$\sum_{t=1}^TY_t\leq \frac{1}{2}\sum_{t=1}^T\sum_{(s,a)\in\cS\times \cA}\frac{q_t(s,a,h)}{\max\left\{1,N_t(s,a,h)\right\}} + 2\ln(H/\delta).$$
Plugging this inequality to~\eqref{mds:bound}, it follows that
$$\sum_{t=1}^T\sum_{(s,a)\in \cS\times \cA}\frac{q_t(s,a,h)}{\max\left\{1,N_t(s,a,h)\right\}}=2\sum_{t=1}^T\sum_{(s,a)\in \cS\times \cA}\frac{n_t(s,a,h)}{\max\left\{1,N_t(s,a,h)\right\}}+4\ln(H/\delta).$$
Here, the first term on the right-hand side can be bounded as follows. We have
\begin{align*}
    \sum_{t=1}^T\frac{n_t(s,a,h)}{\max\left\{1,N_t(s,a,h)\right\}}
    &=\sum_{t=1}^T\frac{n_t(s,a,h)}{\max\left\{1,N_{t+1}(s,a,h)\right\}} +\sum_{t=1}^T\left(\frac{n_t(s,a,h)}{\max\left\{1,N_{t}(s,a,h)\right\}}-\frac{n_t(s,a,h)}{\max\left\{1,N_{t+1}(s,a,h)\right\}}\right)\\
    &\leq\sum_{t=1}^T\frac{n_t(s,a,h)}{\max\left\{1,N_{t+1}(s,a,h)\right\}} +\sum_{t=1}^T\left(\frac{1}{\max\left\{1,N_{t}(s,a,h)\right\}}-\frac{1}{\max\left\{1,N_{t+1}(s,a,h)\right\}}\right)\\
    &\leq\sum_{t=1}^T\frac{n_t(s,a,h)}{\max\left\{1,N_{t+1}(s,a,h)\right\}} +1\\
    &=O(\ln T).
    \end{align*}
    where the first inequality is due to $n_t(s,a,h)\leq 1$ and  the last inequality holds because
    $$n_t(s,a,h)= N_{t+1}(s,a,h) - N_t(s,a,h)\quad\text{and}\quad N_T(s,a,h)+n_T(s,a,h)\leq T.$$
Therefore, it follows that
\begin{align*}
\sum_{t=1}^T\sum_{(s,a)\in \cS\times \cA}\frac{n_t(s,a,h)}{\max\left\{1,N_t(s,a,h)\right\}}=\sum_{(s,a)\in \cS\times \cA}\sum_{t=1}^T\frac{n_t(s,a,h)}{\max\left\{1,N_t(s,a,h)\right\}}=O(SA\ln T).
\end{align*}
As a result, for any fixed $h\in [H]$,
$$\sum_{t=1}^T\sum_{(s,a)\in \cS\times \cA}\frac{q_t(s,a,h)}{\max\left\{1,N_t(s,a,h)\right\}}=O\left(SA\ln T + \ln\left(H/\delta\right)\right)$$
holds with probability at least $1-\delta/H$. By union bound,~\eqref{first-ineq} holds with probability at least $1-\delta$. 

Next, we will show that~\eqref{second-ineq} holds.
\begin{equation}\label{mds:bound'}
\sum_{t=1}^T\sum_{(s,a)\in \cS\times \cA}\frac{q_t(s,a,h)}{\sqrt{\max\left\{1,N_t(s,a,h)\right\}}}=\sum_{t=1}^T\sum_{(s,a)\in \cS\times \cA}\frac{n_t(s,a,h)}{\sqrt{\max\left\{1,N_t(s,a,h)\right\}}}+\sum_{t=1}^TZ_t
\end{equation}
where
$$Z_t=\sum_{(s,a)\in \cS\times \cA}\frac{-n_t(s,a,h)+q_t(s,a,h)}{\sqrt{\max\left\{1,N_t(s,a,h)\right\}}}.$$
As~$\mathbb{E}_t\left[n_t(s,a,h)\right] = q_t(s,a,h)$ holds for every $(s,a,h)\in\cS\times \cA\times [H]$, we know that $Z_1,\ldots, Z_T$ is a martingale difference sequence. We know that $Z_t\leq 1$ for each $t\in[T]$. Then we deduce 
\begin{align*}
\mathbb{E}_t\left[ Z_t^2\right]&\leq\sum_{(s,a),(s',a')\in\cS\times \cA}\frac{\mathbb{E}_t\left[n_t(s,a,h)n_t(s',a',h)\right]}{\sqrt{\max\left\{1,N_t(s,a,h)\right\}}\cdot \sqrt{\max\left\{1,N_t(s',a',h)\right\}}}\\
&=\sum_{(s,a)\in\cS\times \cA}\frac{\mathbb{E}_t\left[n_t(s,a,h)\right]}{\max\left\{1,N_t(s,a,h)\right\}}\\
&=\sum_{(s,a)\in\cS\times \cA}\frac{q_t(s,a,h)}{\max\left\{1,N_t(s,a,h)\right\}}
\end{align*}
where the first inequality is derived by the same argument when bounding $\mathbb{E}_t[Y_t^2]$,
the first equality holds because $n_t(s,a,h)n_t(s',a',h)=0$ if $(s,a)\neq (s',a')$, and the last equality holds true because $\mathbb{E}_t\left[n_t(s,a,h)\right] = q_t(s,a,h)$ for any $(s,a,h)\in\cS\times\cA\times[H]$. Then we may apply \Cref{bernstein2} with $\lambda=1$, and we deduce that with probability at least $1-\delta/H$,
$$\sum_{t=1}^TZ_t\leq \sum_{t=1}^T\sum_{(s,a)\in\cS\times \cA}\frac{q_t(s,a,h)}{\max\left\{1,N_t(s,a,h)\right\}} + \ln(H/\delta).$$
Then with probability at least $1-2\delta$,~\eqref{first-ineq} holds and 
\begin{align}\label{mds:bound''}
\begin{aligned}
\sum_{h\in[H]}\sum_{t=1}^TZ_t&\leq \sum_{t=1}^T\sum_{(s,a,h)\in\cS\times \cA\times[H]}\frac{q_t(s,a,h)}{\max\left\{1,N_t(s,a,h)\right\}} + H\ln(H/\delta)=O(SAH\ln T + H\ln(H/\delta)).
\end{aligned}
\end{align}
holds.
Moreover, we have
\begin{align*}
    &\sum_{t=1}^T\frac{n_t(s,a,h)}{\sqrt{\max\left\{1,N_t(s,a,h)\right\}}}\\
    &=\sum_{t=1}^T\frac{n_t(s,a,h)}{\sqrt{\max\left\{1,N_{t+1}(s,a,h)\right\}}} +\sum_{t=1}^T\left(\frac{n_t(s,a,h)}{\sqrt{\max\left\{1,N_{t}(s,a,h)\right\}}}-\frac{n_t(s,a,h)}{\sqrt{\max\left\{1,N_{t+1}(s,a,h)\right\}}}\right)\\
    &\leq\sum_{t=1}^T\frac{n_t(s,a,h)}{\sqrt{\max\left\{1,N_{t+1}(s,a,h)\right\}}} +\sum_{t=1}^T\left(\frac{1}{\sqrt{\max\left\{1,N_{t}(s,a,h)\right\}}}-\frac{1}{\sqrt{\max\left\{1,N_{t+1}(s,a,h)\right\}}}\right)\\
    &\leq\sum_{t=1}^T\frac{n_t(s,a,h)}{\sqrt{\max\left\{1,N_{t+1}(s,a,h)\right\}}} +1\\
    &=O(\sqrt{N_{T+1}(s,a,h)} + 1).
    \end{align*}
    where the last equality holds because
    $n_t(s,a,h)= N_{t+1}(s,a,h) - N_t(s,a,h)$.
    Then
    \begin{align*}
        \sum_{t=1}^T\sum_{(s,a,h)\in\cS\times\cA\times[H]}\frac{n_t(s,a,h)}{\sqrt{\max\left\{1,N_t(s,a,h)\right\}}}&=O\left(\sum_{\cS\times\cA\times [H]}\left(\sqrt{N_{T+1}(s,a,h)} +1 \right)\right)\\
        &=O\left(\sqrt{SAH\sum_{\cS\times\cA\times [H]}N_{T+1}(s,a,h)} + SAH\right)\\
        &=O\left(H\sqrt{SAT} + SAH\right)
    \end{align*}
    where the second equality is due to the Cauchy-Schwarz inequality. Then it follows from~\eqref{mds:bound'} and~\eqref{mds:bound''} that~\eqref{second-ineq} holds. 
\end{proof}

\begin{lemma}\label{lemma10}
Assume that $P\in\cP_t$ for every episode $t\in[T]$. Then
\begin{align*}
&\sum_{t=1}^T\left|\sum_{(s,a,s',h)}q_t(s,a,h)\left(\left(P-P_t\right)(s'\mid s,a,h)\right)\left(\left(J^{P_t,\pi_t,\ell}-J^{P,\pi_t,\ell}\right)(s',h+1)\right)\right| = O\left(H^2S^2A\left(\ln(HSAT/\delta)\right)^2 \right)
\end{align*}
for any $P_t\in \cP_t$ where $\left(P-P_t\right)(s'\mid s,a,h)=P(s'\mid s,a,h)-P_t(s'\mid s,a,h)$ and $\left(J^{P_t,\pi_t,\ell}-J^{P,\pi_t,\ell}\right)(s',h+1)= J^{P_t,\pi_t,\ell}(s',h+1)-J^{P,\pi_t,\ell}(s',h+1)$.
\end{lemma}
\begin{proof}
Let $\bm{q^{P_t,\pi_t}_{(s',h+1)}},\bm{q^{P,\pi_t}_{(s',h+1)}},\bm{\ell}$ be the vector representations of $q^{P_t,\pi_t}(\cdot \mid s',h+1),q^{P,\pi_t}(\cdot \mid s',h+1),\ell:\cS\times\cA\times\cS\to[0,1]$, respectively. Note that 
\begin{align*}
&\sum_{t=1}^T\left|\sum_{(s,a,s',h)}q_t(s,a,h)\left(\left(P-P_t\right)(s'\mid s,a,h)\right)\left(\left(J^{P_t,\pi_t,\ell}-J^{P,\pi_t,\ell}\right)(s',h+1)\right)\right|\\
&\leq \sum_{t=1}^T\sum_{(s,a,s',h)}q_t(s,a,h)\epsilon_t^\star(s'\mid s,a,h)\left|\left(J^{P_t,\pi_t,\ell}-J^{P,\pi_t,\ell}\right)(s',h+1)\right|\\
&= \sum_{t=1}^T\sum_{(s,a,s',h)}q_t(s,a,h)\epsilon_t^\star(s'\mid s,a,h)\left|\langle\bm{q^{P_t,\pi_t}_{(s',h+1)}}-\bm{q^{P,\pi_t}_{(s',h+1)}}, \bm{\ell}\rangle\right|\\
&\leq H\sum_{t=1}^T\sum_{(s,a,s',h)}q_t(s,a,h)\epsilon_t^\star(s'\mid s,a,h) \sum_{(s'',a'',s'''), m\geq h}q_t(s'',a'',m\mid s',h+1)\epsilon_t^\star(s'''\mid s'',a'',m)\\
\end{align*}
where the first inequality is from~\Cref{lemma:confidence'}, the first equality holds because $J^{P_t,\pi_t,\ell}(s',h+1)=\langle\bm{q^{P_t,\pi_t}_{(s',h+1)}}, \bm{\ell}\rangle$ and $J^{P,\pi_t,\ell}(s',h+1)=\langle\bm{q^{P,\pi_t}_{(s',h+1)}}, \bm{\ell}\rangle$, the second inequality is due to \Cref{lemma:confidence'''}. Then plugging in the definition of $\epsilon_t^*$, it follows that
\begin{align*}
&\frac{1}{H\ln(HSAT/\delta)}\sum_{t=1}^T\left|\sum_{(s,a,s',h)}q_t(s,a,h)\left(\left(P-P_t\right)(s'\mid s,a,h)\right)\left(\left(J^{P_t,\pi_t,\ell}-J^{P,\pi_t,\ell}\right)(s',h+1)\right)\right|\\
&=O\left(\sum_{\substack{t,(s,a,s',h)\\(s'',a'',s''')\\ m\geq h+1}}q_t(s,a,h)\sqrt{\frac{P(s'\mid s,a,h)}{\max\{1,n_t(s,a,h)\}}} q_t(s'',a'',m\mid s',h+1)\sqrt{\frac{P(s'''\mid s'',a'',m)}{\max\{1,n_t(s'',a'',m)\}}}\right)\\
&= O\left(\sum_{\substack{t,(s,a,s',h)\\(s'',a'',s''')\\ m\geq h+1}}\sqrt{{\scriptstyle\frac{q_t(s,a,h)P(s'''\mid s'',a'',m) q_t(s'',a'',m\mid s',h+1)}{\max\{1,n_t(s,a,h)\}}}} \sqrt{{\scriptstyle\frac{q_t(s,a,h) P(s'\mid s,a,h)q_t(s'',a'',m\mid s',h+1)}{\max\{1,n_t(s'',a'',m)\}}}}\right)\\
&= O\left(\sqrt{\sum_{\substack{t,(s,a,s',h)\\(s'',a'',s''')\\ m\geq h+1}}{\scriptstyle\frac{q_t(s,a,h)P(s'''\mid s'',a'',m) q_t(s'',a'',m\mid s',h+1)}{\max\{1,n_t(s,a,h)\}}}} \sqrt{\sum_{\substack{t,(s,a,s',h)\\(s'',a'',s''')\\ m\geq h+1}}{\scriptstyle\frac{q_t(s,a,h)P(s'\mid s,a,h)q_t(s'',a'',m\mid s',h+1)}{\max\{1,n_t(s'',a'',m)\}}}}\right)\\
&= O\left(\sqrt{S\sum_{t=1}^T\sum_{(s,a,h)}{\frac{q_t(s,a,h)}{\max\{1,n_t(s,a,h)\}}}} \sqrt{S\sum_{t=1}^T\sum_{(s'',a'',m)}{\frac{q_t(s'',a'',m)}{\max\{1,n_t(s'',a'',m)\}}}}\right)\\
&= O\left(S^2AH\ln T + SH\ln\left(H/\delta\right)\right)
\end{align*}
where the third equality follows from the Cauchy-Schwarz inequality and the last equality is due to \Cref{lemma8}. Therefore, we deduce that
\begin{align*}
&\sum_{t=1}^T\left|\sum_{(s,a,s',h)}q_t(s,a,h)\left(\left(P-P_t\right)(s'\mid s,a,h)\right)\left(\left(J^{P_t,\pi_t,\ell}-J^{P,\pi_t,\ell}\right)(s',h+1)\right)\right| \\
&= O\left(H^2S^2A\ln T\ln(HSAT/\delta) + SH^2\ln\left(H/\delta\right)\ln(HSAT/\delta)\right)\\
&= O\left(H^2S^2A\left(\ln(HSAT/\delta)\right)^2 \right),
\end{align*}
as required.
\end{proof}

Next, we provide \Cref{lemma9}, which is a modification of \citep[Lemma 9]{ssp-adversarial-unknown} to our finite-horizon MDP setting.

\begin{lemma}\label{lemma9}
	Let $\pi_t$ be any policy for episode $t$, and let $P_t$ be any transition kernel from $\cP_t$. Let $q_t,\widehat q_t$ denote the occupancy measures $q^{P,\pi_t},q^{P_t,\pi_t}$, respectively. Let $\ell_t:\cS\times \cA\times[H]\to[0,1]$ be an arbitrary reward function for episode $t\in[T]$. Then with probability at least $1-4 \delta$, 
	\begin{align*}
		\sum_{t=1}^T \left|\langle  \bm{\ell_t}, \bm{q_t}-\bm{\widehat q_t}\rangle\right|=O\left(\left(\sqrt{HS^2A\left(\sum_{t=1}^T\langle \bm{q_t},\vec h\odot \bm{\ell_t}\rangle+H^3\sqrt{T}\right)}+H^2S^2A\right)\left(\ln\frac{HSAT}{\delta}\right)^2 \right).
	\end{align*}
\end{lemma}
\begin{proof}
We define $\xi_1$ as 
$\xi_1=\left\{\ell_1,\pi_1\right\}$
and for $t\geq 2$, we define $\xi_t$ as 
$$\left\{s_1^{P,\pi_{t-1}},a_1^{P,\pi_{t-1}},\ldots, s_h^{P,\pi_{t-1}}, a_h^{P,\pi_{t-1}}, \ell_t,\pi_t\right\}$$
where $\pi_{t-1}$ and $\pi_t$ denote the policies for episode $t-1$ and episode $t$, respectively, and $$\left(s_1^{P,\pi_{t-1}},a_1^{P,\pi_{t-1}},\ldots, s_h^{P,\pi_{t-1}}, a_h^{P,\pi_{t-1}}\right)$$
is the trajectory generated under policy $\pi_{t-1}$ and transition kernel $P$. Then for $t\in[T]$, let $\cH_t$ be defined as the $\sigma$-algebra generated by the random variables in $\xi_1\cup\cdots\cup \xi_{t}$. Then it follows that $\cH_1,\ldots, \cH_T$ give rise to a filtration.

Let us define $$\mu_t(s,a,h)=\mathbb{E}_{s'\sim P(\cdot\mid s,a,h)}\left[J^{P,\pi_t,\ell_t}(s',h+1)\right].$$

Note that
\begin{align*}
\sum_{t=1}^T \left|\langle \bm{q_t}-\bm{\widehat q_t}, \bm{\ell_t}\rangle\right|
&=\sum_{t=1}^T\left|\sum_{(s,a,s',h)\in\cS\times\cA\times\cS\times[H]}q_t(s,a,h)\left(P(s'\mid s,a,h)-P_t(s'\mid s,a,h)\right)J^{P_t, \pi_t,\ell_t}(s',h+1)\right|\\
&\leq \sum_{t=1}^T\left|\sum_{(s,a,s',h)\in\cS\times\cA\times\cS\times[H]}q_t(s,a,h)\left(P(s'\mid s,a,h)-P_t(s'\mid s,a,h)\right)J^{P, \pi_t,\ell_t}(s',h+1)\right|\\
&\quad + O\left(H^2S^2A\left(\ln(HSAT/\delta)\right)^2 \right)
\end{align*}
where the first equality is due to \Cref{lemma7} and the first inequality is by~\Cref{lemma10}.
Moreover,
\begin{align*}
&\sum_{t=1}^T\left|\sum_{(s,a,s',h)\in\cS\times\cA\times\cS\times[H]}q_t(s,a,h)\left(P(s'\mid s,a,h)-P_t(s'\mid s,a,h)\right)J^{P, \pi_t,\ell_t}(s',h+1)\right|\\
&=\sum_{t=1}^T\left|\sum_{(s,a,s',h)\in\cS\times\cA\times\cS\times[H]}q_t(s,a,h)\left(\left(P-P_t\right)(s'\mid s,a,h)\right)\left(J^{P, \pi_t,\ell_t}(s',h+1)-\mu_t(s,a,h)\right)\right|\\
&\leq \sum_{t=1}^T\sum_{(s,a,s',h)\in\cS\times\cA\times\cS\times[H]}q_t(s,a,h)\epsilon_t^\star(s'\mid s,a,h)\left|J^{P, \pi_t,\ell_t}(s',h+1)-\mu_t(s,a,h)\right|\\
&\leq O\left(\sum_{t=1}^T\sum_{\substack{(s,a,s',h)\in\\ \cS\times\cA\times\cS\times[H]}}q_t(s,a,h)\sqrt{\frac{P(s'\mid s,a,h)\ln(HSAT/\delta)}{\max\{1,N_t(s,a,h)\}}\left(J^{P, \pi_t,\ell_t}(s',h+1)-\mu_t(s,a,h)\right)^2}\right)\\
&\quad + O\left(HS\sum_{t=1}^T\sum_{(s,a,h)\in\cS\times\cA\times[H]}\frac{q_t(s,a,h)\ln(HSAT/\delta)}{\max\{1,N_t(s,a,h)\}}\right)\\
&\leq O\left(\sum_{t=1}^T\sum_{\substack{(s,a,s',h)\in\\ \cS\times\cA\times\cS\times[H]}}q_t(s,a,h)\sqrt{\frac{P(s'\mid s,a,h)\ln(HSAT/\delta)}{\max\{1,N_t(s,a,h)\}}\left(J^{P, \pi_t,\ell_t}(s',h+1)-\mu_t(s,a,h)\right)^2}\right)\\
&\quad +  O\left(H^2S^2A\left(\ln(HSAT/\delta)\right)^2 \right)
\end{align*}
where $\left(P-P_t\right)(s'\mid s,a,h)=P(s'\mid s,a,h)-P_t(s'\mid s,a,h)$, the first equality holds because $\sum_{s'\in\cS}\left(P-P_t\right)(s'\mid s,a,h)=0$ and $\mu_t(s,a,h)$ is independent of $s'$, the first inequality is due to~\Cref{lemma:confidence'}, and the second inequality is from~\Cref{lemma:confidence'} and $\left|J^{P, \pi_t,\ell_t}(s',h+1)-\mu_t(s,a,h)\right|\leq H$. Recall that
$q_t(s,a,h)=\mathbb{E}\left[n_t(s,a,h)\mid  \pi_t, P\right]$, which implies that
\begin{align*}
&\sum_{t=1}^T\sum_{\substack{(s,a,s',h)\in\\ \cS\times\cA\times\cS\times[H]}}q_t(s,a,h)\sqrt{\frac{P(s'\mid s,a,h)\ln(HSAT/\delta)}{\max\{1,N_t(s,a,h)\}}\left(J^{P, \pi_t,\ell_t}(s',h+1)-\mu_t(s,a,h)\right)^2}=\sum_{t=1}^T \mathbb{E}\left[X_t\mid \cH_t, P\right]
\end{align*}
where 
$$X_t= \sum_{\substack{(s,a,s',h)\in\\ \cS\times\cA\times\cS\times[H]}}n_t(s,a,h)\sqrt{\frac{P(s'\mid s,a,h)\ln(HSAT/\delta)}{\max\{1,N_t(s,a,h)\}}\left(J^{P, \pi_t,\ell_t}(s',h+1)-\mu_t(s,a,h)\right)^2}.$$
Here, we have
$$0\leq X_t \leq O\left(HS\sum_{(s,a,h)\in \cS\times \cA\times [H]}n_t(s,a,h) \sqrt{\ln(HSAT/\delta)}\right)= O(H^2S\sqrt{\ln(HSAT/\delta)}).$$
Then it follows from \Cref{cohen-concentration} that with probability at least $1-\delta$, 
\begin{align*}
&\sum_{t=1}^T \mathbb{E}\left[X_t\mid \cH_t,  P\right]\\
&\leq 2\sum_{t=1}^T \sum_{\substack{(s,a,s',h)\in\\ \cS\times\cA\times\cS\times[H]}}n_t(s,a,h)\sqrt{\frac{P(s'\mid s,a,h)\ln(HSAT/\delta)}{\max\{1,N_t(s,a,h)\}}\left(J^{P, \pi_t,\ell_t}(s',h+1)-\mu_t(s,a,h)\right)^2}\\
&\quad + O\left(H^2S \left(\ln(HSAT/\delta)\right)^{3/2}\right).
\end{align*}
Note that
\begin{align*}
&\sum_{t=1}^T \sum_{\substack{(s,a,s',h)\in\\ \cS\times\cA\times\cS\times[H]}}n_t(s,a,h)\sqrt{\frac{P(s'\mid s,a,h)\ln(HSAT/\delta)}{\max\{1,N_t(s,a,h)\}}\left(J^{P, \pi_t,\ell_t}(s',h+1)-\mu_t(s,a,h)\right)^2}\\
&\leq \sum_{t=1}^T \sum_{\substack{(s,a,s',h)\in\\ \cS\times\cA\times\cS\times[H]}}n_t(s,a,h)\sqrt{\frac{P(s'\mid s,a,h)\ln(HSAT/\delta)}{\max\{1,N_{t+1}(s,a,h)\}}\left(J^{P, \pi_t,\ell_t}(s',h+1)-\mu_t(s,a,h)\right)^2}\\
&\quad + H\sum_{t=1}^T \sum_{\substack{(s,a,s',h)\in\\ \cS\times\cA\times\cS\times[H]}}n_t(s,a,h)\left(\sqrt{\frac{P(s'\mid s,a,h)\ln(HSAT/\delta)}{\max\{1,N_{t}(s,a,h)\}}}-\sqrt{\frac{P(s'\mid s,a,h)\ln(HSAT/\delta)}{\max\{1,N_{t+1}(s,a,h)\}}}\right)\\
&\leq \sum_{t=1}^T \sum_{\substack{(s,a,s',h)\in\\ \cS\times\cA\times\cS\times[H]}}n_t(s,a,h)\sqrt{\frac{P(s'\mid s,a,h)\ln(HSAT/\delta)}{\max\{1,N_{t+1}(s,a,h)\}}\left(J^{P, \pi_t,\ell_t}(s',h+1)-\mu_t(s,a,h)\right)^2}\\
&\quad + H\sqrt{S}\sum_{t=1}^T \sum_{{(s,a,h)\in \cS\times\cA\times[H]}}\left(\sqrt{\frac{\ln(HSAT/\delta)}{\max\{1,N_{t}(s,a,h)\}}}-\sqrt{\frac{\ln(HSAT/\delta)}{\max\{1,N_{t+1}(s,a,h)\}}}\right)\\
&\leq \sum_{t=1}^T \sum_{\substack{(s,a,s',h)\in\\ \cS\times\cA\times\cS\times[H]}}n_t(s,a,h)\sqrt{\frac{P(s'\mid s,a,h)\ln(HSAT/\delta)}{\max\{1,N_{t+1}(s,a,h)\}}\left(J^{P, \pi_t,\ell_t}(s',h+1)-\mu_t(s,a,h)\right)^2}\\
&\quad +  O\left(H^2S^{3/2}A \sqrt{\ln(HSAT/\delta)}\right).
\end{align*}
where the first inequality holds because $\left|J^{P, \pi_t,\ell_t}(s',h+1)-\mu_t(s,a,h)\right|\leq H$, the second inequality holds because $n_t(s,a,h)\leq 1$ and the Cauchy-Schwarz inequality implies that
$$\sum_{s'\in \cS}\sqrt{P(s'\mid s,a,h)}\leq \sqrt{S\sum_{s'\in \cS}P(s'\mid s,a,h)}=\sqrt{S},$$
and the third inequality follows from 
$$\sum_{t=1}^T \left(\sqrt{\frac{1}{\max\{1,N_{t}(s,a,h)\}}}-\sqrt{\frac{1}{\max\{1,N_{t+1}(s,a,h)\}}}\right)\leq \sqrt{\frac{1}{\max\{1,N_{1}(s,a,h)\}}}=1.$$
Next, the Cauchy-Schwarz inequality implies the following. 
\begin{align*}
& \sum_{t=1}^T \sum_{\substack{(s,a,s',h)\in\\ \cS\times\cA\times\cS\times[H]}}n_t(s,a,h)\sqrt{\frac{P(s'\mid s,a,h)\ln(HSAT/\delta)}{\max\{1,N_{t+1}(s,a,h)\}}\left(J^{P, \pi_t,\ell_t}(s',h+1)-\mu_t(s,a,h)\right)^2}\\
&\leq  \sqrt{\sum_{t=1}^T \sum_{\substack{(s,a,s',h)\in\\ \cS\times\cA\times\cS\times[H]}}n_t(s,a,h){P(s'\mid s,a,h)}\left(J^{P, \pi_t,\ell_t}(s',h+1)-\mu_t(s,a,h)\right)^2}\\
&\quad \times\sqrt{ \sum_{t=1}^T \sum_{\substack{(s,a,s',h)\in\\ \cS\times\cA\times\cS\times[H]}}n_t(s,a,h)\frac{\ln(HSAT/\delta)}{\max\{1,N_{t+1}(s,a,h)\}}}
\end{align*}
Here, the second term can be bounded as follows.
\begin{align*}
&\sum_{t=1}^T \sum_{\substack{(s,a,s',h)\in \cS\times\cA\times\cS\times[H]}}n_t(s,a,h)\frac{\ln(HSAT/\delta)}{\max\{1,N_{t+1}(s,a,h)\}}\\
&=S\ln\left(\frac{HSAT}{\delta}\right)\sum_{t=1}^T \sum_{\substack{(s,a,h)\in \cS\times\cA\times[H]}}\frac{n_t(s,a,h)}{\max\{1,N_{t+1}(s,a,h)\}}\\
&=S\ln\left(\frac{HSAT}{\delta}\right) \sum_{\substack{(s,a,h)\in \cS\times\cA\times[H]}}\sum_{t=1}^T\frac{n_t(s,a,h)}{\max\{1,N_{t+1}(s,a,h)\}}\\
&=O\left(HS^2 A\left(\ln\left({HSAT}/{\delta}\right)\right)^2\right).
\end{align*}
$(s,a,h)\in \cS\times \cA\times [H]$, we define 
$$\mathbb{V}_t(s,a,h)=\var_{s'\sim P(\cdot\mid s,a,h)}\left[J^{P,\pi_t,\ell_t}(s',h+1)\right].$$
Then 
\begin{align*}
\mathbb{V}_t(s,a,h)&= \mathbb{E}_{s'\sim P(\cdot\mid s,a,h)}\left[\left(J^{P,\pi_t,\ell_t}(s',h+1)-\mu_t(s,a,h)\right)^2\right]\\
&= \sum_{s'\in\cS} P(s'\mid s,a,h)\left(J^{P,\pi_t,\ell_t}(s',h+1)-\mu_t(s,a,h)\right)^2
\end{align*}
Furthermore, 
\begin{align*}
&\sum_{t=1}^T \sum_{\substack{(s,a,s',h)\in\\ \cS\times\cA\times\cS\times[H]}}n_t(s,a,h){P(s'\mid s,a,h)}\left(J^{P, \pi_t,\ell_t}(s',h+1)-\mu_t(s,a,h)\right)^2\\
&=\sum_{t=1}^T \sum_{\substack{(s,a,h)\in \cS\times\cA\times[H]}}n_t(s,a,h) \mathbb{V}_t(s,a,h)\\
&=\sum_{t=1}^T \langle\bm{q_t},\bm{\mathbb{V}_t}\rangle+\sum_{t=1}^T \sum_{\substack{(s,a,h)\in \cS\times\cA\times[H]}}(n_t(s,a,h)-q_t(s,a,h)) \mathbb{V}_t(s,a,h)\\
&\leq \sum_{t=1}^T\var\left[\langle n_t,\ell_t\rangle\mid \ell_t,\pi_t,P\right]+O\left(H^3\sqrt{T\ln(1/\delta)}\right)
\end{align*}
where $\bm{\mathbb{V}_t}\in\mathbb{R}^{SAH}$ is the vector representation of $\mathbb{V}_t$ and the  inequality follows from \Cref{lemma4}, $ \mathbb{V}_t(s,a,h)\leq H^2$,
\begin{align*}\sum_{\substack{(s,a,h)\in \cS\times\cA\times[H]}}(n_t(s,a,h)-q_t(s,a,h)) \mathbb{V}_t(s,a,h)&\leq \sum_{\substack{(s,a,h)\in \cS\times\cA\times[H]}}(n_t(s,a,h)+q_t(s,a,h)) H^2\leq 2H^3,
\end{align*}
and \Cref{azuma}.
Therefore, we finally have proved that
\begin{align*}
\sum_{t=1}^T \left|\langle \bm{q_t}-\bm{\widehat q_t}, \bm{\ell_t}\rangle\right|
&=O\left(\sqrt{HS^2A\left(\ln\frac{HSAT}{\delta}\right)^2\left(\sum_{t=1}^T\var\left[\langle n_t,\ell_t\rangle\mid \ell_t,\pi_t,P\right]+H^3\sqrt{T\ln\frac{1}{\delta}}\right)}\right)\\
&\quad + O\left(H^2S^2A\left(\ln\frac{HSAT}{\delta}\right)^2 \right).
\end{align*}
Moreover, we know from~\Cref{lemma2} that
$$\var\left[\langle \bm{n_t}, \bm{\ell_t}\rangle^2\mid  \ell_t, \pi_t, P\right]\leq \mathbb{E}\left[\langle \bm{n_t}, \bm{\ell_t}\rangle^2\mid \ell_t, \pi_t, P\right]\leq 2 \langle \bm{q_t},\vec h\odot \bm{\ell_t}\rangle,$$
and therefore, it follows that
\begin{align*}
\sum_{t=1}^T \left|\langle \bm{q_t}-\bm{\widehat q_t}, \bm{\ell_t}\rangle\right|
&=O\left(\left(\sqrt{HS^2A\left(\sum_{t=1}^T\langle \bm{q_t},\vec h\odot \bm{\ell_t}\rangle+H^3\right)}+H^2S^2A\right)\left(\ln\frac{HSAT}{\delta}\right)^2 \right),
\end{align*}
as required.
\end{proof}

Based on Lemmas \ref{lemma2} and \ref{lemma9}, we can prove \Cref{lemma:regret-term3} that bounds the difference between the expected reward and the realized reward.

\section{REGRET ANALYSIS FOR THE OBSERVE-THEN-DECIDE REGIME}

\subsection{Proofs of Lemmas 5.1 and 5.2 %
}
Based on Lemmas \ref{lemma2} and \ref{lemma9}, we can prove \Cref{lemma:regret-term3} that bounds the difference between the expected reward and the realized reward and \Cref{lemma:regret-term2} that bounds the regret due to the estimation error.

\begin{proof}[\rm \bfseries Proof of \Cref{lemma:regret-term3}]
We closely follow the proof of~\citep[Theorem 6]{ssp-adversarial-unknown}. For ease of notation, let us use notation $\pi_t$ for an arbitrary policy for episode $t$, $q_t$ denotes the occupancy measure $q^{P,\pi_t}$, and $n_t$ denotes $n^{P,\pi_t}$. Then \Cref{lemma2} implies that
$$\mathbb{E}\left[\langle \bm{n_t}, \bm{\ell_t}\rangle^2\mid \ell_t, \pi_t, P\right]\leq 2\langle \bm{q_t},\vec h\odot \bm{\ell_t}\rangle$$
where $\bm{q_t}, \bm{n_t},\bm{\ell_t}$ are the vector representations of $q_t,n_t,\ell_t:\cS\times\cA\times[H]\to\mathbb{R}.$
We define $\xi_1$ as 
$\xi_1=\left\{\ell_1,\pi_1\right\}$
and for $t\geq 2$, we define $\xi_t$ as 
$$\left\{s_1^{P,\pi_{t-1}},a_1^{P,\pi_{t-1}},\ldots, s_h^{P,\pi_{t-1}}, a_h^{P,\pi_{t-1}}, \ell_t,\pi_t\right\}$$
where $\pi_{t-1}$ and $\pi_t$ denote the policies for episode $t-1$ and episode $t$, respectively, and $$\left(s_1^{P,\pi_{t-1}},a_1^{P,\pi_{t-1}},\ldots, s_h^{P,\pi_{t-1}}, a_h^{P,\pi_{t-1}}\right)$$
is the trajectory generated under policy $\pi_{t-1}$ and transition kernel $P$. Then for $t\in[T]$, let $\cH_t$ be defined as the $\sigma$-algebra generated by the random variables in $\xi_1\cup\cdots\cup \xi_{t}$. Then it follows that $\cH_1,\ldots, \cH_T$ give rise to a filtration. Then it follows that
\begin{align*}
\sum_{t=1}^T\mathbb{E}\left[\langle \bm{n_t}, \bm{\ell_t}\rangle^2\mid \cH_t,P\right]=\sum_{t=1}^T\mathbb{E}\left[\langle \bm{n_t}, \bm{\ell_t}\rangle^2\mid \ell_t,\pi_t,P\right]&\leq 2\sum_{t=1}^T\langle \bm{q_t}-\bm{\widehat q_t},\vec h\odot \bm{\ell_t}\rangle +2\sum_{t=1}^T\langle \bm{\widehat q_t},\vec h\odot \bm{\ell_t}\rangle.
\end{align*}
Note that the first term on the right-hand side can be bounded as follows.
$$\sum_{t=1}^T\langle \bm{\widehat q_t},\vec h\odot \bm{\ell_t}\rangle=O(H^2T).$$
To upper bound the first term, we consider
\begin{align*}
\sum_{t=1}^T\langle \bm{q_t}-\bm{\widehat q_t},\vec h\odot \bm{\ell_t}\rangle&\leq H\sum_{t=1}^T\langle \bm{q_t}-\bm{\widehat q_t},\bm{\ell_t}\rangle.
\end{align*}
Applying \Cref{lemma9} with function $\ell_t$,  we deduce that with probability at least $1-4\delta$,
\begin{align*}
\sum_{t=1}^T \langle \bm{q_t}-\bm{\widehat q_t}, \bm{\ell_t}\rangle
&=O\left(\left(\sqrt{HS^2A\left(\sum_{t=1}^T\langle \bm{q_t},\vec h\odot \bm{\ell_t} \rangle+H^3\sqrt{T}\right)}+H^2S^2A\right)\left(\ln\frac{HSAT}{\delta}\right)^2 \right)\\
&=O\left(\left(\sqrt{HS^2A\left(H^2T+H^3\sqrt{T}\right)}+H^2S^2A\right)\left(\ln\frac{HSAT}{\delta}\right)^2 \right)\\
&=O\left(\left(\sqrt{H^4S^2AT}+H^2S^2A\right)\left(\ln\frac{HSAT}{\delta}\right)^2 \right)\\
&=O\left(\left(H T +H^3S^2A\right)\left(\ln\frac{HSAT}{\delta}\right)^2 \right)
\end{align*}
where the second equality holds because $\langle \bm{q_t},\vec h\odot \bm{\ell_t} \rangle=O(H^2)$ and the fourth equality holds because $\sqrt{H^4S^2AT}\leq H\sqrt{T} + H^3S^2A$. Then it follows that
$$\sum_{t=1}^T\langle \bm{q_t}-\bm{\widehat q_t},\vec h\odot \bm{\ell_t}\rangle= O\left(\left(H^2 T +H^4S^2A\right)\left(\ln\frac{HSAT}{\delta}\right)^2 \right).$$
Therefore, we obtain
$$\sum_{t=1}^T\mathbb{E}\left[\langle \bm{n_t}, \bm{\ell_t}\rangle^2\mid \cH_t,P\right]=O\left(\left(H^2 T +H^4S^2A\right)\left(\ln\frac{HSAT}{\delta}\right)^2 \right).$$
Next, we apply \Cref{bernstein2} with $\lambda$ is set to $$\lambda=\frac{1}{\sqrt{H^2 T + H^4 S^2 A}}\leq \frac{1}{H}\leq \frac{1}{\langle \bm{n_t},\bm{\ell_t}\rangle}.$$
Then we get that with probability at least $1-\delta$,
\begin{align*}
\sum_{t=1}^T\langle \bm{q_t}- \bm{n_t}, \bm{\ell_t}\rangle&\leq \lambda \sum_{t=1}^T\mathbb{E}\left[\langle \bm{q_t}-\bm{n_t}, \bm{\ell_t}\rangle^2\mid \cH_t,P\right]+\frac{1}{\lambda} \ln\frac{1}{\delta}\\
&\leq \lambda\sum_{t=1}^T\mathbb{E}\left[\langle \bm{n_t}, \bm{\ell_t}\rangle^2\mid \cH_t,P\right]+\frac{1}{\lambda} \ln\frac{1}{\delta}\\
&=O\left(\left(H\sqrt{T} + H^2S\sqrt{A}\right)\left(\ln\frac{HSAT}{\delta}\right)^2\right).
\end{align*}
Similarly, we get that with probability at least $1-\delta$,
\begin{align*}
	\sum_{t=1}^T\langle \bm{n_t}- \bm{q_t}, \bm{\ell_t}\rangle&\leq \lambda \sum_{t=1}^T\mathbb{E}\left[\langle \bm{n_t}-\bm{q_t}, \bm{\ell_t}\rangle^2\mid \cH_t,P\right]+\frac{1}{\lambda} \ln\frac{1}{\delta}\\
	&\leq \lambda\sum_{t=1}^T\mathbb{E}\left[\langle \bm{n_t}, \bm{\ell_t}\rangle^2\mid \cH_t,P\right]+\frac{1}{\lambda} \ln\frac{1}{\delta}\\
	&=O\left(\left(H\sqrt{T} + H^2S\sqrt{A}\right)\left(\ln\frac{HSAT}{\delta}\right)^2\right),
\end{align*}
as required.
\end{proof}

\begin{proof}[\rm \bfseries Proof of \Cref{lemma:regret-term2}]
We closely follow the proof of~\citep[Theorem 6]{ssp-adversarial-unknown}.
For ease of notation, let us use notation $\pi_t$ for an arbitrary policy for episode $t$, $q_t$ denotes the occupancy measure $q^{P,\pi_t}$, and $\widehat q_t$ denotes the occupancy measure $q^{P_t,\pi_t}$.

Applying \Cref{lemma9} with function $\ell_t$,  we deduce that with probability at least $1-4\delta$,
\begin{align*}
\sum_{t=1}^T\left| \langle \bm{q_t}-\bm{\widehat q_t}, \bm{\ell_t}\rangle\right|
&=O\left(\left(\sqrt{HS^2A\left(\sum_{t=1}^T\langle \bm{q_t},\vec h\odot \bm{\ell_t} \rangle+H^3\sqrt{T}\right)}+H^2S^2A\right)\left(\ln\frac{HSAT}{\delta}\right)^2 \right)\\
&=O\left(\left(\sqrt{H^3S^2AT+H^4S^2A\sqrt{T}}+H^2S^2A\right)\left(\ln\frac{HSAT}{\delta}\right)^2 \right)\\
&=O\left(\left(\sqrt{H^3S^2AT+H^3S^2AT + H^5S^2A}+H^2S^2A\right)\left(\ln\frac{HSAT}{\delta}\right)^2 \right)\\
&=O\left(\left(H^{3/2}S\sqrt{AT} + H^{5/2}S\sqrt{A}+H^2S^2A\right)\left(\ln\frac{HSAT}{\delta}\right)^2 \right)
\end{align*}
where the second equality holds because $\langle \bm{q_t},\vec h\odot \bm{\ell_t} \rangle=O(H^2)$ and the third equality holds because $H^4S^2A\sqrt{T}=O(H^3S^2AT + H^5S^2A)$.
\end{proof}

\subsection{Bound on the Regret Term (II)}

In this section, we prove \Cref{lemma:regret-term1} that bounds the regret term (II). We follow the analysis of the online dual mirror descent algorithm due to~\citep[Theorem 1]{balseiro2022}. In our analysis, we need Lemmas 5.1 and 5.2. %

We consider 
$$\widehat L_t(\lambda) =\max_{\bm{q}\in \Delta(P,t)}\left\{\langle \bm{f_t},\bm{q}\rangle +\lambda(H\rho- \langle\bm{g_t},\bm{q}\rangle)\right\},\quad L_t(\lambda) =\max_{\bm{q}\in \Delta(P)}\left\{\langle \bm{f_t},\bm{q}\rangle +\lambda(H\rho- \langle\bm{g_t},\bm{q}\rangle)\right\}.$$
\begin{lemma}{\rm \citep[Proposition 1]{balseiro2022}}\label{lemma:duality}
	For any $\lambda \in\mathbb{R}_+$, we have $\opt(\vec\gamma) \leq \sum_{t=1}^ T L_t(\lambda).$
\end{lemma}
Moreover, it follows from \Cref{lemma:relaxation} that with probability at least $1-4\delta$, 
$$\widehat L_t(\lambda_t)\geq L_t(\lambda_t),\quad \forall t\in[T],$$
which implies that
$$\langle \bm{f_t},\bm{\widehat q_t}\rangle\geq L_t(\lambda_t) -\lambda_t(H\rho- \langle\bm{g_t},\bm{\widehat q_t}\rangle),\quad \forall t\in[T].$$

\begin{lemma}\label{lemma:regret-term1}
	The following holds for the regret term (II).
	$$\mathbb{E}\left[\opt(\vec\gamma) - \sum_{t=1}^T \langle \bm{f_t}, \bm{\widehat q_t}\rangle\mid P \right]= O\left(\left(\frac{H^{3/2}}{\rho}S\sqrt{AT} +\frac{H^{5/2}}{\rho}S^2A \right)\left(\ln HSAT\right)^2\right)$$
	where the expectation is taken with respect to the randomness of the reward and resource consumption functions and the randomness in the trajectories of episodes. 
\end{lemma}
\begin{proof}
For $t\in[T]$, let $G_t$ denote the amount of resource consumed in episode $t$.
We define the stopping time $\tau$ of \Cref{alg:online-alloc-mdp-unknown} as 
$$\min\left\{t:\ \sum_{k=1}^tG_k + H\geq  TH\rho\right\}.$$
By definition, we have $TH\rho -\sum_{t=1}^{\tau-1} G_t > H$.  Since $g_t(s,a,h)\leq 1$ for any $(s,a,h,t)\in\cS\times\cA\times[H]\times [T]$, it follows that Algorithm 1 %
does not terminate until the end of episode $\tau$. Then we have
$$G_t = \langle\bm{n_t},\bm{g_t}\rangle,\quad t\leq \tau.$$

By \Cref{lemma:relaxation}, 
with probability at least $1-4\delta$, we have 
\begin{equation}\label{eq:upperbound}
\langle \bm{f_t},\bm{\widehat q_t}\rangle=\widehat L_t(\lambda_t) -\lambda_t(H\rho- \langle\bm{g_t},\bm{\widehat q_t}\rangle)\geq L_t(\lambda_t) -\lambda_t(H\rho- \langle\bm{g_t},\bm{\widehat q_t}\rangle),\quad \forall t\in[T]\end{equation}
where \begin{align*}\widehat L_t(\lambda) &=\max_{\bm{q}\in \Delta(P,t)}\left\{\langle \bm{f_t},\bm{q}\rangle +\lambda(H\rho- \langle\bm{g_t},\bm{q}\rangle)\right\},\\
L_t(\lambda) &=\max_{\bm{q}\in \Delta(P)}\left\{\langle \bm{f_t},\bm{q}\rangle +\lambda(H\rho- \langle\bm{g_t},\bm{q}\rangle)\right\}.
\end{align*}
Recall that the pair $(f_t,g_t)$ of reward and resource consumption functions follows distribution $\mathcal{D}$.
Then we define $\bar L(\lambda)$ as
$$\bar L(\lambda) = \mathbb{E}_{(f,g)\sim \mathcal{D}}\left[\max_{\bm{q}\in \Delta(P)}\left\{\langle \bm{f},\bm{q}\rangle +\lambda(H\rho- \langle\bm{g},\bm{q}\rangle)\right\}\right].$$
Since $(f_t,g_t)$ for $t\in[T]$ are i.i.d. with distribution $\mathcal{D}$, it follows that
$$ \frac{1}{T}\mathbb{E}\left[\sum_{t=1}^T L_t(\lambda)\right]=\bar L(\lambda)=\mathbb{E}_{(f,g)\sim \mathcal{D}}\left[\max_{\bm{q}\in \Delta(P)}\left\{\langle \bm{f},\bm{q}\rangle - \lambda \langle\bm{g},\bm{q}\rangle\right\}\right] +\lambda H\rho.$$
Recall that $\cG_t$ is the $\sigma$-algebra generated by the information up to episode $t-1$.
Consider 
$$Z_t=\sum_{k=1}^t\lambda_k(H\rho- \langle\bm{g_k},\bm{\widehat q_k}\rangle )- \sum_{k=1}^t \mathbb{E}\left[\lambda_k(H\rho- \langle\bm{g_k},\bm{\widehat q_k}\rangle )\mid \cG_{k-1}\right] $$
for $t\in[T]$. Then $Z_t$ is $\cG_t$-measurable and $\mathbb{E}\left[Z_{t+1}\mid \cG_t\right]= Z_t$. Therefore, $Z_1,\ldots, Z_T$ is a martingale. Since the stopping time $\tau$ is with respect to $\{\cG_t\}_{t\in[T]}$ and $\tau$ is bounded, the Optional Stopping Theorem implies that
$$\mathbb{E}\left[\sum_{t=1}^\tau \lambda_t(H\rho- \langle\bm{g_t},\bm{\widehat q_t}\rangle )\right] = \mathbb{E}\left[\sum_{t=1}^\tau \mathbb{E}\left[\lambda_t(H\rho- \langle\bm{g_t},\bm{\widehat q_t}\rangle )\mid \cG_{t-1}\right]\right].$$
Likewise, we can argue by the Optional Stopping Theorem that
$$\mathbb{E}\left[\sum_{t=1}^\tau \langle\bm{f_t},\bm{\widehat q_t}\rangle \right] = \mathbb{E}\left[\sum_{t=1}^\tau \mathbb{E}\left[\langle\bm{f_t},\bm{\widehat q_t}\rangle\mid \cG_{t-1}\right]\right].$$
Taking the conditional expectation with respect to $\cG_{t-1}$ of both sides of~\eqref{eq:upperbound}, it follows that
\begin{align*}
&\mathbb{E}\left[\langle \bm{f_t},\bm{\widehat q_t}\rangle\mid \cG_{t-1}\right]\\
&\geq \mathbb{E}\left[\max_{\bm{q}\in \Delta(P)}\left\{\langle \bm{f_t},\bm{q}\rangle +\lambda_t(H\rho- \langle\bm{g_t},\bm{q}\rangle)\right\}\mid\cG_{t-1}\right] -\mathbb{E}\left[\lambda_t(H\rho- \langle\bm{g_t},\bm{\widehat q_t}\rangle)\mid\cG_{t-1}\right]\\
&= \mathbb{E}\left[\mathbb{E}_{(f_t,g_t)\sim\mathcal{D}}\left[\max_{\bm{q}\in \Delta(P)}\left\{\langle \bm{f_t},\bm{q}\rangle +\lambda_t(H\rho- \langle\bm{g_t},\bm{q}\rangle)\right\}\right]\mid\cG_{t-1}\right] -\mathbb{E}\left[\lambda_t(H\rho- \langle\bm{g_t},\bm{\widehat q_t}\rangle)\mid\cG_{t-1}\right]\\
&=\mathbb{E}\left[\bar L(\lambda_t)\mid\cG_{t-1}\right] -\mathbb{E}\left[\lambda_t(H\rho- \langle\bm{g_t},\bm{\widehat q_t}\rangle)\mid\cG_{t-1}\right]\\
&=\bar L(\lambda_t) -\mathbb{E}\left[\lambda_t(H\rho- \langle\bm{g_t},\bm{\widehat q_t}\rangle)\mid\cG_{t-1}\right]
\end{align*}
where the first equality is due to the tower rule and the last equality holds because $\lambda_t$ is $\cG_{t-1}$-measurable. Therefore, 
\begin{align*}
    \mathbb{E}\left[\sum_{t=1}^\tau \langle\bm{f_t},\bm{\widehat q_t}\rangle \right]& \geq \mathbb{E}\left[\sum_{t=1}^\tau \bar L(\lambda_t)\right]-  \mathbb{E}\left[\sum_{t=1}^\tau \mathbb{E}\left[\lambda_t(H\rho- \langle\bm{g_t},\bm{\widehat q_t}\rangle )\mid \cG_{t-1}\right]\right]\\
    &=\mathbb{E}\left[\sum_{t=1}^\tau \bar L(\lambda_t)\right]-\mathbb{E}\left[\sum_{t=1}^\tau \lambda_t(H\rho- \langle\bm{g_t},\bm{\widehat q_t}\rangle )\right]
    \end{align*}
    where the second equality comes from the Optional Stopping Theorem. Furthermore, note that $L_t(\lambda)$ is the maximum of linear functions in terms of $\lambda$, so $L_t(\lambda)$ is  convex for any $t\in[T]$. Then $\bar L(\lambda)$ is also convex with respect to $\lambda$, and therefore,
    $$\mathbb{E}\left[\sum_{t=1}^\tau \bar L(\lambda_t)\right]\geq \mathbb{E}\left[\tau  \bar L\left(\frac{1}{\tau} \sum_{t=1}^\tau \lambda_t\right)\right].$$
    This implies that
    $$\mathbb{E}\left[\sum_{t=1}^\tau \langle\bm{f_t},\bm{\widehat q_t}\rangle \right]\geq\mathbb{E}\left[\tau  \bar L\left(\frac{1}{\tau} \sum_{t=1}^\tau \lambda_t\right)\right]- \mathbb{E}\left[\sum_{t=1}^\tau \lambda_t(H\rho- \langle\bm{g_t},\bm{\widehat q_t}\rangle )\right].$$
    Next, consider the second term on the right-hand side of this inequality:
    $$\sum_{t=1}^\tau \lambda_t(H\rho- \langle\bm{g_t},\bm{\widehat q_t}\rangle ).$$
    Let $w_t(\lambda)$ be defined as 
    $$w_t(\lambda)=\lambda(H\rho- \langle\bm{g_t},\bm{\widehat q_t}\rangle).$$
    Then the dual update rule $$\lambda_{t+1}=\max\left\{0,\lambda_t - \eta\left(H\rho - \langle \bm{g_t}, \bm{\widehat q_t}\rangle\right)\right\}$$
    corresponds to the online mirror descent algorithm applied to the linear functions $w_t(\lambda)$ for $t\in[\tau]$. Since  $|H\rho - \langle \bm{g_t}, \bm{\widehat q_t}\rangle|\leq 2H$, the standard analysis of online mirror descent (see \citep{Hazan}) gives us that
    $$\sum_{t=1}^\tau \lambda_t(H\rho- \langle\bm{g_t},\bm{\widehat q_t}\rangle )- \sum_{t=1}^\tau \lambda(H\rho- \langle\bm{g_t},\bm{\widehat q_t}\rangle )\leq 2H^2\eta \tau + \frac{1}{2\eta} (\lambda- \lambda_1)^2\leq 2H^2\eta T + \frac{1}{2\eta} (\lambda- \lambda_1)^2.$$
 Next, note that for any $\lambda \geq 0$,
    \begin{align*}
    \mathbb{E}\left[\opt(\vec \gamma)\right]&=\frac{T-\tau}{T}\mathbb{E}\left[\opt(\vec \gamma)\right]+\frac{\tau}{T}\mathbb{E}\left[\opt(\vec \gamma)\right]\\
    &\leq (T-\tau)H +\frac{\tau}{T}\mathbb{E}\left[\sum_{t=1}^T L_t\left(\lambda \right)\right]\\
    &= (T-\tau)H +\tau \bar L \left(\lambda\right)
     \end{align*}
     where the second inequality is implied by $\opt(\vec \gamma)\leq TH$ and \Cref{lemma:duality}. 
     In particular, we set $\lambda = \frac{1}{\tau}\sum_{t=1}^\tau \lambda_t$, and obtain
     $$\mathbb{E}\left[\opt(\vec \gamma)\right]\leq (T-\tau)H +\tau \bar L \left(\frac{1}{\tau}\sum_{t=1}^\tau \lambda_t\right).$$
     Then it follows that
     \begin{align*}
\mathbb{E}\left[\opt(\vec \gamma)- \sum_{t=1}^T\langle \bm{f_t},\bm{\widehat q_t}\rangle\mid P\right]&\leq \mathbb{E}\left[\opt(\vec \gamma)- \sum_{t=1}^\tau \langle \bm{f_t},\bm{\widehat q_t}\rangle\mid P\right]\\
&\leq \mathbb{E}\left[(T-\tau)H + \sum_{t=1}^\tau \lambda_t(H\rho- \langle\bm{g_t},\bm{\widehat q_t}\rangle ) \mid P\right]\\
&\leq \mathbb{E}\left[(T-\tau)H + \sum_{t=1}^\tau \lambda(H\rho- \langle\bm{g_t},\bm{\widehat q_t}\rangle ) + 2H^2\eta T+ \frac{1}{2\eta}(\lambda-\lambda_1)^2 \mid P\right]
     \end{align*}
     where the last inequality is from the online mirror descent analysis. 
     
     If $\tau =T$, then we set $\lambda =0$, in which case
     $$\mathbb{E}\left[\opt(\vec \gamma)- \sum_{t=1}^T\langle \bm{f_t},\bm{\widehat q_t}\rangle\mid P\right]\leq 2H^2\eta T + \frac{1}{2\eta}\lambda_1^2.$$

     If $\tau<T$, then we have $$\sum_{t=1}^\tau G_t + H =\sum_{t=1}^\tau \langle\bm{g_t},\bm{n_t}\rangle + H \geq TH\rho.$$
     In this case, we set $\lambda  =  1/\rho$. Then
     \begin{align*}
        \sum_{t=1}^\tau \lambda(H\rho- \langle\bm{g_t},\bm{\widehat q_t}\rangle )
    &= \tau H - \frac{1}{\rho}\sum_{t=1}^\tau \langle\bm{g_t},\bm{n_t}\rangle+\frac{1}{\rho}\sum_{t=1}^\tau \langle\bm{g_t},\bm{n_t}-\bm{q_t}\rangle+\frac{1}{\rho}\sum_{t=1}^\tau \langle\bm{g_t},\bm{q_t}-\bm{\widehat q_t}\rangle.
     \end{align*}
     By \Cref{lemma:regret-term2} and \Cref{lemma:regret-term3}, 
     with probability at least $1-14\delta$, we have
      \begin{align*}
 &\tau H - \frac{1}{\rho}\sum_{t=1}^\tau \langle\bm{g_t},\bm{n_t}\rangle+\frac{1}{\rho}\sum_{t=1}^\tau \langle\bm{g_t},\bm{n_t}-\bm{q_t}\rangle+\frac{1}{\rho}\sum_{t=1}^\tau \langle\bm{g_t},\bm{q_t}-\bm{\widehat q_t}\rangle\\
 &\leq \tau H - TH + \frac{H}{\rho} + O\left(\left(\frac{H^{3/2}}{\rho}S\sqrt{AT} +\frac{H^{5/2}}{\rho}S^2A\right)\left(\ln\frac{HSAT}{\delta}\right)^2 \right).
     \end{align*}
     In this case, we deduce that 
    \begin{align*}
    &\mathbb{E}\left[\opt(\vec \gamma)- \sum_{t=1}^T\langle \bm{f_t},\bm{\widehat q_t}\rangle\mid P\right]\\
    &\leq 2H^2\eta T + \frac{1}{2\eta}\left(\frac{1}{\rho}-\lambda_1\right)^2+O\left(\left(\frac{H^{3/2}}{\rho}S\sqrt{AT} +\frac{H^{5/2}}{\rho}S^2A\right)\left(\ln\frac{HSAT}{\delta}\right)^2 \right)\\
    &\leq 2H^2\eta T + \frac{1}{\eta\rho^2}\left(1+\lambda_1\right)^2+O\left(\left(\frac{H^{3/2}}{\rho}S\sqrt{AT} +\frac{H^{5/2}}{\rho}S^2A\right)\left(\ln\frac{HSAT}{\delta}\right)^2 \right)
    \end{align*}
    where the second inequality holds because $\rho<1$.
    Setting $$\eta = \frac{1}{\rho H\sqrt{T}},$$
    we deduce that 
    $$ \mathbb{E}\left[\opt(\vec \gamma)- \sum_{t=1}^T\langle \bm{f_t},\bm{\widehat q_t}\rangle\mid P\right]= O\left(\left(\frac{H^{3/2}}{\rho}S\sqrt{AT} +\frac{H^{5/2}}{\rho}S^2A\right)\left(\ln\frac{HSAT}{\delta}\right)^2 \right).$$
    Now we may set 
    $$\delta = \frac{1}{13HT}.$$
    Note that with probability at most ${1}/{HT}$,
    $$ \mathbb{E}\left[\opt(\vec \gamma)- \sum_{t=1}^T\langle \bm{f_t},\bm{\widehat q_t}\rangle\mid P\right]\leq \opt(\vec \gamma)\leq HT.$$ 
    Moreover, with probability at least $1-1/HT$, 
    \begin{align*}\mathbb{E}\left[\opt(\vec \gamma)- \sum_{t=1}^T\langle \bm{f_t},\bm{\widehat q_t}\rangle\mid P\right]&=O\left(\left(\frac{H^{3/2}}{\rho}S\sqrt{AT} +\frac{H^{5/2}}{\rho}S^2A\right)\left(\ln\frac{H^2SAT^2}{13}\right)^2 \right)\\
    &=O\left(\left(\frac{H^{3/2}}{\rho}S\sqrt{AT} +\frac{H^{5/2}}{\rho}S^2A\right)\left(\ln{HSAT}\right)^2 \right).
    \end{align*}
    
    Then it follows that
    $$ \mathbb{E}\left[\opt(\vec \gamma)- \sum_{t=1}^T\langle \bm{f_t},\bm{\widehat q_t}\rangle\mid P\right]= O\left(\left(\frac{H^{3/2}}{\rho}S\sqrt{AT} +\frac{H^{5/2}}{\rho}S^2A\right)\left(\ln{HSAT}\right)^2 \right),$$
    as required.
\end{proof}

\subsection{Proof of \Cref{theorem:regret1}
}

Now we are ready to prove \Cref{theorem:regret1}. Recall that
\begin{align*}
	&\regret\left(\vec\gamma,\vec\pi\right)=\opt(\vec\gamma)-\reward\left(\vec\gamma,\vec\pi\right)\\
	&=\opt(\vec\gamma) - \sum_{t=1}^T \langle \bm{f_t}, \bm{\widehat q_t}\rangle + \sum_{t=1}^T \langle \bm{f_t}, \bm{\widehat q_t}-\bm{q_t}\rangle + \sum_{t=1}^T \langle \bm{f_t}, \bm{q_t}\rangle-\sum_{t=1}^{T}\sum_{h=1}^{H}f_{t}\left(s_{h}^{P,\pi_t},a_{h}^{P,\pi_t},h\right).
\end{align*} Then it follows from Lemmas~\ref{lemma:regret-term1}, \ref{lemma:regret-term3}, and \ref{lemma:regret-term2} 
that
$$ \mathbb{E}\left[\regret\left(\vec\gamma,\vec\pi\right)\mid P\right]= O\left(\left(\frac{H^{3/2}}{\rho}S\sqrt{AT} +\frac{H^{5/2}}{\rho}S^2A\right)\left(\ln{HSAT}\right)^2 \right),$$
as required.

 \section{REGRET ANALYSIS FOR THE DECIDE-THEN-OBSERVE REGIME}
 
 Suppose that the statements of \Cref{lemma:estimator} 
 hold, which is the case with probability at least $1-2\delta$. 
 
 Recall that
 \begin{align*}
 	\regret\left(\vec\gamma,\vec\pi\right)
 	&=\underbrace{\opt(\vec\gamma) - \sum_{t=1}^T \langle \bm{f}, \bm{q^*}\rangle}_{\text{(I)}} +\underbrace{\sum_{t=1}^T \langle \bm{f}, \bm{q^*}\rangle-\sum_{t=1}^T \langle \bm{\widehat f_t}, \bm{\widehat q_t}\rangle}_{\text{(II)}} \\
 	&\quad + \underbrace{\sum_{t=1}^T \langle \bm{\widehat f_t}, \bm{\widehat q_t}-\bm{q_t}\rangle}_{\text{(III)}}+\underbrace{\sum_{t=1}^T \langle \bm{\widehat f_t}-\bm{f}, \bm{q_t}\rangle}_{\text{(IV)}} + \underbrace{\sum_{t=1}^T \langle \bm{f}, \bm{q_t}\rangle- \reward\left(\vec\gamma,\vec\pi\right)}_{\text{(V)}}.
 \end{align*}

\subsection{Bound on the regret term (I)}
\begin{lemma}\label{lemma:second-term1}
With probability at least $1-7\delta$, 
$$\opt(\vec\gamma) - \sum_{t=1}^T \langle \bm{f}, \bm{q^*}\rangle=O\left(\left(H\sqrt{T} + H^2S\sqrt{A}\right)\left(\ln\frac{HSAT}{\delta}\right)^2\right).$$
\end{lemma}
\begin{proof}
Note that term (I) equals $$\sum_{t=1}^T\langle\bm{f_t},\bm{n^*}-\bm{q^*}\rangle+\sum_{t=1}^T\langle\bm{f_t}-\bm{f},\bm{q^*}\rangle.$$ 
Note that $f_1,\ldots, f_T$ are i.i.d. and that 
$$\mathbb{E}\left[\langle\bm{f_t},\bm{q^*}\rangle\right] =\langle\bm{f},\bm{q^*}\rangle.$$
Then it follows from \Cref{cohen-concentration0} that
$$\sum_{t=1}^T\langle\bm{f_t}-\bm{f},\bm{q^*}\rangle\leq 2\sqrt{H^2T \ln \frac{2T}{\delta}} + H\ln\frac{2T}{\delta}$$
holds wity probability at least $1-\delta$. With probability at least $1-6\delta$, \begin{align*}
\sum_{t=1}^T\langle\bm{f_t},\bm{n^*}-\bm{q^*}\rangle
	&=O\left(\left(H\sqrt{T} + H^2S\sqrt{A}\right)\left(\ln\frac{HSAT}{\delta}\right)^2\right),
\end{align*}
as required.
\end{proof}

\subsection{Proofs of Lemmas \ref{lemma:second-term4} and \ref{lemma:second-term5}
}

\begin{proof}[Proof of \Cref{lemma:second-term4}]
	Suppose that the statements of \Cref{lemma:estimator} 
	hold, which is the case with probability at least $1-2\delta$. Then it follows that
	\begin{align*}
		&\sum_{t=1}^T \langle \bm{\widehat f_t}-\bm{f}, \bm{q_t}\rangle,\ \sum_{t=1}^T \langle\bm{g}- \bm{\widehat g_t}, \bm{q_t}\rangle
		\\
		&\leq 8\underbrace{\sum_{t=1}^T \sum_{h=1}^H \sum_{(s,a)\in\cS\times\cA} q_t(s,a,h)\sqrt{R_t(s,a,h) f(s,a,h)}}_{\text{Term 1}} + 34\underbrace{\sum_{t=1}^T \sum_{h=1}^H \sum_{(s,a)\in\cS\times\cA} q_t(s,a,h)R_t(s,a,h)}_{\text{Term 2}}.
	\end{align*}
	We first bound Term 2 under both the full information setting and the bandit feedback setting. Let $x_t(s,a,h)$ be the probability that $f_t(s,a,h)$ and $g_t(s,a,h)$ are observed for state-action paper $(s,a)\in\cS\times \cA$ at step $h\in[H]$ of episode $t\in[T]$. Note that
	$$q_t(s,a,h)R_t(s,a,h)=  \frac{q_t(s,a,h)}{C_t(s,a,h)}\ln\frac{2HSAT}{\delta} \leq \frac{x_t(s,a,h)}{C_t(s,a,h)}\ln\frac{2HSAT}{\delta}$$
	and that
	$$\frac{x_t(s,a,h)}{C_t(s,a,h)}=\mathbb{E}\left[\frac{\mathbbm{1}_t(s,a,h)}{C_t(s,a,h)}\mid \cG_t, P\right]$$
	where $\cG_t$ is the $\sigma$-algebra generated by the information up to episode $t-1$. Note that by \Cref{cohen-concentration}, we deduce that
	$$\sum_{t=1}^T\sum_{h=1}^H\sum_{(s,a)\in\cS\times\cA}\frac{x_t(s,a,h)}{C_t(s,a,h)}\leq \sum_{t=1}^T\sum_{h=1}^H\sum_{(s,a)\in\cS\times\cA} \mathbb{E}\left[\frac{\mathbbm{1}_t(s,a,h)}{C_t(s,a,h)}\mid \cG_t, P\right]+4HSA\ln\frac{2HSAT}{\delta}$$
	holds with probability at least $1-\delta$. Under both the full information setting and the bandit feedback setting, we have
	$$\sum_{t=1}^T \mathbb{E}\left[\frac{\mathbbm{1}_t(s,a,h)}{C_t(s,a,h)}\mid \cG_t, P\right]\leq \sum_{t=1}^T \frac{1}{\max\{1,t-1\}}= O\left(\ln\frac{HSAT}{\delta}\right).$$
	Therefore, Term 2 can be bounded as
	$$\text{Term 2}=O\left(HSA\left(\ln\frac{HSAT}{\delta}\right)^2\right).$$
	For Term 1, note that
	\begin{align*}
		&\sum_{t=1}^T \sum_{h=1}^H \sum_{(s,a)\in\cS\times\cA} q_t(s,a,h)\sqrt{R_t(s,a,h) f(s,a,h)}\\&\leq \sqrt{\sum_{t=1}^T \sum_{h=1}^H \sum_{(s,a)\in\cS\times\cA} q_t(s,a,h)}\sqrt{\sum_{t=1}^T \sum_{h=1}^H \sum_{(s,a)\in\cS\times\cA} q_t(s,a,h)f(s,a,h) R_t(s,a,h)} \\
		&\leq \sqrt{HT\sum_{t=1}^T\sum_{h=1}^H \sum_{(s,a)\in\cS\times\cA} q_t(s,a,h)R_t(s,a,h)}
	\end{align*}
	where the first inequality is due to the Cauchy-Schwarz inequality and the second inequality holds because $ \sum_{(s,a)\in\cS\times\cA} q_t(s,a,h)=1$ and $f(s,a,h)\leq 1$. Under the full information setting, we have $R_t(s,a,h) =\ln(2HSAT/\delta)/\max\{1, t-1\}$. Since  $ \sum_{(s,a)\in\cS\times\cA} q_t(s,a,h)=1$, it follows that
	\begin{align*}
		\sqrt{HT\sum_{t=1}^T\sum_{h=1}^H \sum_{(s,a)\in\cS\times\cA} q_t(s,a,h)R_t(s,a,h)}&\leq \sqrt{H^2T \ln\frac{2HSAT}{\delta} \sum_{t=1}^T \frac{1}{\max\{1,t-1\}}}\\
		&= O\left(H\sqrt{T}\ln \frac{HSAT}{\delta}\right).
	\end{align*}
	Under the bandit feedback setting, 
	\begin{align*}
		\sqrt{HT\sum_{t=1}^T\sum_{h=1}^H \sum_{(s,a)\in\cS\times\cA} q_t(s,a,h)R_t(s,a,h)}= \sqrt{HT\cdot \text{Term 2}}=O\left(H\sqrt{SAT}\ln \frac{HSAT}{\delta}\right).
	\end{align*}
	Therefore, 
	$$\sum_{t=1}^T \langle \bm{\widehat f_t} - \bm{f}, \bm{q_t}\rangle,\ \sum_{t=1}^T \langle\bm{g}- \bm{\widehat g_t}, \bm{q_t}\rangle=\begin{cases}
		O\left((H\sqrt{T} + HSA)\left(\ln \frac{HSAT}{\delta}\right)^2\right),&\text{the full information setting},\\
		O\left((H\sqrt{SAT} + HSA)\left(\ln \frac{HSAT}{\delta}\right)^2\right),&\text{the bandit feedback setting},
	\end{cases}$$
	as required.
\end{proof}

\begin{proof}[Proof of \Cref{lemma:second-term5}]
Recall that $\cG_t$ is the $\sigma$-algebra generated by the information up to episode $t-1$. By \Cref{lemma2},
we have $\mathbb{E}\left[\langle \bm{n_t}, \bm{f_t}\rangle^2\mid f_t, \pi_t, P\right]\leq  2\langle \bm{q_t},\vec h\odot \bm{f_t}\rangle$, which implies that
$$\sum_{t=1}^T\mathbb{E}\left[\langle \bm{n_t}, \bm{f_t}\rangle^2\mid  \cG_t, P\right]\leq  2\sum_{t=1}^T\mathbb{E}\left[\langle \bm{q_t},\vec h\odot \bm{f_t}\rangle\mid \cG_t,P\right]=2\sum_{t=1}^T \langle \bm{q_t},\vec h\odot \bm{f}\rangle$$
because $\pi_t$ is $\cG_t$-measurable.
Therefore, it follows that
$$\sum_{t=1}^T\mathbb{E}\left[\langle \bm{n_t}, \bm{f_t}\rangle^2\mid  \cG_t, P\right]\leq 2\underbrace{\sum_{t=1}^T \langle \bm{q_t},\vec h\odot (\bm{f}-\bm{\widehat f_t})\rangle }_{\text{Term 1}}+2\underbrace{\sum_{t=1}^T \langle \bm{q_t}- \bm{\widehat q_t},\vec h\odot \bm{\widehat f_t}\rangle}_{\text{Term 2}}+ 2\underbrace{\sum_{t=1}^T \langle \bm{\widehat q_t},\vec h\odot \bm{\widehat f_t}\rangle}_{\text{Term 3}}.$$ 

Since $f(s,a,h)-\widehat f_t(s,a,h)\leq 0$ for any $(s,a,h)\in \cS\times \cA\times[H]$ by \Cref{lemma:estimator} 
\begin{align*}
	\text{Term 1} &= \sum_{t=1}^T \sum_{h=1}^H h\sum_{(s,a)\in\cS\times \cA} q_t(s,a,h) \left(f(s,a,h)-\widehat f_t(s,a,h)\right)\leq 0
\end{align*}
For Term 2, applying \Cref{lemma9}
to functions $(\vec h\odot \widehat f_t)/H$, we obtain 
\begin{align*}
	\text{Term 2}&=O\left(H\left(\sqrt{HS^2A\left(\sum_{t=1}^T\frac{1}{H}\langle \bm{q_t},\vec h\odot(\vec h \odot \bm{\widehat f_t}) \rangle+H^3\sqrt{T}\right)}+H^2S^2A\right)\left(\ln\frac{HSAT}{\delta}\right)^2 \right)\\
	&=O\left(\left(\sqrt{H^3S^2A\left(\sum_{t=1}^T\langle \bm{q_t},\vec h \odot \bm{\widehat f_t} \rangle+H^5\sqrt{T}\right)}+H^3S^2A\right)\left(\ln\frac{HSAT}{\delta}\right)^2 \right)\\
	&=O\left(\left(\sqrt{H^5S^2AT}+H^3S^2A\right)\left(\ln\frac{HSAT}{\delta}\right)^2 \right)\\
	&=O\left(\left(H^2T + H^3 S^2A\right)\left(\ln\frac{HSAT}{\delta}\right)^2 \right)
\end{align*}
For Term 3, note that
$$\sum_{t=1}^T \langle \bm{\widehat q_t},\vec h\odot \bm{\widehat f_t}\rangle\leq H^2T.$$
Based on the upper bounds on Terms 1, 2, and 3, we obtain
$$\sum_{t=1}^T\mathbb{E}\left[\langle \bm{n_t}, \bm{f_t}\rangle^2\mid  \cG_t, P\right]=O\left(\left(H^2T + H^3 S^2A\right)\left(\ln\frac{HSAT}{\delta}\right)^2 \right).$$
Moreover,
$$\sum_{t=1}^T \langle \bm{f}, \bm{q_t}\rangle- \reward\left(\vec\gamma,\vec\pi\right)= \sum_{t=1}^T \left(\langle \bm{f}, \bm{q_t}\rangle  - \langle \bm{f_t}, \bm{n_t}\rangle\right),$$
and $Y_t=\langle \bm{f}, \bm{q_t}\rangle  - \langle \bm{f_t}, \bm{n_t}\rangle$ for $t\in[T]$ give rise to a martingale difference sequence as 
$\mathbb{E}\left[Y_t\mid \cG_t, P\right] = 0$. By \Cref{bernstein2} with 
$$\lambda = \frac{1}{\sqrt{H^2T + H^3S^2A}}\leq \frac{1}{H},$$
we obtain
$$\sum_{t=1}^T \langle \bm{f}, \bm{q_t}\rangle- \sum_{t=1}^T \langle \bm{f_t}, \bm{n_t}\rangle= O\left(\left(H\sqrt{T} + H^{3/2} S\sqrt{A}\right)\left(\ln\frac{HSAT}{\delta}\right)^2 \right).$$

Next, again by  \Cref{lemma2},
we have $\mathbb{E}\left[\langle \bm{n_t}, \bm{g_t}\rangle^2\mid g_t, \pi_t, P\right]\leq  2\langle \bm{q_t},\vec h\odot \bm{f_t}\rangle$, which implies that
$$\sum_{t=1}^T\mathbb{E}\left[\langle \bm{n_t}, \bm{g_t}\rangle^2\mid  \cG_t, P\right]\leq  2\sum_{t=1}^T\mathbb{E}\left[\langle \bm{q_t},\vec h\odot \bm{g_t}\rangle\mid \cG_t,P\right]=2\sum_{t=1}^T \langle \bm{q_t},\vec h\odot \bm{g}\rangle$$
because $\pi_t$ is $\cG_t$-measurable.
Therefore, it follows that
$$\sum_{t=1}^T\mathbb{E}\left[\langle \bm{n_t}, \bm{g_t}\rangle^2\mid  \cG_t, P\right]\leq 2\underbrace{\sum_{t=1}^T \langle \bm{q_t},\vec h\odot (\bm{g}-\bm{\widehat g_t})\rangle }_{\text{Term 4}}+2\underbrace{\sum_{t=1}^T \langle \bm{q_t}- \bm{\widehat q_t},\vec h\odot \bm{\widehat g_t}\rangle}_{\text{Term 5}}+ 2\underbrace{\sum_{t=1}^T \langle \bm{\widehat q_t},\vec h\odot \bm{\widehat g_t}\rangle}_{\text{Term 6}}.$$ 
Note that by \Cref{lemma:second-term5},
\begin{align*}
	\text{Term 4}&\leq H\sum_{t=1}^T \langle \bm{q_t}, \bm{g}-\bm{\widehat g_t}\rangle\\
	&=\begin{cases}
		O\left((H^2\sqrt{T} + H^2SA)\left(\ln \frac{HSAT}{\delta}\right)^2\right),&\text{the full information setting},\\
		O\left((H^2\sqrt{SAT} + H^2SA)\left(\ln \frac{HSAT}{\delta}\right)^2\right),&\text{the bandit feedback setting}
	\end{cases}
\end{align*}
Terms 5 and 6 can be bounded similarly as terms 2 and 3, respectively. Therefore, we deduce
\begin{align*}
	\text{Term 5}&=O\left(\left(H^2T + H^3 S^2A\right)\left(\ln\frac{HSAT}{\delta}\right)^2 \right)\\
		\text{Term 6}&\leq H^2T.
\end{align*}
Then it follows that
$$\sum_{t=1}^T\mathbb{E}\left[\langle \bm{n_t}, \bm{g_t}\rangle^2\mid  \cG_t, P\right]=O\left(\left(H^2T\sqrt{SA} + H^3 S^2A\right)\left(\ln\frac{HSAT}{\delta}\right)^2 \right).$$
Moreover, $Y_t=  \langle \bm{g_t}, \bm{n_t}\rangle-\langle \bm{g}, \bm{q_t}\rangle$ for $t\in[T]$ give rise to a martingale difference sequence as 
$\mathbb{E}\left[Y_t\mid \cG_t, P\right] = 0$. By \Cref{bernstein2} with 
$$\lambda = \frac{1}{\sqrt{H^2T\sqrt{SA} + H^3S^2A}}\leq \frac{1}{H},$$
we obtain
$$\sum_{t=1}^T \langle \bm{g_t}, \bm{n_t}\rangle-\sum_{t=1}^T \langle \bm{g}, \bm{q_t}\rangle= O\left(\left(H\sqrt{SAT} + H^{3/2} S\sqrt{A}\right)\left(\ln\frac{HSAT}{\delta}\right)^2 \right),$$
as required.
\end{proof}

\subsection{Bound on the Regret Term (II)}

\begin{lemma}\label{lemma:second-term2}
Suppose that the statements of \Cref{lemma:estimator} 
hold. Then the following holds for the regret term (II).
	$$\sum_{t=1}^T \langle \bm{f}, \bm{q^*}\rangle-\sum_{t=1}^T \langle \bm{\widehat f_t}, \bm{\widehat q_t}\rangle= O\left(\left(\frac{H^{3/2}}{\rho}S\sqrt{AT} +\frac{H^{5/2}}{\rho}S^2A \right)\left(\ln \frac{HSAT}{\delta}\right)^2\right)$$
with probability at least $1-4\delta$.
\end{lemma}
\begin{proof}
For $t\in[T]$, let $G_t$ denote the amount of resource consumed in episode $t$.
We define the stopping time $\tau$ of \Cref{alg:online-alloc-mdp-unknown} 
as 
$$\min\left\{t:\ \sum_{k=1}^tG_k + H\geq  TH\rho\right\}.$$
By definition, we have $TH\rho -\sum_{t=1}^{\tau-1} G_t > H$.  Since $G_t\leq H$ for any $t\in[T]$, it follows that $TH\rho -\sum_{t=1}^{\tau} G_t > 0$, and therefore, \Cref{alg:online-alloc-mdp-unknown} 
does not terminate until the end of episode $\tau$. Then we have
$$G_t = \langle\bm{n_t},\bm{g_t}\rangle,\quad t\leq \tau.$$
Note that
\begin{align*}
\langle \bm{f}, \bm{q^*}\rangle&\leq \langle \bm{f}, \bm{q^*}\rangle + \lambda_t \left(H\rho - \langle \bm{g}, \bm{q^*}\rangle\right)\\
&\leq \langle \bm{\widehat f_t}, \bm{q^*}\rangle + \lambda_t \left(H\rho - \langle \bm{\widehat g_t}, \bm{q^*}\rangle\right)\\
&\leq \max_{\bm{q}\in \Delta(P,t)}\left\{\langle \bm{\widehat f_t}, \bm{q}\rangle + \lambda_t \left(H\rho - \langle \bm{\widehat g_t}, \bm{q}\rangle\right)\right\}\\
&=\langle \bm{\widehat f_t}, \bm{\widehat q_t}\rangle + \lambda_t \left(H\rho - \langle \bm{\widehat g_t}, \bm{\widehat q_t}\rangle\right).
\end{align*}
Then it follows that
\begin{align*}\sum_{t=1}^T\langle \bm{f}, \bm{q^*}\rangle - \sum_{t=1}^T \langle \bm{\widehat f_t}, \bm{\widehat q_t}\rangle&\leq \sum_{t=1}^T\langle \bm{f}, \bm{q^*}\rangle - \sum_{t=1}^\tau \langle \bm{\widehat f_t}, \bm{\widehat q_t}\rangle\\
	&\leq \sum_{t=\tau+1}^T \langle \bm{f},\bm{q^*}\rangle + \sum_{t=1}^\tau\langle \bm{f}, \bm{q^*}\rangle - \sum_{t=1}^\tau \langle \bm{\widehat f_t}, \bm{\widehat q_t}\rangle\\
	&\leq (T-\tau)H + \sum_{t=1}^\tau \lambda_t\left(H\rho - \langle \bm{\widehat g_t}, \bm{\widehat q_t}\rangle\right).
\end{align*}
Then the dual update rule $$\lambda_{t+1}=\max\left\{0,\lambda_t - \eta\left(H\rho - \langle \bm{\widehat g_t}, \bm{\widehat q_t}\rangle\right)\right\}$$
corresponds to the online mirror descent algorithm applied to the linear functions $w_t(\lambda)=\lambda(H\rho - \langle \bm{\widehat g_t},\bm{\widehat q_t}\rangle)$ for $t\in[\tau]$. Since  $|H\rho - \langle \bm{g_t}, \bm{\widehat q_t}\rangle|\leq 2H$, the standard analysis of online mirror descent (see \citep{Hazan}) gives us that
$$\sum_{t=1}^\tau \lambda_t(H\rho- \langle\bm{\widehat g_t},\bm{\widehat q_t}\rangle )- \sum_{t=1}^\tau \lambda(H\rho- \langle\bm{\widehat g_t},\bm{\widehat q_t}\rangle )\leq 2H^2\eta \tau + \frac{1}{2\eta} (\lambda- \lambda_1)^2\leq 2H^2\eta T + \frac{1}{2\eta} (\lambda- \lambda_1)^2.$$
Then we deduce that
\begin{align*}\sum_{t=1}^T\langle \bm{f}, \bm{q^*}\rangle - \sum_{t=1}^T \langle \bm{\widehat f_t}, \bm{\widehat q_t}\rangle&\leq  (T-\tau)H + \sum_{t=1}^\tau \lambda\left(H\rho - \langle \bm{\widehat g_t}, \bm{\widehat q_t}\rangle\right)+2H^2\eta T + \frac{1}{2\eta} (\lambda- \lambda_1)^2.
\end{align*}
If $\tau = T$, then we set $\lambda=0$, in which case
$$\sum_{t=1}^T\langle \bm{f}, \bm{q^*}\rangle - \sum_{t=1}^T \langle \bm{\widehat f_t}, \bm{\widehat q_t}\rangle\leq 2H^2\eta T + \frac{1}{2\eta} (\lambda- \lambda_1)^2.$$
If $\tau<T$, then we have
$$\sum_{t=1}^\tau G_t + H =\sum_{t=1}^\tau \langle\bm{g_t},\bm{n_t}\rangle + H \geq TH\rho.$$
In this case, we set $\lambda  =  1/\rho$. Then
\begin{align*}
	\sum_{t=1}^\tau \lambda(H\rho- \langle\bm{\widehat g_t},\bm{\widehat q_t}\rangle)
	&= \underbrace{\tau H - \frac{1}{\rho}\sum_{t=1}^\tau \langle\bm{ g_t},\bm{n_t}\rangle}_{\text{Term 1}}+ \frac{1}{\rho}\underbrace{\sum_{t=1}^\tau\left( \langle\bm{ g_t},\bm{n_t}\rangle-\langle\bm{g},\bm{q_t}\rangle\right)}_{\text{Term 2}}+\frac{1}{\rho}\underbrace{\sum_{t=1}^\tau\langle\bm{g}-\bm{\widehat g_t},\bm{q_t}\rangle}_{\text{Term 3}}+\frac{1}{\rho}\underbrace{\sum_{t=1}^\tau \langle\bm{\widehat g_t},\bm{q_t}-\bm{\widehat q_t}\rangle}_{\text{Term 4}}.
\end{align*}
Note that
$$\text{Term 1} \leq\tau H - \frac{1}{\rho}(TH\rho -H)\leq-(T-\tau)H+ \frac{H}{\rho}.$$
By Lemma 5.4, %
$$\text{Term 2} =O\left((H\sqrt{SAT} + HSA)\left(\ln \frac{HSAT}{\delta}\right)^2\right).$$
By Lemma 5.3, %
$$\text{Term 3} = O\left(\left(H\sqrt{SAT} + H^{3/2} S\sqrt{A}\right)\left(\ln\frac{HSAT}{\delta}\right)^2 \right).$$
By \Cref{lemma:regret-term2}, with probability at least $1-4\delta$, 
$$\text{Term 4} =O\left(\left(H^{3/2}S\sqrt{AT} +H^{5/2}S^2A\right)\left(\ln\frac{HSAT}{\delta}\right)^2 \right).$$
Therefore, it follows that
\begin{align*}\sum_{t=1}^T\langle \bm{f}, \bm{q^*}\rangle - \sum_{t=1}^T \langle \bm{\widehat f_t}, \bm{\widehat q_t}\rangle&=O\left(\left(H^{3/2}S\sqrt{AT} +H^{5/2}S^2A\right)\left(\ln\frac{HSAT}{\delta}\right)^2 \right),
\end{align*}
as required.
\end{proof}
\subsection{Proof of \Cref{theorem:regret2}}

Recall that
\begin{align*}
	\regret\left(\vec\gamma,\vec\pi\right)
	&=\underbrace{\opt(\vec\gamma) - \sum_{t=1}^T \langle \bm{f}, \bm{q^*}\rangle}_{\text{(I)}}+\underbrace{\sum_{t=1}^T \langle \bm{f}, \bm{q^*}\rangle-\sum_{t=1}^T \langle \bm{\widehat f_t}, \bm{\widehat q_t}\rangle}_{\text{(II)}} \\
	&\quad + \underbrace{\sum_{t=1}^T \langle \bm{\widehat f_t}, \bm{\widehat q_t}-\bm{q_t}\rangle}_{\text{(III)}}+\underbrace{\sum_{t=1}^T \langle \bm{\widehat f_t}-\bm{f}, \bm{q_t}\rangle}_{\text{(IV)}} + \underbrace{\sum_{t=1}^T \langle \bm{f}, \bm{q_t}\rangle- \reward\left(\vec\gamma,\vec\pi\right)}_{\text{(V)}}.
\end{align*}

With probability at least $1-2\delta$, 
the statements of Lemma 3.4 %
hold. Then it follows from Lemmas 5.4 and 5.5 %
\begin{align*}
\text{Term (IV)}&=O\left((H\sqrt{SAT} + HSA)\left(\ln \frac{HSAT}{\delta}\right)^2\right),\\
\text{Term (V)}&=O\left(\left(H\sqrt{SAT} + H^{3/2} S\sqrt{A}\right)\left(\ln\frac{HSAT}{\delta}\right)^2 \right).
\end{align*}
Moreover, by \Cref{lemma:second-term2}, 
\begin{align*}
\text{Term (II)}&=O\left(\left(\frac{H^{3/2}}{\rho}S\sqrt{AT} +\frac{H^{5/2}}{\rho}S^2A \right)\left(\ln \frac{HSAT}{\delta}\right)^2\right).
\end{align*}
It follows from \Cref{lemma:second-term1} that with probability at least $1-7\delta$, we have
\begin{align*}
\text{Term (I)}&=O\left(\left(H\sqrt{T} + H^2S\sqrt{A}\right)\left(\ln\frac{HSAT}{\delta}\right)^2\right).
\end{align*}
Lastly, Lemma 5.2 %
with probability at least $1-4\delta$, we have
\begin{align*}
\text{Term (III)}&=O\left(\left(H^{3/2}S\sqrt{AT} +H^{5/2}S^2A\right)\left(\ln\frac{HSAT}{\delta}\right)^2 \right).
\end{align*}
Hence, by taking the union bound, 
$$\regret(\vec\gamma,\vec\pi)=O\left(\left(\frac{H^{3/2}}{\rho}S\sqrt{AT} +\frac{H^{5/2}}{\rho}S^2A \right)\left(\ln \frac{HSAT}{\delta}\right)^2\right)$$
with probability at least $1-13\delta$.

\section{CONCENTRATION INEQUALITIES}

\begin{lemma}{\rm \citep[Theorem 4]{Maurer-bernstein}}\label{bernstein} Let $Z_1,\ldots, Z_n\in[0,1]$ be i.i.d. random variables with mean $z$, and let $\delta>0$. Then with probability at least $1-\delta$,  
$$z- \frac{1}{n}\sum_{j=1}^n Z_j \leq \lambda \sqrt{\frac{2 V_n\ln(2/\delta)}{n}} + \frac{7\ln (2/\delta)}{3(n-1)}$$
where $V_n$ is the sample variance given by
$$ V_n=\frac{1}{n(n-1)}\sum_{1\leq j<k\leq n} (Z_j-Z_k)^2.$$
\end{lemma}

Next, we need the following Bernstein-type concentration inequality for martingales due to \citep{beygelzimer11a}. We take the version used in~\citep[Lemma 9]{Jin2020}.
\begin{lemma}{\rm \citep[Theorem 1]{beygelzimer11a}}\label{bernstein2}
Let $Y_1,\ldots, Y_T$ be a martingale difference sequence with respect to a filtration $\cF_1,\ldots, \cF_T$. Assume that $Y_t\leq R$ almost surely for all $t\in [T]$. Then for any $\delta\in(0,1)$ and $\lambda\in(0,1/R]$, with probability at least $1-\delta$, we have
$$\sum_{t=1}^T Y_t\leq \lambda \sum_{t=1}^T\mathbb{E}\left[Y_t^2\mid \cF_t\right]+ \frac{\ln(1/\delta)}{\lambda}.$$
\end{lemma}

\begin{lemma}[Azuma's inequality]\label{azuma}
Let $Y_1,\ldots, Y_T$ be a martingale difference sequence with respect to a filtration $\cF_1,\ldots, \cF_T$. Assume that $|Y_t|\leq B$ for $t\in[T]$. Then with probability at least $1-\delta$, we have
$$\left|\sum_{t=1}^T Y_t\right|\leq B\sqrt{2T\ln(2/\delta)}.$$
\end{lemma}

Next, we need the following concentration inequalities due to \citep{cohen2020}.

\begin{lemma}{\rm \citep[Lemma D.3]{cohen2020}}\label{cohen-concentration0}
	Let $\{X_n\}_{n=1}^\infty$ be a sequence of i.i.d. random variables with expectation $\mu$. Suppose that $0\leq X_n\leq B$ holds almost surely for all $n$. Then with probability at least $1-\delta$, the following holds for all $n\geq 1$ simultaneously:
	\begin{align*}
		\left|\sum_{i=1}^n(X_i-\mu)\right|&\leq 2\sqrt{B\mu n\ln \frac{2n}{\delta}} + B\ln\frac{2n}{\delta},\\
	\left|\sum_{i=1}^n(X_i-\mu)\right|&\leq 2\sqrt{B\sum_{i=1}^n X_i \ln \frac{2n}{\delta}} + 7B\ln\frac{2n}{\delta}.
\end{align*}
\end{lemma}

\begin{lemma}{\rm \citep[Lemma D.4]{cohen2020}}\label{cohen-concentration}
Let $\{X_n\}_{n=1}^\infty$ be a sequence of random variables adapted to the filtration $\{\cF_n\}_{n=1}^\infty$. Suppose that $0\leq X_n\leq B$ holds almost surely for all $n$. Then with probability at least $1-\delta$, the following holds for all $n\geq 1$ simultaneously:
$$\sum_{i=1}^n \mathbb{E}\left[ X_i\mid \cF_i\right]\leq 2 \sum_{i=1}^n X_i + 4B\ln\left(2n/\delta\right).$$
\end{lemma}

\end{document}